%% file: manuscript.tex
\documentclass[12pt]{article}%
\usepackage{amsfonts}
\usepackage{amsmath}
\usepackage{amssymb}
\usepackage{verbatim}
\usepackage{natbib}
\usepackage{hyperref}
\usepackage{xcolor}
\usepackage{bbold}
\usepackage{comment}
\definecolor{darkblue}{rgb}{0,0, .55}
\hypersetup{
    colorlinks=true,
    citecolor = darkblue,
    linkcolor=darkblue,
    filecolor=magenta,
    urlcolor=darkblue,
}
\usepackage{graphicx}
\usepackage{float}
\usepackage{rotating}
\usepackage{threeparttable}
\usepackage{caption}
\usepackage[pagewise]{lineno}
\usepackage{tabularx}
\usepackage[section]{placeins}%
\usepackage{titlesec}
\usepackage[normalem]{ulem}
\usepackage[framemethod=tikz]{mdframed}
\usepackage{mathtools}
\providecommand{\U}[1]{\protect\rule{.1in}{.1in}}
\graphicspath{{./graphics/}}
\renewcommand{\baselinestretch}{0.95}
\newtheorem{theorem}{Theorem}
\newtheorem{theoremL}{TheoremL}

\newtheorem{lemma}[theoremL]{Lemma}

\newtheorem{proposition}[theorem]{Proposition}
\newtheorem{assumption}{Assumption}

\newenvironment{proof}[1][Proof]{\noindent\textbf{#1.} }{\ \rule{0.5em}{0.5em}}
\newcommand\fnote[1]{\captionsetup{font=footnotesize}\caption*{#1}}
\allowdisplaybreaks[1]
\def\citetpos#1{\citeauthor{#1}'s (\citeyear{#1})}

\newcommand{\OA}{Appendix }

\begin{document}

\title{Testing the effectiveness of unconventional monetary policy in Japan
and the United States\thanks{We appreciate comments from three anonymous referees, Jes\'{u}s Fern\'{a}ndez-Villaverde, Jordi Gal\'{i}, Simon Gilchrist, Fumio Hayashi, Hibiki Ichiue, Junko Koeda, Davide Porcellacchia, Mototsugu Shintani, and Nao Sudo, seminar and conference participants at CIGS End of the Year Conference 2019, ECB-BoJ-BoE Joint Research Workshop, Econometric Society World Congress 2020, SED 2021, Bank of Japan, University of Pavia, Kobe University, University of Pompeu Fabra, Osaka University, and University of Oxford. Mavroeidis acknowledges the financial support of the European Research Council via Consolidator grant
number 647152. Views expressed in the paper are those of the authors and do not necessarily reflect the official views of the Bank of Japan.}}
\author{Daisuke Ikeda\thanks{Bank of Japan, \texttt{daisuke.ikeda@boj.or.jp.}}
\and Shangshang Li\thanks{University of Liverpool and University of Oxford:
\texttt{shangshang.li@economics.ox.ac.uk}}
\and Sophocles Mavroeidis \thanks{University of Oxford and INET:
\texttt{sophocles.mavroeidis@economics.ox.ac.uk.}}
\and Francesco Zanetti \thanks{University of Oxford:
\texttt{francesco.zanetti@economics.ox.ac.uk.}}}
\date{\today}
\maketitle

\begin{abstract}
Unconventional monetary policy (UMP) may make the effective lower bound (ELB) on the short-term interest rate irrelevant. We develop a theoretical model that underpins our empirical test of this `irrelevance hypothesis,' based on the simple idea that under the hypothesis, the short rate can be excluded in any empirical model that accounts for alternative measures of monetary policy. We test the hypothesis for Japan and the United States using a structural vector autoregressive model with the ELB. We firmly reject the hypothesis but find that UMP has had strong delayed effects.

\vspace{1.5cm}



\end{abstract}

\maketitle

\setcounter{page}{0} \thispagestyle{empty}
\setlength{\baselineskip}{20pt}\newpage

\section{Introduction}  %for arxiv only, AEJM comment out this line

Adjustments in the overnight nominal interest rate have been the primary tool
for the implementation of monetary policy since the early 1980s. In recent
years, however, the short-term nominal interest rate reached an effective lower
bound (ELB) in several countries, making the standard policy tool \emph{de
facto} ineffective. Two prominent examples are Japan whose policy rate has been at the ELB
for most of the past quarter century, and the United States that
reached the ELB in the aftermath of the global financial crisis of 2007--2008.
The central banks in these countries countervailed the inapplicability of the
standard policy tool by embarking on unconventional monetary policy (UMP) that involves
the purchases of long-term government bonds and the use of forward guidance to
signal future policy action.\footnote{See \cite{Christensen_Rudebusch_2012},
 \cite{liu2019changing}, \cite{Campbell_etal2020}, and \cite{Carlson_etal2020} for the U.S., and \cite{Ugai2007}, and \cite{BoJ2016}
for Japan. \cite{Ueda2012} provides a comparison of monetary policy between
the U.S. and Japan.}

The effectiveness of UMP is a central issue for policymakers. One view is that the ELB restricts the effectiveness of monetary policy, thus representing an important constraint on what monetary policy can achieve, as argued by \cite{EggertssonWoodford2003}, \cite{Gust_AER17} and \cite{Eberly_etal2020}.
An alternative view is that UMP can affect long-term interest rates so significantly that UMP has been fully effective in circumventing the ELB
constraint, as argued by \cite{SwansonWilliams2014} and
\cite{DebortoliGaliGambetti2019}.
This latter view has been termed as the ELB `irrelevance hypothesis.'
The issue of the effectiveness of UMP has gained relevance since the ELB is likely to bind more often in the future with historically low levels of the longer-run natural rate of interest \citep{Fernando_etal2020}.

This paper studies the irrelevance hypothesis of the ELB both theoretically and empirically.
It develops a dynamic stochastic general equilibrium (DSGE) model with UMP. To the best of our knowledge, this is the first study that analytically characterizes the irrelevance hypothesis in a DSGE model.
The model provides the theoretical
underpinning to our novel empirical tests, and it shows that censored and kinked vector autoregressive models (VARs) are a suitable empirical framework to test the irrelevance
hypothesis.
Our empirical results show that the hypothesis is strongly rejected for both Japan and the U.S.
Despite the rejection, the estimated impulse responses to a monetary policy shock indicate strong delayed effects of UMP in each country.

\label{para:idea of IH}
The idea of the irrelevance hypothesis is that observable properties of macroeconomic variables such as their dynamics and volatilities remain unchanged when the economy moves in and out of an ELB regime, for otherwise the ELB is empirically relevant. %In principle,
This idea can be applied to any theoretical model with the ELB. This insight allows us to test the irrelevance hypothesis using reduced-form VAR models and therefore, our empirical results about this hypothesis hold for any theoretical model with the ELB as long as it has a VAR representation.
The purpose of our DSGE model is to formalize this idea by showing that such a theoretical model actually exists, and to guide our empirical approach to identifying the effectiveness of UMP when the irrelevance hypothesis does not hold.

In our model, UMP consists of (i) quantitative easing (QE) implemented by long-term government bond purchases that directly affect long-term government bond yields, and (ii) forward guidance (FG) under which the central bank commits to keeping short-term interest rates low in the future.
A key variable in the model is
the `shadow rate,' which we label $i_{t}^{\ast}$ and it is defined as the short-term interest rate that the central bank would set if there were no ELB. The short rate $i_{t}$ is given by $i_{t}=\max\{i_{t}^{\ast},\underline{i}\}$, where $\underline{i}$ is the ELB. The shadow rate is equal to the short rate in a non-ELB regime, but it is unobservable in an ELB regime where the policy rate is constrained at the ELB. The shadow rate can be negative and interpreted as the indicator of the desired stance of monetary policy in terms of the short rate.

In the model the central bank resorts to UMP in an ELB regime by using the shadow
rate for the guidance of its policy stance, as in the case of a non-ELB regime where the short rate is equal to the shadow rate. The model shows that UMP entails wide degrees of effectiveness, including the irrelevance of the ELB in which UMP retains the same effectiveness as the conventional policy that adjusts the short rate as if there were no ELB. We show that under the irrelevance of the ELB, the log-linearized DSGE model can be written in terms of inflation, the output-gap, the long-term interest rate or the shadow rate, and it retains the same VAR representation for both ELB and
non-ELB regimes, thus providing the theoretical foundation to our empirical tests.

Motivated by our theoretical results, we use the censored and kinked structural VAR model developed by \cite{Mavroeidis2019} to test the irrelevance hypothesis. Our DSGE model
shows that a direct appraisal of the irrelevance hypothesis is to test whether a short-term interest rate, which is subject to the ELB, can be excluded from the VARs that include alternative measures of monetary policy that are not subject to the ELB, such as a long-term interest rate and the shadow rate. The exclusion of the short rate in VARs that include the long rate is a novel empirical test developed in this paper. The exclusion of the short rate in VARs that include the shadow rate, which is proposed by \cite{Mavroeidis2019}, is our second test for robustness.

Our theoretical model implies that under the irrelevance hypothesis there is no attenuation in the response of the long rate to shocks when short rates are at the ELB, a focal point for the assessment of the irrelevance hypothesis in \cite{SwansonWilliams2014} and \cite{SG_LS_EZ_2015AEJM}. We study the attenuation effect by focusing on the different impact of the monetary policy shock on the long rate in the non-ELB and ELB regimes. Our monetary policy shock is a traditional monetary policy shock -- a shock to the short rate -- in a non-ELB regime while in an ELB regime it is a UMP shock -- a shock to the shadow rate. Motivated by our theory, we use the VAR model to characterize analytically a formal test for the attenuation effect in the response of the long rate to a shock to monetary policy at the ELB.

We conduct the tests of the irrelevance hypothesis of the ELB on postwar data for Japan and the U.S. We consider several different VAR specifications, varying the lag order and the estimation sample (to account for structural change), and
using alternative measures of monetary policy such as government bond yields with different maturities. In all cases, the tests overwhelmingly and consistently
reject the hypothesis that the ELB has been empirically irrelevant for both economies. Our conclusion is therefore fairly robust:\ the
ELB does represent a constraint on what monetary policy can achieve in those
economies.\footnote{This evidence corroborates \cite{Bernanke2020} who
 claims that \textquotedblleft  it also seems unlikely
that the new tools deployed during the Great Recession entirely compensated for the limits
imposed by the lower bound,\textquotedblright
   and is consistent with the findings in \cite{Gust_AER17} and
\cite{Del_Negro_AER17}, who attribute an important role to the ELB for the
decline in output during the financial crisis.} We also firmly reject the hypothesis of no attenuation in the response of the yield curve to monetary policy shocks during ELB regimes in both economies.

The rejection of the irrelevance hypothesis leaves open the question of the degree of effectiveness of UMP in an ELB regime compared to the conventional policy in a non-ELB regime. To address this question, we identify the dynamic effects of conventional and unconventional policies by
combining the identifying power of the ELB with additional sign restrictions on impulse
responses to a monetary policy shock \`a la \cite{Uhlig2005}. The ELB enables partial identification of impulse responses to a monetary policy shock, as shown in \cite{Mavroeidis2019}, because a change in the behaviour of the economy across ELB and non-ELB regimes is informative about the relative impact of conventional and unconventional policy. The
identified set based only on the ELB turns out to be fairly wide, so we use the insights from our DSGE model to impose the theoretically-congruous sign restrictions that were used in \cite{DebortoliGaliGambetti2019}. 
The sign restrictions
markedly sharpen the identified set of impulse responses.

We find that the effects of monetary policy on inflation and output on impact (i.e., within one quarter) declined when the economy entered an ELB regime: they dropped by more than $15$ percent in the U.S., and more than $50$ percent in Japan, relative to conventional policy. However, the cumulative effects of monetary policy exhibited the opposite pattern one and two years ahead: they appear to have been stronger during an ELB regime relative to a non-ELB regime, except for the response of output gap in the U.S., which remained weaker. Therefore, UMP seems to have had a delayed but stronger effect than conventional policy on inflation in the U.S., and on both inflation and output in Japan. Thus, we conclusively reject the hypothesis that the ELB has been empirically irrelevant in both countries, and find that responses of inflation and output to UMP have been different across time and across countries.

\paragraph{Related literature}
Our analysis is closely related to two strands of research. The first pertains to theoretical studies that investigate the
transmission mechanism of unconventional monetary policy. Among those,
regarding QE, our theoretical model is close in spirit to
\cite{AndrsLopezSalidoNelson2004}, \cite{ChenCurdiaFerrero2012},
\cite{Harrison2012}, \cite{GertlerKaradi2013},  \cite{liu2019changing}, and \cite{SudoTanaka}.
These studies introduce assets with different maturities
and limit arbitrage across assets to break the irrelevance of QE
that is shown by \cite{EggertssonWoodford2003}.\footnote{For other possible channels of QE, see \cite{Krishnamurthy_Vissing-Jorgensen_2011}. See also \cite{Sims_Wu_JME2020} for a recent discussion on the theoretical frameworks to study UMP.} Regarding FG, our model
follows \cite{ReifschneiderWilliams2000}, and it considers this mechanism in a
general equilibrium model that directly accounts for QE. Our main contribution to this
first strand of literature is to develop a simple model of UMP, which incorporates the shadow rate and provides the theoretical underpinnings to our empirical analyses.

The second strand of literature pertains to empirical studies that assess
the effectiveness of unconventional policy. In addition to \cite{SwansonWilliams2014} who estimate the time-varying sensitivity of longer-maturity yields to macroeconomic news using
high-frequency data, it includes
 \cite{DebortoliGaliGambetti2019} who use a SVAR to
investigate the (ir)relevance of the ELB constraint by comparing impulse
responses to shocks between normal times and ELB episodes. Differing from our SVAR, their SVAR does not include short-term interest rates. Another related study
by \cite{InoueRossi2018} uses an SVAR with shocks to the entire yield
curve and finds evidence that UMP has been
effective in the U.S. Our empirical methodology is closely related to
\citetpos{HayashiKoeda2019}, who propose an SVAR model for Japan that includes
short rates and takes into account the ELB, and our empirical model for Japan relies heavily on the insights from their empirical analysis. The main difference of our methodology from \citetpos{HayashiKoeda2019} is that we use a shadow rate to model UMP, which nests QE as long-term government bond purchases and FG as a policy rule as in \cite{ReifschneiderWilliams2000}, while \cite{HayashiKoeda2019} use excess reserves to model QE and an inflation exit condition to model FG. Our methodology provides a simpler framework to test the irrelevance hypothesis of the ELB and to compare the effectiveness of UMP relative to conventional policy.

Finally, our empirical analysis uses the estimation methodology in \cite{Mavroeidis2019}, who also reports evidence against the irrelevance hypothesis for the U.S. using a three-equation VAR model. We have several differences from that study: we develop a novel theoretical DSGE model of UMP that provides the underpinnings for a new test of the irrelevance hypothesis based on the exclusion of short rates; we characterize analytically and obtain a new formal test of no attenuation of the effect of monetary policy on long rates at the ELB; we use sign restrictions motivated from our theoretical model to sharpen the identification of impulse responses; we estimate the dynamic effects of UMP and the shadow rates in each country; we study Japanese data and conduct several robustness checks.

%\bigskip
The structure of the paper is as follows. Section \ref{s: dsge} develops a simple New Keynesian
DSGE model with UMP that provides theoretical underpinnings to our empirical model and the tests of the irrelevance hypothesis of the ELB. Section \ref{s: cksvar} introduces
the econometric methodology and presents the tests of the irrelevance hypothesis from the reduced-form solution of the SVAR. Section \ref{s: data} describes the data and reports our empirical results. Section \ref{sec: impact of mp} studies the effectiveness of UMP and its differences relative to conventional monetary policy. Section \ref{s: conclusion}
concludes. The \OA provides supporting material on the derivation of the
DSGE model, additional empirical results, and the estimates of the shadow rates for Japan and the U.S.

\section{A theoretical model of UMP\label{s: dsge}}

In this section, we develop a simple theoretical model of UMP and provide theoretical underpinnings to our empirical specifications and testing approaches to the irrelevance hypothesis. Section \ref{s: equations} presents the model with a focus on UMP. Section \ref{s: VAR representations} studies the linear and non-linear VAR representations of the model that underpin our empirical analysis.
Section \ref{s: simulations} simulates the model and illustrates how UMP can make the ELB irrelevant. The details of the model, equation derivations, parameterization, and model simulations are reported in \OA
\ref{A: theoretical model}.

\subsection{Central equations\label{s: equations}}

\paragraph{Overview}
The model is a New Keynesian model in which QE and FG are active under the ELB. The economy consists of
households, firms, and a central bank. The firm sector is standard as in a typical New Keynesian model. The household sector comprises two
types of households. Constrained households purchase long-term government bonds only, but unconstrained households can trade both short- and long-term government bonds subject to a trading cost.  The trading cost captures bond market segmentation, as in the preferred habitat theory originally proposed by \cite{ModiglianiSutch1966}, and it introduces imperfect substitutability between long- and short-term government bonds that generates a spread between the yields of these bonds.\footnote{The preferred habitat model is the predominant modelling framework to study UMP. See among others \cite{ChenCurdiaFerrero2012}, \cite{liu2019changing}, and \cite{Sims_Wu_JME2020}.}
The trading cost depends on the amount of long-term government bonds circulated in the market. By purchasing long-term government bonds, the central bank can affect the spread and thereby the long-term yield.

\paragraph{Conventional monetary policy} The central bank sets the short-term nominal interest rate $i_{t}$ using a standard Taylor rule subject to the ELB.
Let $\hat{y}_{t}$, $\hat{\pi}_{t}$, and $\hat{i}_{t}$ denote the deviation of output, inflation, and the short-term interest rate from the steady state in period $t$. Following conventional notation, the caret on a variable denotes the deviation of the variable from steady state. The short-term interest rate is set according to
\begin{align}
&\hat{i}_{t} = \max \{ \hat{i}_{t}^{\ast}, \hat{\underline{i}} \}, \label{i} \\
      &\hat{i}_{t}^{\ast}  = -\alpha \hat{i}_{t} + (1+\alpha) \hat{i}_{t}^{\text{Taylor}}, \label{i*} \\
  &\hat{i}_{t}^{\text{Taylor}} = \rho_{i}\left((1-\lambda^{\ast})\hat{i}_{t-1}+\lambda^{\ast}\hat{i}_{t-1}^{\ast}\right) + (1-\rho_{i})\left(r_{\pi}\hat{\pi}_{t} + r_{y}\hat{y}_{t}\right) + \epsilon_{t}^{i}, \label{Taylor}
\end{align}
where $\alpha \geq 0$, $\lambda^{\ast} \geq 0$, $\rho_i \geq 0$, $r_{\pi} \geq 0$, $r_y \geq 0$, and $\epsilon_{t}^{i}$ is a monetary policy shock. Equation (\ref{i}) encapsulates the ELB constraint, where $\hat{\underline{i}}$ is the ELB and $\hat{i}_{t}^{\ast}$ is the shadow rate.\footnote{For the interest rate, the deviation from steady state is expressed in terms of the gross interest rate. That is, $\hat{i}_{t} = (i_{t}-i)/(1+i)$, where $i$ is the short-term net interest rate in steady state. Hence, since the ELB is equal to $\underline{i}$ and $i_{t}\geq \underline{i}$, the deviation of $\underline{i}$ from the steady-state interest rate can be written as: $\hat{\underline{i}} = (\underline{i}-i)/(1+i)$.}

We use the term `shadow rate' since $\hat{i}_{t}^{\ast}$ is unobserved under the ELB and therefore censored at the ELB, while it is observed and equal to $\hat{i}_{t}$ outside the ELB constraint.
Our shadow rate represents the \emph{desired} stance of monetary policy for the short-term interest rate, as opposed to the \emph{effective} policy stance, e.g., in \cite{WuXia2016}.

A monetary policy shock in our model is a shock to the shadow rate, which is identical to a shock to the short-term interest rate in the non-ELB regime.
Equations (\ref{i*}) and (\ref{Taylor}) allow for FG to influence the system in the ELB regime, as we discuss below. In the non-ELB regime when $\hat{i}_{t}^{\ast},\hat{i}_{t-1}^{\ast}\geq \hat{\underline{i}}$, equations (\ref{i})-(\ref{Taylor}) reduce to $\hat{i}_{t}=\hat{i}_{t}^{\ast} = \hat{i}_{t}^{\text{Taylor}}$ with the Taylor-rule rate $\hat{i}_{t}^{\text{Taylor}}$ being equal to
\begin{equation}
    \hat{i}_{t}^{\text{Taylor}} = \rho_{i}\hat{i}_{t-1} + (1-\rho_{i})\left(r_{\pi}\hat{\pi}_{t} + r_{y}\hat{y}_{t}\right) + \epsilon_{t}^{i}. \label{standard Taylor}
\end{equation}
This equation is a standard Taylor rule that sets the current interest rate in response to the interest rate in period $t-1$, and current inflation and output.

\paragraph{Forward guidance\label{sec_for_guid}}
Equations (\ref{i*}) and (\ref{Taylor}) allow FG to maintain the short-term interest rate at a lower level than the rate implied by the standard Taylor rule (\ref{standard Taylor}).
The intensity of FG is governed by the two parameters $\lambda^{\ast}$ and $\alpha$. To see the isolated effect of $\lambda^{\ast}$, we first consider the case of $\alpha=0$ that implies $\hat{i}_{t}^{\ast}=\hat{i}_{t}^{\text{Taylor}}$ in equation (\ref{i*}). In the non-ELB regime equation (\ref{Taylor}) collapses to the standard Taylor rule (\ref{standard Taylor}). In the ELB regime of $\hat{i}_{t-1}^{\ast}<\hat{\underline{i}}$, if $\lambda^{\ast}=0$, the lagged term of equation (\ref{Taylor}) is  $\rho_{i}\hat{i}_{t-1}=\rho_{i}\hat{\underline{i}}$; if $\lambda^{\ast}>0$, however, the lagged term becomes $\rho_{i}\hat{\underline{i}} + \rho_{i}\lambda^{\ast}(\hat{i}_{t-1}^{\ast}-\hat{\underline{i}}) < \rho_{i}\hat{\underline{i}}$, exerting additional downward pressures on $\hat{i}_{t}^{\text{Taylor}}$ in equation (\ref{Taylor}) and on $i_{t}^{\ast}$ in equation (\ref{i*}) since $\hat{i}_{t}^{\ast}=\hat{i}_{t}^{\text{Taylor}}$. The decrease in $\hat{i}_{t}^{\ast}$ today keeps the interest rate $\hat{i}_{t}$ low from equation (\ref{i}), and moreover it keeps the future interest rate low by reducing the shadow rate tomorrow, leading to the same effect of FG in \cite{DebortoliGaliGambetti2019}.

Next consider the case of $\alpha > 0$. Equation (\ref{i*}) implies that FG is additionally strengthened since the shadow rate is kept lower than the Taylor-rule rate, $\hat{i}_{t}^{\ast} = \hat{i}_{t}^{\text{Taylor}} + \alpha(\hat{i}_{t}^{\text{Taylor}}-\hat{\underline{i}}) < \hat{i}_{t}^{\text{Taylor}}$under the ELB. Our implementation of FG is similar to \cite{ReifschneiderWilliams2000}.
For a given degree of interest rate smoothing ($\rho_{i}>0$), the parameter $\lambda^{\ast}>0$ generates downward pressure on the shadow rate in the ELB regime, and $\alpha>0$ further magnifies the downward reduction of the shadow rate. Thus, $\lambda^*$ primarily influences the persistence of the shadow rate, while $\alpha$ influences its level. Note that we also use $\lambda^*$ below to measure the effectiveness of QE, i.e., it is not a free parameter in equation (\ref{Taylor}) but will be determined later by equation \eqref{lambda*}. This helps us simplify the theoretical model, and results in a parameter that jointly influences FG and QE policy, see Lemma \ref{lemma1} below.

\paragraph{Quantitative easing}
In the ELB regime, the short-term interest rate is fixed at the ELB, and the central bank starts QE by purchasing long-term government bonds (consol bonds). The long-term bond issued in period $t$ pays $\mu^{j-1}$ dollars at time $t+j$. Let $P_{L,t}$ denote the price of the long-term bond, and let $R_{L,t+1}$ denote the return of holding it from period $t$ to $t+1$. The price and the return conditional on period-$t$ information are linked as
\begin{equation}
    P_{L,t}  =E_{t}\left(  \frac{1 + \mu P_{L,t+1}}{R_{L,t+1}}\right). \label{RL_PL}
\end{equation}
where $E_{t}$ is the expectation operator conditional on period-$t$ information.
The gross yield to maturity (or the long-term interest rate) at time $t$,
$\bar{R}_{L,t}$, can be defined as%
\[
  P_{L,t} = \frac{1}{\bar{R}_{L,t}}+\frac{\mu}{\left(  \bar{R}%
_{L,t}\right)  ^{2}}+\frac{\mu^{2}}{\left(  \bar{R}_{L,t}\right)  ^{3}%
}+...
\]
or%
\begin{equation}
P_{L,t}=\frac{1}{\bar{R}_{L,t}-\mu}. \label{RLbar}%
\end{equation}
Log-linearizing equations (\ref{RL_PL}) and (\ref{RLbar}) around steady state and combining them yields
\begin{equation}
   \hat{\bar{R}}_{L,t} = \left(1-\frac{\mu}{\bar{R}_{L}}\right)E_{t}\hat{R}_{L,t+1} + \frac{\mu}{\bar{R}_{L}}E_{t}\hat{\bar{R}}_{L,t+1}. \label{RLbarhat}
\end{equation}
where $\bar{R}_{L}>\mu$ is the long-term interest rate in steady state.

Unrestricted households that trade both long- and short-term government bonds pay a unitary cost for trading long-term bonds. An arbitrage between holding short- and long-term bonds yields
\begin{equation}
    E_{t}\hat{R}_{L,t+1} = \hat{i}_{t} + \frac{\zeta}{1+\zeta}\hat{\zeta}_{t}, \label{RLhat}
\end{equation}
where $\hat{\zeta}_{t}$ and $\zeta$ are the trading cost in deviation from steady state and in steady state, respectively. The trading cost introduces a spread between the returns of holding long- and short-term bonds. The trading cost is assumed to be increasing in the real amount of long-term bonds circulated in the market, $\hat{b}_{L,t}$, and is given by
\begin{equation}\label{trading cost}
    \hat{\zeta}_{t}=\rho_{\zeta}\hat{b}_{L,t} , %\hspace{.5cm}\rho_{\zeta}>0.
\end{equation}
where $\rho_{\zeta}>0$ represents the elasticity of the trading cost with respect to the amount of long-term bonds in the market.

The central bank conducts QE according to the rule:
\begin{equation}\label{QErule}
    \hat{b}_{L,t} = \min\left\{\gamma(\hat{i}_{t}^{\ast}-\hat{\underline{i}}), 0\right\}.% \hspace{.5cm}\gamma \geq 0,
\end{equation}
Parameter $\gamma \geq 0$ governs how aggressively the central bank purchases long-term bonds under the ELB, with $\gamma=0$ corresponding to no purchase. In the ELB regime, where $\hat{i}_{t}^{\ast}\leq\hat{\underline{i}}$ holds, the central bank purchases and absorbs long-term bonds from the market, so that $\hat{b}_{L,t}\leq 0$. The central bank conducts QE by using the shadow rate $\hat{i}_{t}^{\ast}$ as policy guidance similar to the conventional monetary policy. For instance, assume $\alpha=0$ in equation (\ref{i*}), such that $\hat{i}_{t}^{\ast}=\hat{i}_{t}^{\text{Taylor}}$. In the ELB regime, the more the shadow rate drops as a result of a decrease in inflation or output, the more the central bank purchases long-term government bonds. Thus, the central bank consistently aims at stabilizing inflation and output in both the non-ELB and ELB regimes.

Since constrained households hold long-term government bonds only while unconstrained households also hold short-term government bonds, the `effective' interest rate relevant to output and inflation is the weighted sum of the returns of holding short- and long-term government bonds, $\omega_{u}\hat{i}_{t}+(1-\omega_{u})E_{t}\hat{R}_{L,t+1}$, where $\omega_{u}\in (0,1)$ is the population share of unconstrained households.\footnote{The underlying assumption in deriving the effective interest rate is that the consumption of the two types of households is identical in steady state.} The Euler equation is then given by
\begin{equation}
    \hat{y}_{t} = E_{t}\hat{y}_{t+1} - \frac{1}{\sigma}\left(\omega_{u}\hat{i}_{t} + (1-\omega_{u})E_{t}\hat{R}_{L,t+1} - E_{t}\hat{\pi}_{t+1}\right) - \chi_{b}z_{t}^{b}, \label{Euler0}
\end{equation}
where $z_{t}^{b}$ is a demand (preference) shock, and $\sigma, \chi_{b}>0$. Note that the expected return of holding the long-term bonds $E_{t}\hat{R}_{L,t+1}$ depends on the trading cost $\hat{\zeta}_{t}$ (equation \ref{RLhat}), which depends on the real value of long-term bonds in the market $\hat{b}_{L,t}$ (equation \ref{trading cost}), which in turn is controlled by QE that uses the shadow rate $\hat{i}_{t}^{\ast}$ as policy guidance (equation \ref{QErule}). Hence, $E_{t}\hat{R}_{L,t+1}$ can be written as a function of $\hat{i}_{t}^{\ast}$. Substituting equations (\ref{RLhat})--(\ref{QErule}) into the Euler equation (\ref{Euler0}) yields
\begin{equation}
    \hat{y}_{t} = E_{t}\hat{y}_{t+1} - \frac{1}{\sigma}\left((1-\lambda^{\ast})\hat{i}_{t} + \lambda^{\ast}\hat{i}_{t}^{\ast} - E_{t}\hat{\pi}_{t+1}\right) - \chi_{b}z_{t}^{b}, \label{Euler}
\end{equation}
where $\lambda^{\ast}$ is given by
\begin{align}
    &\lambda^{\ast} \equiv (1-\omega_{u})\frac{\zeta}{1+\zeta}\rho_{\zeta}\gamma \label{lambda*}
\end{align}

The parameter $\lambda^{\ast}$ determines the effectiveness of QE. If $\lambda^{\ast}=0$, for instance as a result of $\rho_{\zeta}=0$ or $\gamma=0$, equation (\ref{Euler}) reduces to the standard Euler equation that omits the shadow rate. If $\lambda^{\ast}=1$, QE is `fully effective' and the interest rate $\hat{i}_{t}$ becomes irrelevant to the Euler equation and the dynamics of the system, and consequently the ELB will be irrelevant for the dynamics of the economy. From equation (\ref{lambda*}), the effectiveness of QE increases in the share of restricted households, $1-\omega_{u}$, the trading cost in steady state, $\zeta$, the elasticity of the trading cost with respect to the amount of long-term government bonds circulated in the market, $\rho_{\zeta}$, and strength in the purchasing of long-term bonds by the central bank in response to a change in the shadow rate, $\gamma$.

Comparing equations (\ref{Euler0}) and (\ref{Euler}) shows that the effective interest rate relevant to output and inflation is equal to the weighted interest rate, $(1-\lambda^{\ast})\hat{i}_{t}+\lambda^{\ast}\hat{i}_{t}^{\ast}$. Thus the weighted interest rate can also be interpreted as the effective interest rate, and this appears as a lagged variable in the Taylor rule (\ref{Taylor}) that implements FG. The parameter $\lambda^{\ast}$ depends on structural parameters pertaining to QE, as shown in equation (\ref{lambda*}). In this sense, the parameter $\lambda^{\ast}$ encapsulates the effectiveness of UMP that reflects both QE and FG.

\paragraph{Long-term interest rates}
Combining equations (\ref{RLbarhat})--(\ref{QErule}) yields the long-term interest rate, given by
\begin{equation}
    \hat{\bar{R}}_{L,t} = \begin{dcases}
    \left(1-\frac{\mu}{\bar{R}_{L}}\right)\hat{i}_{t} + \frac{\mu}{\bar{R}_{L}}E_{t}\hat{\bar{R}}_{L,t+1}, & \hat{i}_{t}^{\ast} > \hat{\underline{i}} \\
    \left(1-\frac{\mu}{\bar{R}_{L}}\right)\left(\frac{(1-\lambda^{\ast})\hat{i}_{t}+\lambda^{\ast}\hat{i}^{\ast}_{t} - \omega_{u}\hat{i}_{t}}{1-\omega_{u}}\right) + \frac{\mu}{\bar{R}_{L}}E_{t}\hat{\bar{R}}_{L,t+1}, & \hat{i}^{\ast}_{t} \leq \hat{\underline{i}}
    \end{dcases} \label{RLbarhat2}
\end{equation}
where $\hat{i}_{t}$ is set by equation (\ref{i}).
Equation (\ref{RLbarhat2}) shows that the effects of the shadow rate and the degree of the effectiveness of UMP on the long-term interest rate are different between the two regimes. In the non-ELB regime, $\hat{i}_{t}=\hat{i}_{t}^{\ast}$ and the expectation hypothesis holds: the long rate is given by the weighted sum of the expected short rates today and in the future. In the ELB regime, the short rate is bounded at $\hat{i}_{t}=\hat{\underline{i}}$ and the long rate is affected by QE as the shadow rate appears in the lower equation in (\ref{RLbarhat2}).

In the limiting case where the number of unconstrained households becomes infinitesimally small asymptotically, $\omega_{u}\rightarrow 0$, the difference between the two regimes in equation (\ref{RLbarhat2}) vanishes and the long-term interest rate can be written in the form
\begin{equation}
    \hat{\bar{R}}_{L,t} =
    \left(1-\frac{\mu}{\bar{R}_{L}}\right)\left(\lambda^{\ast}\hat{i}^{\ast}_{t} + (1-\lambda^{\ast})\hat{i}_{t}\right) + \frac{\mu}{\bar{R}_{L}}E_{t}\hat{\bar{R}}_{L,t+1}. \label{RLbarhat3}
\end{equation}
In this case, the expectation hypothesis holds with respect to the effective interest rate, $\lambda^{\ast}\hat{i}_{t}^{\ast}+(1-\lambda^{\ast})\hat{i}_{t}$. In addition, equation (\ref{RLbarhat3}) makes clear the presence of the attenuation bias of the long rate in the ELB regime and that it is related to the UMP effectiveness parameter $\lambda^{\ast}$. In the case of full effectiveness of UMP, i.e., $\lambda^{\ast}=1$, the long rate is always given by the expected sum of the shadow rates and there will be no attenuation of the long rate in the ELB regime. However, in the case of less effective UMP, i.e., $\lambda^{\ast}<1$, the long rate responds less to the shadow rate, giving rise to its attenuated response in the ELB regime.

\paragraph{System of equations}
We close the model by reporting the standard New Keynesian Phillips curve that relates the current inflation to expected inflation and output, given by
\begin{align}
  &\hat{\pi}_{t} = \delta E_{t}\hat{\pi}_{t+1} + \kappa \hat{y}_{t} - \chi_{a}z_{t}^{a}, \label{Phillips}
\end{align}
where $\delta \geq 0$, $\kappa \geq 0$, and $z_{t}^{a}$ is a supply (productivity) shock. The supply and demand shocks follow an AR(1) process.
To summarize, our model of UMP comprises the six equations (\ref{i}), (\ref{i*}), (\ref{Taylor}), (\ref{Euler}), (\ref{RLbarhat2}), and (\ref{Phillips}), the six endogenous variables $\{\hat{y}_{t},\hat{\pi}_{t}, \hat{i}_{t}, \hat{i}_{t}^{\ast}, \hat{i}_{t}^{\text{Taylor}}, \hat{\bar{R}}_{L,t}\}$, and the three shocks $\{z_{t}^{a},z_{t}^{b},\epsilon_{t}^{i}\}$.

\subsection{DSGE model and VAR representations}
\label{s: VAR representations}
We use our DSGE model developed in Section \ref{s: equations} to derive the VAR representations that underpin our empirical models and tests to be introduced in Section \ref{s: cksvar}.

\paragraph{The joint effect of QE and FG}
Before deriving VAR representations, it is useful to note that we cannot separate out the effects of the two parameters $\lambda^{\ast}$ and $\alpha$ in the system, but the model allows us to identify their joint effect on the system, encapsulated by $\xi^{\ast}\equiv\lambda^{\ast}(1+\alpha)$. The next lemma states the result formally.
\begin{lemma}\label{lemma1}
For any $\lambda^{\ast}\neq \lambda^{\ast\prime}$ and $\alpha\neq \alpha^{\prime}$ that satisfy $\xi^{\ast}\equiv \lambda^{\ast}(1+\alpha)=\lambda^{\ast\prime}(1+\alpha^{\prime})$, the model with $\lambda^{\ast}$ and $\alpha$ is observationally equivalent to the model with $\lambda^{\ast\prime}$ and $\alpha^{\prime}$.
\end{lemma}
The proof is straightforward. By using equation (\ref{i*}) to substitute out $\hat{i}_{t}^{\ast}$, the effective interest rate, $(1-\lambda^{\ast})\hat{i}_{t}+\lambda^{\ast}\hat{i}_{t}^{\ast}$, can be replaced with $(1-\xi^{\ast})\hat{i}_{t}+\xi^{\ast}\hat{i}_{t}^{\text{Taylor}}$. The parameter  $\lambda^{\ast}$ does not appear anywhere in the model except in $\xi^{\ast}=\lambda^{\ast}(1+\alpha)$. The parameter $\alpha$ appears only in equation (\ref{i}) as $\hat{i}_{t}=\max\{\hat{i}_{t}^{\text{Taylor}}-\alpha(\hat{i}_{t}-\hat{i}_{t}^{\text{Taylor}}), \hat{\underline{i}}\}$. But this equation is observationally equivalent to $\hat{i}_{t}=\max\{\hat{i}_{t}^{\text{Taylor}},\hat{\underline{i}}\}$ because only $\hat{i}_{t}$ is observable, and thus the joint effect of QE and FG is summarized by the parameter $\xi^{\ast}$ that encompasses $\lambda^{\ast}$ and $\alpha$.\footnote{\label{footnote: lambda*}QE and FG are not separately identifiable here because they both operate only in the ELB regime during which the shadow rate is unobserved and only identified up to scale, see \cite{Mavroeidis2019} for further discussion. It might be possible to disentangle the two policies if, for example, one of them operated also outside the ELB, as in \cite{Swanson2021}. This would require some additional assumptions to carefully model more than two endogenously-switching regimes.}

As an illustrative example, consider a case in which QE is half as effective as
the conventional monetary policy (i.e., $\lambda^{\ast}=0.5$) but FG is active with $\alpha=1$. In this scenario, the monetary policy shock $\epsilon_{t}^{i}$ is twice as large in the ELB regime as the same shock in the non-ELB regime (or, equivalently, at the ELB regime with $\lambda^{\ast}=1$ and $\alpha=0$).
The impact of such a monetary policy shock at the ELB regime is of the same magnitude as the
equivalent shock in the non-ELB regime.
Thus, the effectiveness of UMP in the theoretical model depends on $\xi^{\ast}$ that encapsulates the joint effect of $\lambda^{\ast}$ and $\alpha$, and encompasses the combination of QE and FG.

\paragraph{Irrelevance of the ELB and VAR representations}
Now we establish the VAR representations of the DSGE model under the irrelevance hypothesis of the ELB, where UMP is as effective as the conventional policy: $\xi^{\ast}=1$. There will be different VAR representations, depending on what assumptions we impose on the model. Specifically, we begin by considering the following assumption. \begin{assumption}\label{assumption1}
The number of unrestricted households becomes infinitesimally small asymptotically: $\omega_{u}\rightarrow 0$.
\end{assumption}
Under the irrelevance of the ELB, the solution to the model entails two VAR(1) representations, as formalized in the following proposition.
\begin{proposition}\label{prop1}
 Consider the DSGE model in equations (\ref{i}), (\ref{i*}), (\ref{Taylor}), (\ref{Euler}), (\ref{RLbarhat2}), and (\ref{Phillips}) under the irrelevance of the ELB: $\xi^{\ast}=1$. Then,
 \begin{enumerate}
     %\item[i)] Under Assumption \ref{assumption1},  $[\hat{y}_{t}, \hat{\pi}_{t}, \hat{i}_{t}^{\ast}, \hat{\bar{R}}_{L,t}]$ has a VAR(1) representation.
     \item[i)] $[\hat{y}_{t}, \hat{\pi}_{t}, \hat{i}_{t}^{\ast}]$ has a VAR(1) representation.
     \item[ii)] Under Assumption \ref{assumption1}, $[\hat{y}_{t}, \hat{\pi}_{t}, \hat{\bar{R}}_{L,t}]$ has a VAR(1) representation.
 \end{enumerate}
\end{proposition}
\begin{proof}
\OA \ref{a: proof 1}.
\end{proof}\vspace{.5cm}

Proposition \ref{prop1}(i) shows that under the irrelevance of the ELB, the short-term interest rate $i_{t}$ is redundant for the dynamics of output and inflation once the shadow rate is included in the VAR.
Whether the short rate is constrained by the ELB or not does not influence the dynamics of output and inflation, so that the ELB is irrelevant. What distinguishes our VAR from a standard linear VAR  is that the shadow rate is censored at the ELB. In other words, for an econometrician, the shadow rate is observable, and equal to the nominal interest rate, only when it is above the ELB (equation \ref{i}).

Proposition \ref{prop1}(ii) shows that the long rate can be a sufficient indicator of monetary policy. Under Assumption \ref{assumption1}, the shadow rate has no direct effect on the economy while having an indirect effect by affecting the long rate, and under the irrelevance of the ELB, the shadow rate and the long rate become interchangeable as implied by equation (\ref{RLbarhat3}) with $\lambda^{\ast}=1$ and $\alpha=0$ without loss of generality from Lemma \ref{lemma1}. 
Intuitively, in such a special case, the expectation hypothesis holds and the long rate is given by the expected sum of the shadow rates today and in the future, giving rise to a one-to-one relation between the long rate and the shadow rate.
The corollary of Proposition \ref{prop1}(ii) is that under the irrelevance hypothesis there is no attenuation of the response of the long rate, i.e., the dynamics of the long rate are identical between the non-ELB and ELB regimes.
We will use Proposition \ref{prop1} to underpin our approach to testing the irrelevance hypothesis of the ELB and no attenuation  effect in Section \ref{s: cksvar}.

Assumption \ref{assumption1} plays a critical role in deriving the VAR representation with the long rate under the irrelevance of the ELB, i.e., the specific irrelevance hypothesis used in the literature \citep[e.g.,][]{DebortoliGaliGambetti2019}. Without the assumption, equation (\ref{RLbarhat2}) shows that even under the irrelevance of the ELB ($\xi^{\ast}=1$), the long rate is kinked at the ELB: the first-term in equation (\ref{RLbarhat2}) is smaller by $\omega_{u}/(1-\omega_{u})\times\hat{\underline{i}}$ in the ELB regime than in the non-ELB regime. Why does the long rate have to decrease more in the ELB regime to make the ELB irrelevant?
The reason is straightforward. Unrestricted households hold short-term government bonds whose relevant interest rate is the short rate. To stimulate the aggregate consumption as if the economy were not at the ELB, consumption by the restricted households must be stimulated more, which requires a stronger decrease in the long rate under the ELB. 
Put differently, the decrease in the long rate needed to satisfy the irrelevance of monetary policy at the ELB regime becomes less pronounced with the reduction of unrestricted agents. Thus, Assumption \ref{assumption1} is important to validate the VAR representation with the long rate.

\paragraph{Relevance of the ELB and VAR representations}
Now consider the case of a less effective UMP than conventional policy (i.e., $\xi^{\ast}\leq1$). In this general case, the model does not have a tractable solution under rational expectations and thus it would not have a VAR representation. However, the model admits a tractable solution if the formation of expectations slightly deviates from rational expectations by the following assumption.
\begin{assumption}\label{assumption3}
In each period $t$, agents know the true $\xi^{\ast}\leq 1$ today and form expectations under the presumption of $\xi^{\ast}=1$ from period $t+1$ onward.
\end{assumption}
This assumption implies that agents entail behavioral expectations and believe that UMP will be as effective as the conventional policy from the next period $t+1$ onwards. Under this assumption, the VAR model has a piecewise linear representation, as stated in the following proposition.

\begin{proposition}\label{prop2}
Under Assumption \ref{assumption3},  the DSGE model in equations (\ref{i}), (\ref{i*}), (\ref{Taylor}), (\ref{Euler}), and (\ref{Phillips}) has a piecewise linear VAR representation with a kink at the ELB.
\end{proposition}
\begin{proof}
\OA \ref{a: proof2}
\end{proof}\vspace{.5cm}

Proposition \ref{prop2} implies that in the case of $\xi^{\ast}<1$ the dynamics of the economy differ between the non-ELB and ELB regimes as the dynamics in each regime is represented by a distinct VAR where a change in regimes occurs when the short-term interest rate crosses the ELB. This result echos with \cite{Aruoba-etal2021}, who argue that a piecewise linear solution to a DSGE model with an occasionally binding ELB constraint can be interpreted as describing the behavior of boundedly rational agents. Our specific assumption on the expectation formation in Proposition \ref{prop2} represents a form of bounded rationality. 

In our VAR representation in Proposition \ref{prop2}, equation (\ref{i}) continues to hold, and the shadow rate is unobserved and censored in the ELB regime. By allowing for a kink in the dynamics, Proposition \ref{prop2} shows that a censored and kinked VAR model provides a suitable empirical specification for testing the irrelevance hypothesis.

\subsection{Simulations}
\label{s: simulations}

Before developing our empirical framework in the next section, we illustrate the effects of UMP and study impulse responses to a monetary policy shock in the ELB regime by solving the model under Proposition \ref{prop2}.\footnote{\label{footnote_simul_shocks}While our analysis focuses on monetary policy, \OA \ref{app:pref_supp_shocks} reports the responses of output and inflation to demand and supply shocks under the ELB to illustrate the role of the effectiveness of UMP for these shocks.}
The analysis aims at illustrating the effects of UMP and it is not designed to draw quantitative implications.
\OA \ref{A: parameterization} reports the parameterization of the model.

\begin{figure}[t]
\caption{The effects of UMP}%
\label{Fig:UMP}%
\centering
\includegraphics[width=16cm]{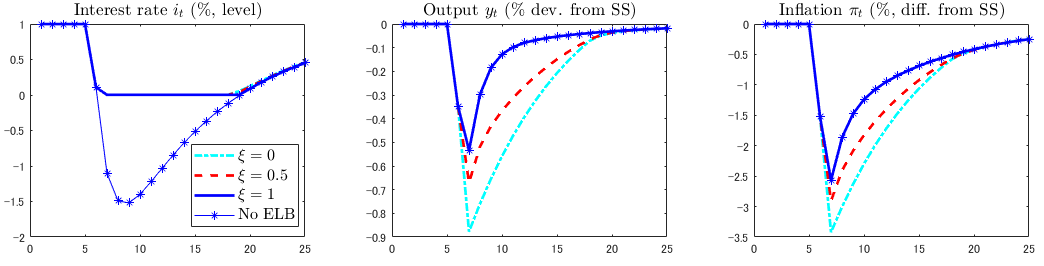}
\vspace{.1cm}
\fnote{\footnotesize {Note: The figure shows the dynamic path of the model under Proposition \ref{prop2} where a severe demand shock hits the economy in periods $t=6$ and $t=7$. The dynamic path is computed by transforming the theoretical model under Proposition \ref{prop2} into the reduced form equations (\ref{eq: RF Y1}) and (\ref{eq: RF Y2}) and calculating the response to the shocks.`No ELB' represents the model without the ELB, where the interest rate equation (\ref{i}) is replaced by $i_{t}=i_{t}^{\ast}$. `SS' denotes a steady state, `dev.' denotes a deviation, and `diff.' denotes a difference.} }
\end{figure}

\paragraph{No UMP}
The dash-dotted line in Figure \ref{Fig:UMP} shows simulated paths for the theoretical model under Proposition \ref{prop2} in the case of no UMP ($\xi^{\ast}=0$). The economy starts from the steady state and large negative demand shocks hit in periods $t=6$ and $t=7$. The consecutive negative demand shocks bring the economy to the ELB
and generate a severe recession by decreasing output and inflation sharply. At the ELB (dash-dotted line), the interest rate $i_{t}$ cannot be lowered in response to the fall in inflation. This raises the real interest rate, decreases consumption and output, and puts further downward pressure on inflation through the Phillips curve (\ref{Phillips}). This negative feedback loop magnifies the falls in output and inflation compared to the hypothetical economy without the ELB (the star-marked line).

\paragraph{UMP\label{s:theory_QE}}
UMP can offset the negative impact of the ELB. When UMP is partially effective ($\xi^\ast=0.5$; the dashed line), the magnitude of the falls in output and inflation are mitigated relative to the case without UMP ($\xi^\ast=0$; the dash-dotted line). When UMP is fully effective ($\xi^\ast=1$; the solid line), although the interest rate $i_{t}$ remains at the ELB, output and inflation follow the same paths as in the case of no ELB (the star-marked line), as shown in Figure \ref{Fig:UMP}. In response to a decrease in the shadow rate $i_{t}^{\ast}$, the central bank increases the purchase of long-term government bonds and, by doing so, it lowers the long-term government bond yield by compressing its premium, which boosts consumption and output. When $\xi^\ast=1$, UMP perfectly offsets the contractionary effect of the ELB. The interest rate $i_{t}$ becomes irrelevant to the dynamics of the economy, which evolves as if there were no ELB.

\paragraph{Impulse responses to a monetary policy shock\label{s:theory_mon_pol_shock}} Figure \ref{Fig:MPshock} plots impulse responses to a 0.25 percentage points cut in the shadow rate under the ELB starting from period $t=1$ for the theoretical model solved under Proposition \ref{prop2} (the solid line) and the model solved by the OccBin algorithm (the dashed line), developed by \cite{GuerrieriIacoviello2015}, which has been a popular approach to solving DSGE models at the ELB (see \citealp{Atkinson-etal2019}).\footnote{The impulse responses are computed by using the same method employed in reporting our empirical results. For the details, see Section \ref{sec: IRFs}.}  The OccBin solution assumes that agents in the model form expectations by treating the non-ELB regime as an absorbing state, i.e., by assuming that interest rates will remain positive once the economy exits the ELB regime. Thus, OccBin uses an alternative behavioural assumption on the formation of expectations than the one used in Proposition \ref{prop2}.

\begin{figure}[t]
\caption{Impulse responses to a monetary policy shock at the ELB}%
\label{Fig:MPshock}
\centering
\includegraphics[width=16.5cm]{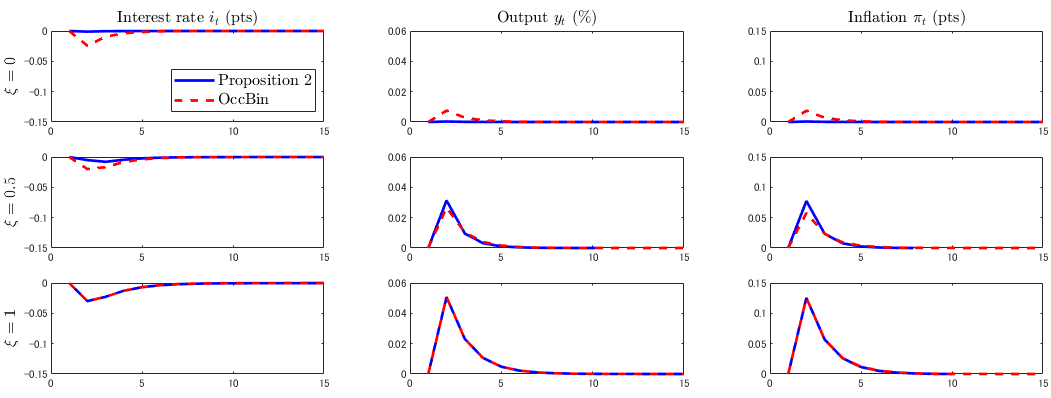}
\vspace{-.5cm}
\fnote{\footnotesize {Note: `Proposition 2' denotes the theoretical model under Proposition \ref{prop2} and `OccBin' denotes the model solved by the algorithm developed by \cite{GuerrieriIacoviello2015}. 
For analyzing the impulse responses under the ELB, for each case of $\xi^\ast$, the initial condition is set as endogenous variables which are realized using OccBin in response to a severe negative demand shock.}}
\end{figure}

Overall the responses of the interest rate, output, and inflation are similar between the model solution under Proposition \ref{prop2} and the OccBin solution, as shown in Figure \ref{Fig:MPshock}. The responses of the interest rate are muted because the economy starts from the ELB triggered by a severe demand shock in period $t=1$. Without the ELB, the interest rate (left panels) would fall by about 0.15 percentage points (pts), reported in the figure as the lowest value on the y-axis.\footnote{The responses of the interest rate are slightly negative because they are calculated relative to the expected interest rate conditional on no monetary policy shock in the initial period, which is slightly positive due to the realizations of shocks that bring the economy above the ELB.} In the case of no UMP ($\xi^\ast=0$; top panels), the responses of output (central panels) and inflation (right panels) are muted for both the model solution under Proposition \ref{prop2} and the OccBin solution.  Because the economy is at the ELB, the monetary policy shock in period $t=2$ does not have significant effects on the economy without UMP. In the case of partial UMP ($\xi^\ast=0.5$; middle panels), QE is activated in response to a decrease in the shadow rate triggered by the monetary policy shock, and output and inflation increase. In the case of fully effective UMP ($\xi^\ast=1$; bottom panels), the `irrelevance hypothesis' holds and the responses of output and inflation coincide with those under the hypothetical economy with no ELB under both solution methods.

\section{Empirical model\label{s: cksvar}}

Our theoretical model in Section \ref{s: dsge} highlights three important features to study the irrelevance hypothesis:
i) the censoring of the shadow rate at the ELB, ii) a potential kink of the dynamics of the economy at the ELB, and
iii) a parameter that encapsulates the effectiveness of UMP.
We embed these features in our empirical model designed to test the irrelevance hypothesis using the flexible VAR approach of the literature on monetary policy. In this section, we present our empirical model, derive the tests of the irrelevance of the ELB and no attenuation hypotheses, and describe our approach to the identification of the effects of conventional and unconventional monetary policy.

\subsection{Censored and kinked SVAR}
\label{s: sub_cksvar}

The econometric model that we use is the censored and kinked SVAR (CKSVAR)
developed by \cite{Mavroeidis2019}. Its structural form is given by:
\begin{subequations}\label{eq: CKSVAR}
\begin{align}%\label{eq: CKSVAR}
i_{t} &  =\max\left\{  i_{t}^{\ast},\underline{i}_{t}\right\}  ,\label{eq: Y2} \\
i_{t}^{\ast} &  =-\alpha i_{t}+\left(  1+\alpha\right)  \left(  \gamma
Y_{1t}+B_{2}X_{t}+B_{22}^{\ast}X_{2t}^{\ast}+A_{22}^{*-1}\varepsilon_{2t}\right)
,\label{eq: Y2*}\\
Y_{1t} &  =\beta\left(  \lambda i_{t}^{\ast}+\left(  1-\lambda\right)
i_{t}\right)  +B_{1}X_{t}+B_{12}^{\ast}X_{2t}^{\ast}+A_{11}^{-1}\varepsilon
_{1t},\label{eq: Y1}
\end{align}
\end{subequations}
where $i_{t}$ is the short-term interest rate that is subject to the observable lower bound of $\underline{i}_{t}$, $i_{t}^{\ast}$ is the shadow rate,  $Y_{1t}$ is a vector of unconstrained endogenous variables such as inflation and output,
 $X_{t}$ comprises exogenous and predetermined
variables, including lags of $Y_{1t}$ and $i_{t}$, $X_{2t}^{\ast}$ consists of lags of
$i_{t}^{\ast},$ and $\varepsilon_{t}$ are i.i.d.~structural shocks with identity covariance matrix.

Equation (\ref{eq: Y2}) represents the ELB constraint, and it corresponds to equation (\ref{i}) in the theoretical model, except for the lower bound that is allowed to vary over time in equation (\ref{eq: Y2}). Thus, the shadow rate is censored and unobservable under the ELB, as in the theoretical model.

Equation (\ref{eq: Y2*}) represents the short-term interest rate rule that nests the FG rule of \cite{ReifschneiderWilliams2000}, and it corresponds to equations (\ref{i*}) and (\ref{Taylor}) in the theoretical model. The parameter $\alpha$ has the same interpretation as in the theoretical model.

Equation (\ref{eq: Y1}) describes the dynamics of variables of interest such as inflation and output and the relationship between these variables and interest rates. The parameter $\lambda$ is the equivalent of $\lambda^{\ast}$ in the theoretical model and characterizes the effectiveness of UMP relative
to conventional policy \textit{on impact}. Specifically, from equation
(\ref{eq: Y1}) we see that above the ELB (i.e., when $i_{t}=i_{t}^{\ast
}>\underline{i}_{t}$), the \emph{contemporaneous} effect of a change in the short-term interest
rate $i_{t}$ by one unit on $Y_{1t}$ is $\beta,$ but the corresponding effect
at the ELB, driven by a change in the shadow rate $i_{t}^{\ast},$\ is $\lambda\beta.$ When
$\lambda=1,$ the two effects are equal, while $\lambda=0$ corresponds to the
case in which UMP has no contemporaneous effect on $Y_{1t}.$

Similar to the theoretical model, the parameter $\lambda$ partially characterizes the impulse
responses to a monetary policy shock at the ELB, since the response of the interest rate also depends on the degree of FG, and thus the joint effect of $\alpha$ and $\lambda$ determines the response of endogenous variables to UMP. To see this in the context of the SVAR model, consider the impulse response to the monetary policy shock of
\label{ref1_comm}$\varepsilon_{2t} = A_{22}^{*}$ ignoring nonlinearities.\footnote{We will discuss the specification of the impulse response functions in Section \ref{sec: impact of mp}, see equation (\ref{eq: IRFs}).} The effect on $Y_{1t}$ is $\beta/\left(  1-\gamma\beta\right)  $
above the ELB, and $\xi\beta/\left(  1-\xi\gamma\beta\right)$ at the ELB, where:%
\begin{equation}
\xi=\lambda\left(  1+\alpha\right)  .\label{eq: xi}%
\end{equation}
So, it is, in fact, $\xi$, not $\lambda$, that measures the effectiveness of an
UMP shock -- a shock to the shadow rate below the ELB, see \cite{Mavroeidis2019}
for further discussion. The parameter $\xi$ is the equivalent $\xi^{\ast}$ in Lemma \ref{lemma1} of the theoretical model. In the proof of Proposition \ref{prop2} (see \OA \ref{a: proof2}), we show that the UMP parameter $\xi^{\ast}$ in the theoretical model under Proposition \ref{prop2} can be mapped in the system of equations (\ref{eq: Y2})-(\ref{eq: Y1}), and the $\xi^{\ast}$ in the theoretical model coincides with the corresponding parameter $\xi$ in the empirical model.

Our discussion about the parameter $\xi$
concerned the relative effectiveness of UMP \emph{on impact}. The \emph{dynamic}
effects of UMP on $Y_{1t}$ are governed by the  coefficients on the lags of the
shadow rate $B_{12}^{\ast}$ in equation (\ref{eq: Y1}). For example, the case of completely
ineffective UMP on $Y_{1t}$ at all horizons can be represented by the joint restrictions $\xi=0$
\emph{and} $B_{12}^{\ast}=0$. A more restrictive case is that UMP is also ineffective on the short-term interest rate $i_{t}$, in addition to having no effect on $Y_{1t}$, which  can be implemented by $\xi=0$,
$B_{12}^{\ast}=0$, \emph{and} $B_{22}^{\ast}=0.$ This implies that the shadow rate has no contemporaneous and cumulative impact on the endogenous variables and completely
drops out of the right-hand side of equations (\ref{eq: Y2})-(\ref{eq: Y1}%
). We refer to this case as the \emph{kinked }SVAR (KSVAR), using the same terminology in  \cite{Mavroeidis2019}.

\subsection{Reduced-form solution of the SVAR and identification}
\label{s: identification}
To implement our empirical tests and gain intuition on the identification and estimation of the CKSVAR model, we derive the reduced-form solution for $Y_{1t}$ and $i_{t}$ using equations
(\ref{eq: Y2})-(\ref{eq: Y1}). \cite{Mavroeidis2019} develops the methodology for the identification and estimation of the CKSVAR, showing that the model is generally
under-identified, but the parameter $\xi$, defined in equation (\ref{eq: xi}), and the impulse responses to the monetary policy shock $\varepsilon_{2t}$,
are partially identified in general. The reduced-form solution for $Y_{1t}$ and $i_{t}$ is:%
\begin{subequations}\label{eq: CKSVAR-RF}
\begin{align}
\label{eq: RF Y2}
i_{t}&= \max\left\{C_{21}X_{1t}+C_{22}X_{2t}+C_{22}^{\ast}X_{2t}^{\ast}+u_{2t},\underline{i}_{t}\right\} \\
\label{eq: RF Y1}
Y_{1t}&=C_{11}X_{1t}+C_{12}X_{2t}+C_{12}^{\ast}X_{2t}^{\ast}+u_{1t}-\widetilde{\beta}%
D_{t}\left(  C_{2}X_{t}+C_{22}^{\ast}X_{2t}^{\ast}+u_{2t}-\underline{i}_{t}\right)
\end{align}
\end{subequations}
where $D_{t}\equiv \mathbb{1}_{\left\{  i_{t}=\underline{i}_{t}\right\}  }$ is the indicator of the ELB
regime, $X_t \equiv (X_{1t}',X_{2t}')'$, $X_{2t}$ consists of the lags of $i_{t}$, the matrices $C_{11},C_{12},C_{12}^{\ast}$, $C_{2}\equiv (C_{21},C_{22})$,  and $C_{22}^{\ast}$ are
reduced-form coefficients, $u_{t}\equiv (u_{1t}^{\prime},u_{2t}^{\prime})^{\prime}$ are reduced-form
errors, and $\Omega \equiv var\left(  u_{t}\right)  .$

The reduced-form equations in (\ref{eq: CKSVAR-RF}) represent a censored and kinked VAR. Equation (\ref{eq: RF Y2}) entails the censoring of the shadow rate, represented by the $\max$ operator. Equation (\ref{eq: RF Y1}) allows for a kink at the ELB and the coefficients and variance can change across regimes. Therefore, our empirical model nests the theoretical model under Proposition \ref{prop2}, while imposing minimal structure and allowing for flexible coefficients in the VAR.

The coefficient of the kink $\widetilde{\beta}$ in equation (\ref{eq: RF Y1}) is
identified, together with the remaining reduced-form
parameters.
In other words,
we can infer from the data whether the slope coefficients and the variance of
$Y_{1t}$ change across regimes by testing whether $\widetilde{\beta}=0.$
However, the parameter $\widetilde{\beta}$ does not have a structural
interpretation and relates to the underlying structural parameters through the following
equations:
\begin{align}
\widetilde{\beta} &  =\left(  1-\xi\right)  \left(  I-\xi\beta\gamma\right)
^{-1}\beta,\label{eq: betatilde}\\
\gamma &  =\left(  \Omega_{12}^{\prime}-\Omega_{22}\beta^{\prime}\right)
\left(  \Omega_{11}-\Omega_{12}\beta^{\prime}\right)  ^{-1}.\label{eq: gamma}%
\end{align}
As shown in \cite{Mavroeidis2019}, the structural parameters $\xi,\beta$ and
$\gamma$ are partially identified, in the sense that there is a
set of different combinations for values of the structural parameters that satisfy equations (\ref{eq: betatilde}) and (\ref{eq: gamma}) and generate any given value
of the reduced-form parameters $\widetilde{\beta}$ and $\Omega$.
Therefore, the impulse responses to a monetary policy shock are set-identified. In our
empirical analysis below, we will sharpen the identified set by using sign restrictions on the impulse responses.

\subsection{Hypothesis tests of the irrelevance of the ELB}
\label{s: IH}

We now develop our tests for the irrelevance hypothesis (IH) of the ELB from the reduced-form solution of the SVAR. A central implication of our theoretical framework was that the dynamics of the
economy are independent from whether policy rates are at the ELB or not. We use this fundamental implication to formulate two testable hypotheses.

\paragraph{Irrelevance hypothesis 1 (IH$_1$)}
Our first approach to test the IH is motivated by \cite{Swanson2018} and
\cite{DebortoliGaliGambetti2019}, who argue that monetary policy remains
similarly effective across ELB and non-ELB regimes and that
long-term interest rates are a plausible indicator of the stance of monetary
policy. \cite{DebortoliGaliGambetti2019} use SVARs that include long-term, rather than
short-term interest rates as indicators of monetary policy. They use such VARs
to identify the impulse responses of the macroeconomic variables to monetary
policy as well as the response of policy to economic conditions, and find that
those responses are similar across ELB and non-ELB regimes in the U.S. The
implicit and testable assumption that underlies their analysis is that
the short-term interest rate can be excluded from the dynamics of all the other
variables in the system, and the dynamics of the system do not change when the economy enters the ELB regime. In other words, under the IH, the dynamics of the economy can be represented by VARs with the long rate but without the short rate, which we formally showed in Proposition \ref{prop1}(ii).

The hypothesis can be tested as an exclusion
restriction in an SVAR that includes both the short and long rates. Since
the short rate is subject to the ELB constraint, the relevant framework
is the CKSVAR (that allows for the shadow rate to affect the economy in the ELB regime) and the special case of KSVAR (that precludes the shadow rate to affect the economy in the ELB regime) introduced in Section
\ref{s: sub_cksvar}. Specifically, looking at the reduced-form specification in equation (\ref{eq: RF Y1}), the IH can be formulated as:
\begin{equation}
\text{IH}_1: C_{12} = C^*_{12}=0 \text{ and } \widetilde{\beta}=0. \label{eq: IH}
\end{equation}
In words, $C_{12} = C^*_{12}=0$ means that lags of the short rate ($i_{t}$) and the shadow rate ($i^*_{t}$) can be excluded from the equation (\ref{eq: RF Y1}) that governs the unconstrained variables ($Y_{1t}$) in the VAR, and $\widetilde{\beta}=0$ means that the slope coefficients and the variance of the errors of those equations (for $Y_{1t}$) remain the same when the economy moves across regimes.

\paragraph{Irrelevance hypothesis 2 (IH$_{2}$)}
The second test of the IH is motivated by Proposition \ref{prop1}(i), which shows that when UMP is fully effective, the dynamics of the economy can be adequately represented by a VAR that entails the pure censoring of $i_{t}^{\ast}$ and no kink. Such a VAR is a special case of the reduced-form VAR (\ref{eq: CKSVAR-RF}) that arises when the following testable restrictions are imposed:
\begin{equation}\label{eq: IH2}
\text{IH}_2: C_{12} = 0, C_{22} = 0  \text{ and  } \widetilde{\beta}=0.
\end{equation}
In words, $C_{12} = 0$ and $C_{22} = 0$ means that the lags of the short rate ($i_{t}$) can be excluded from equation (\ref{eq: CKSVAR-RF}), and $\widetilde{\beta}=0$ means that the slope coefficients and the variance of the errors of those equations (for $Y_{1t}$) remain the same when the economy moves across regimes. We call the structural form of the VAR under IH$_{2}$ as a purely \emph{censored }SVAR (CSVAR) as the shadow rate is censored but affects the dynamics equally across regimes.

It is worth noting that both IH$_{1}$ and IH$_{2}$ do not rely on a specific model of UMP, such as the one in Section \ref{s: dsge}, as long as the long rate or the shadow rate represents the monetary policy stance, because any VAR that includes short rates must admit a CSVAR representation with constant parameters across regimes when the irrelevance hypothesis holds. Otherwise, the ELB would result in observable changes across regimes, violating the hypothesis that the ELB is empirically irrelevant.

\subsection{Attenuation effect\label{sec_att_effect}}

\cite{SwansonWilliams2014} and \cite{SG_LS_EZ_2015AEJM} argue that an
implication of a binding ELB constraint is that the effect of shocks on long-term interest rates may be attenuated when the policy rate lies at the ELB.
They investigate this phenomenon empirically using time-varying-parameter
regressions of different maturity yields on news shocks at daily frequency.

We shall use the CKSVAR model to characterize analytically and obtain formal
tests of the aforementioned attenuation effect. Specifically, we will show
analytically how the model-implied impulse response of the long yields to a
monetary policy shock on impact is attenuated during ELB-regimes relative to
non-ELB regimes, and that this attenuation effect is state-dependent and thus time-varying.

Let $IR_{j,t}$ denote the impact response of variable $j$ to an infinitesimal monetary policy
shock derived from the CKSVAR model in Equation
(\ref{eq: CKSVAR}). Also, let $IR_{j,NA}$ denote the same response under the
assumption that there is no attenuation. Then, it can be shown that (see
\OA \ref{app: attenuation}):
\begin{equation}
IR_{j,t}=\left(  1-a_{j,t}\right)  IR_{j,NA},\qquad a_{j,t}:=\frac
{\widetilde{\beta}_{j}}{\beta_{j}}\Phi\left(  \frac{\underline{i}_{t}-i_{t|t-1}^{\ast}%
}{\varpi}\right)  , \label{eq: attenuation}%
\end{equation}
where $\beta_{j}$ and $\widetilde{\beta}_{j}$ are the $j\,$th elements of the
coefficient vectors that appear in structural and
reduced-form equations, (\ref{eq: CKSVAR}) and (\ref{eq: CKSVAR-RF}), respectively, $i_{t|t-1}^{\ast
}\allowbreak:=\allowbreak C_{21}X_{1t}\allowbreak+\allowbreak C_{22}%
X_{2t}\allowbreak+\allowbreak C_{22}^{\ast}X_{2t}^{\ast}$ is the predicted
value of the shadow rate in period $t$, $\varpi$ is the standard deviation of
$i_{t}^{\ast}-i_{t|t-1}^{\ast}$ when the monetary policy shock $\varepsilon
_{2t}$ is zero, and $\Phi\left(  \cdot\right)  $ is the Standard Normal cumulative distribution function. The factor $a_{j,t}$ attenuates the impact of the monetary policy shock on $Y_{1t,j}$ as the economy approaches the ELB, provided $a_{j,t}\in\left[  0,1\right]$. This attenuation is notably time-varying and depends on the
distance of the shadow rate from the ELB $\underline{i}_{t}$. When the shadow rate is far above the ELB, $\Phi\left(  \frac{\underline{i}_{t}-i_{t|t-1}^{\ast}}{\varpi}\right)  $ is
approximately zero and there is no attenuation. The attenuation increases in $\underline{i}_{t}-i_{t|t-1}^{\ast}$, provided $\widetilde{\beta}_{j}/\beta_{j}>0$.

Equation (\ref{eq: attenuation}) allows us to derive a formal test of the
hypothesis that there is \emph{no attenuation} based only on the
\emph{reduced-form} VAR specification (\ref{eq: CKSVAR-RF}). In other words,
this test does not rely on any additional assumptions (such as sign
restrictions) that one might use to identify the structural parameters $\beta$
and the structural impulse responses. Specifically, suppose that the vector of unconstrained variables
$Y_{1t}$ in the empirical model (\ref{eq: CKSVAR-RF}) includes a long-term
yield of a particular maturity, $\bar{R}_{L,t}$, and let $\widetilde{\beta}_{L}$
denote the coefficient of the kink in the reduced-form equation associated
with $\bar{R}_{L,t}$ in equation (\ref{eq: RF Y1}). Then, the hypothesis that there
is no attenuation of the response of $\bar{R}_{L,t}$ to monetary policy shocks
arising from the ELB is given by:
\begin{equation}
    \text{H}_{NA}:\widetilde{\beta}_{L}=0. \label{eq: hna}
\end{equation}
In the theoretical model studied in Section \ref{s: dsge}, no attenuation effect is equivalent to the irrelevance of the ELB as implied by Proposition \ref{prop1}(ii). But the empirical model is less restrictive than the theoretical model, and H$_{NA}$ is clearly a weaker hypothesis than the irrelevance hypotheses IH$_{1}$
(\ref{eq: IH}) and IH$_{2}$ (\ref{eq: IH2}) defined above. Therefore, in our empirical framework, failing
to reject H$_{NA}$ does not necessarily imply that the ELB is empirically irrelevant.

We will test H$_{NA}$ using long yields at various maturities in Section
\ref{sec_attenuation}.

\section{Empirical results\label{s: data}}

This section discusses the data for the U.S. and Japan, and tests the hypothesis that the ELB has been empirically irrelevant in each country using the hypothesis testing approach that is underpinned by Proposition \ref{prop1} and established in Section \ref{s: IH}. In addition, it tests the hypothesis of no attenuation effect in the response of long-term interest rates at the ELB, as formulated in Section \ref{sec_att_effect}, and it examines whether the short-term interest rate, including the shadow rate at the ELB, is a sufficient indicator of monetary policy including UMP, as implied by Proposition \ref{prop2}.

\subsection{Data}

Our empirical analysis focuses on the U.S. and Japan. We choose data series for the baseline specification of the SVAR model to maintain the closest specification possible to related studies and thereby include representative series for inflation, output, and measures for short- and long-term yields.
\begin{figure}[ptb]%
\centering
\includegraphics[
height=2.8245in,
width=4.2289in
]%
{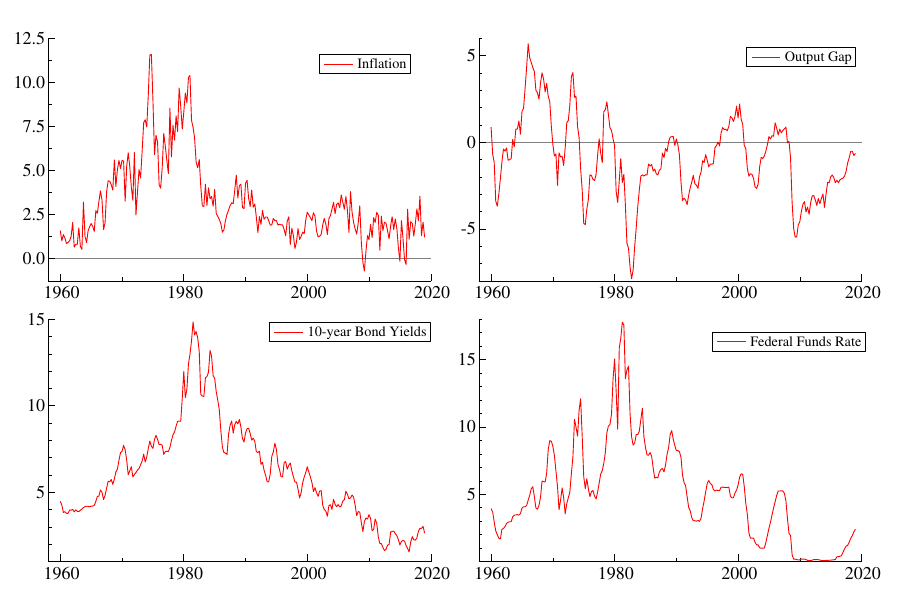}%
\caption{U.S. quarterly data}
\label{fig: US data}%
\end{figure}

For the U.S., we use quarterly data for inflation based on the GDP deflator \citep{FredGDPDEF}, a measure of the output gap \citep{FredGDP, FredGDPPOT}, the short-term interest rate from the Federal Funds Rate \citep{FredFEDFUNDS}, and the 10-year government bond yield from the 10-year Treasury constant maturity rate \citep{FredLR}. Figure \ref{fig: US data} plots these series. We also consider the different measures of monetary aggregates listed in \OA \ref{app:data}. The data are from the FRED database at the Federal Reserve Bank of St. Louis \citep{FredMA} and \cite{Divisia} databases. The estimation sample for the baseline specification is from 1960q1 to 2019q1.\footnote{See \OA \ref{app:data} for further details about the data. Alternative specifications with money are estimated over different time periods due to constraints on data availability.} We set the value of the effective lower bound on the Federal Funds Rate equal to 0.2, such that the short-term interest rate is at the ELB regime for 11 percent of the time, which is consistent with \cite{Bernanke_Reinhart_AERPP04} who suggest that the effective lower bound on nominal interest rates may be above zero for institutional reasons.%

For Japan, we use quarterly data for core CPI inflation, a measure of the output gap provided by the Bank of Japan, and the call rate. In addition, we use two alternative measures for long yields: the 9-year and the 10-year government bond yields, which are available for different sample periods. The data sources are the Bank of Japan for the output gap \citep{BOJOG} and the call rate \cite{BOJCR}, the Ministry of Finance for the 9-year and the 10-year government bond yields \cite{MOFLR}, and Statistics Bureau of Japan for core CPI inflation \citep{SBJINF}.
The available sample is from 1985q3 to 2019q1 if we include the 9-year government bond yield in the VAR, which is our baseline case, and from 1987q4 to 2019q1 if we use the 10-year yield. Following \cite{HayashiKoeda2019}, we set the ELB to track the interest on reserves (IOR) \citep{BOJIOR1,BOJIOR2}.\footnote{Specifically, ELB = IOR + 7bp, which is slightly higher than \cite{HayashiKoeda2019} who use IOR+5bp, in order to treat 2016q1 as being at the ELB.}  For the sample period 1985q3-2019q1, the call rate is at the ELB for 49 percent of the observations. Following \cite{HayashiKoeda2019}, we use a trend growth series \citep{CAOTG} to account for the declining equilibrium real interest rate in Japan during the 1990s.\footnote{Specifically, we use the annual average growth rate of potential GDP as an additional control in our model. See \citet[pp. 1081--1083]{HayashiKoeda2019} for an extended discussion of this issue and its implications.} Figure \ref{fig: JP data} plots these series.

\begin{figure}[ptb]%
\centering
\includegraphics[
height=2.8245in,
width=4.2289in
]%
{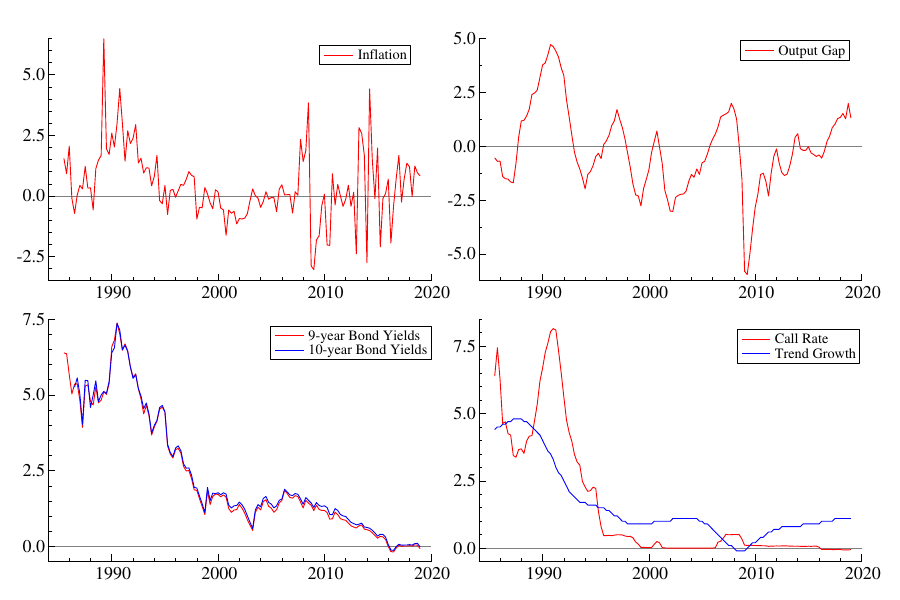}%
\caption{Japanese quarterly data}
\label{fig: JP data}%
\end{figure}

\subsection{Testing the irrelevance hypothesis of the ELB\label{s: IHtest}}

We now test the irrelevance hypothesis of the ELB for the U.S. and Japan using the two null hypotheses, IH$_{1}$ and IH$_{2}$, formalized by (\ref{eq: IH}) and (\ref{eq: IH2}), respectively.

\begin{table}[t]
\begin{centering}
\caption{Test for excluding short rates from VAR that includes long rates}%
\label{tb:excl_sr}%
\begin{tabular}
[c]{c|rrrrrr|rrrrrr}\hline\hline
\multicolumn{13}{c}{Panel A: KSVAR} \\ \hline
 & \multicolumn{6}{c|}{United States} & \multicolumn{6}{c}{Japan} \\\hline
p & loglik  & pv-p & AIC & LR & df & $p$-val & loglik &  pv-p & AIC & LR & df & $p$-val \\\hline
5 & -213.4  & - & 2.62 & 53.12 & 18 & 0.000 & 248.1 &  - & -2.18 & 27.82 & 18 & 0.065\\
4 & -221.5 & 0.446 & 2.55 & 49.57 & 15 & 0.000 & 239.9 &  0.425 & -2.30 & 28.10 & 15 & 0.021\\
3 & -234.4  & 0.112 & 2.53 & 42.13 & 12 & 0.000 & 232.2 & 0.471 & -2.42 & 28.58 & 12 & 0.004\\
2 & -266.0  & 0.000 & 2.66 & 41.93 & 9 & 0.000 & 223.8 &  0.445 & -2.53 & 25.71 & 9 & 0.002\\
1 & -296.7  & 0.000 & 2.78 & 32.87 & 6 & 0.000 & 184.8 &  0.000 & -2.19 & 32.32 & 6 & 0.000\\\hline
\multicolumn{13}{c}{} \\[-2ex]
\multicolumn{13}{c}{Panel B: CKSVAR} \\ \hline
p & loglik  & pv-p & AIC & LR & df & $p$-val & loglik &  pv-p & AIC & LR & df & $p$-val \\\hline
5 & -191.3 & - & 2.60 & 82.43 & 33 & 0.000 & 284.7 & - & -2.42 & 90.39 & 33 & 0.000\\
4 & -202.7 & 0.290 & 2.53 & 72.15 & 27 & 0.000 & 277.1 & 0.766 & -2.61 & 91.55 & 27 & 0.000\\
3 & -223.0 & 0.011 & 2.53 & 51.93 & 21 & 0.000 & 258.1 & 0.081 & -2.62 & 73.52 & 21 & 0.000\\
2 & -256.3 & 0.000 & 2.64 & 49.44 & 15 & 0.000 & 242.1 & 0.018 & -2.68 & 56.16 & 15 & 0.000\\
1 & -290.3 & 0.000 & 2.76 & 37.12 & 9 & 0.000 & 204.8 & 0.000 & -2.43 & 63.03 & 9 & 0.000 \\\hline\hline

\end{tabular}
\\
\end{centering}
\fnote{\footnotesize {Note: Panel A reports results for a KSVAR(p) with inflation, output gap, long rate and policy rate. Panel B reports corresponding results for a CKSVAR(p) that includes shadow rates. Estimation sample is 1960q1-2019q1 for the U.S. and 1985q3-2019q1 for Japan. Long rates are 10-year government bond yields for the U.S. and 9-year yields for Japan. loglik is the value of the log-likelihood. pv-p is the asymptotic $p$-value of a LR test of (C)KSVAR($p$) against (C)KSVAR($p+1$). AIC is the Akaike information criterion. LR is the test statistic for excluding short rates from equations for inflation, output gap and long rates. df is the number of restrictions. $p$-val is the asymptotic $\chi^2_{df}$ $p$-value of the test.}}\end{table}

\paragraph{Testing IH$_{1}$} The basic idea of IH$_{1}$ is that a long-term interest rate is a sufficient indicator of monetary policy stance, and thereby a short-term interest rate becomes redundant and can be excluded from the VAR under the irrelevance of the ELB. The results of the likelihood ratio test of the null hypothesis IH\textsubscript{1} in (\ref{eq: IH}) are reported in Table \ref{tb:excl_sr}. Panel A reports results based on a KSVAR model, in which lags of the shadow rate $i_{t}^{\ast}$ do not appear on the right-hand-side of equation (\ref{eq: CKSVAR-RF}). In this case, we test IH$_{1}$ against an alternative hypothesis that imposes $C^*_{12}=0$, and this test only has power against the violation of $C_{12}=0$. We do so because the KSVAR model is simpler to estimate, and a rejection of $C_{12} = 0$ would suffice to reject IH$_{1}$. Panel B reports the results using the general CKSVAR. The table reports results for specifications with different lag lengths of the VAR($p$), where $p=1,\ldots, 5$. Column pv-p reports the $p$-value of a test for selecting the number of lags in the model,which is an alternative approach to the Akaike Information Criterion (AIC), also reported in the Table. Both measures consistently select four lags for the U.S. and two lags for Japan. Column \emph{p}-val reports the asymptotic $p$-value of our test of (\ref{eq: IH}).

\begin{table}[t]
\begin{centering}
\caption{Test for excluding short rates from VAR with long rates by maturity}%
\label{tb:excl_sr_long}%
\begin{tabular}
[c]{c|rrrrr|rrrrr}\hline\hline
\multicolumn{11}{c}{Panel A: KSVAR} \\ \hline
 & \multicolumn{5}{c|}{United States} & \multicolumn{5}{c}{Japan} \\\hline
YTM & sample & p & LR & df & $p$-val & sample & p & LR & df & $p$-val \\\hline
7 & 1970q4--2019q1 & 3 & 38.69 & 12 & 0.000 & 1985q3--2019q1 & 2 & 23.20 & 9 & 0.006\\
5 & 1963q2--2019q1 & 3 & 41.93 & 12 & 0.000 & 1985q3--2019q1 & 2 & 25.15 & 9 & 0.003\\
3 & 1963q2--2019q1 & 3 & 46.32 & 12 & 0.000 & 1985q3--2019q1 & 2 & 30.97 & 9 & 0.000\\
2 & 1977q4--2019q1 & 4 & 38.82 & 15 & 0.001 & 1985q3--2019q1 & 2 & 33.05 & 9 & 0.000\\
1 & 1963q2--2019q1 & 3 & 59.65 & 12 & 0.000 & 1985q3--2019q1 & 2 & 49.90 & 9 & 0.000\\\hline
\multicolumn{9}{c}{} \\[-2ex]
\multicolumn{11}{c}{Panel B: CKSVAR} \\ \hline
YTM & sample & p & LR & df & $p$-val & sample & p & LR & df & $p$-val \\\hline
7 & 1970q4--2019q1 & 3 & 51.03 & 21 & 0.000 & 1985q3--2019q1 & 2 & 58.32 & 15 & 0.000\\
5 & 1963q2--2019q1 & 3 & 56.33 & 21 & 0.000 & 1985q3--2019q1 & 2 & 58.50 & 15 & 0.000\\
3 & 1963q2--2019q1 & 3 & 58.54 & 21 & 0.000 & 1985q3--2019q1 & 2 & 62.04 & 15 & 0.000\\
2 & 1977q4--2019q1 & 4 & 66.84 & 27 & 0.000 & 1985q3--2019q1 & 2 & 63.77 & 15 & 0.000\\
1 & 1963q2--2019q1 & 3 & 73.06 & 21 & 0.000 & 1985q3--2019q1 & 2 & 72.41 & 15 & 0.000\\\hline\hline
\end{tabular}
\\
\end{centering}
\fnote{\footnotesize {Note: Panel A reports results for a KSVAR(p) with inflation, output gap, long rate, and policy rate. Panel B reports corresponding results for a CKSVAR(p) that includes shadow rates. Estimation sample varies for each long-term rate used for the U.S. and 1985q3--2019q1 for Japan. Long rates are government bond yields with year to maturity reported in the first column (YTM). p is the preferred lag length selected by AIC criteria. LR is the test statistic for excluding short rates from equations for inflation, output gap and long rates. df is the number of restrictions. $p$-val is the asymptotic $\chi^2_{df}$ $p$-value of the test.}}\end{table}

Table \ref{tb:excl_sr} shows that the data strongly reject the exclusion restrictions implied by IH$_{1}$ for both countries and in both the KSVAR and CKSVAR specifications. Hence, the short-term interest rate cannot be excluded from the VAR, and thus it makes the dynamics of the system differ between the non-ELB and ELB regimes even if the long-term interest rate -- 10-year yields for the U.S. and 9-year yields for Japan -- is added to the VAR. The result continues to hold for Japan with 10-year yields, which are available from 1987q3, and is reported in Table \ref{tb:excl_sr_10yearJP} in \OA \ref{app: empirical}.

The result of the rejection of IH$_{1}$ continues to hold even if we use yields with shorter maturities. Table \ref{tb:excl_sr_long} shows the results for the null hypothesis IH\textsubscript{1} when we include yields with maturities in the range from 1 to 7 years in the VAR for the U.S. \citep{FredLRall} and Japan \citep{MOFLR}. The entries show results for the VAR specification with the preferred number of lags according to the AIC. The results show that the exclusion restrictions from IH$_{1}$ are strongly and consistently rejected across the whole range of maturities for both countries. Our results imply that a government bond yield with maturity within the 1- to 10-year range cannot replace the shadow rate as the indicator of monetary policy stance, and the dynamics of the economy differ between the non-ELB and ELB regimes.

\paragraph{Testing IH$_{2}$}
The idea of the hypothesis IH$_{2}$ in (\ref{eq: IH2}) is that once the shadow rate is included in the VAR, the short-term interest rate can be excluded from the VAR under the irrelevance hypothesis. We test IH$_{2}$ with the three core observables, inflation and output gap in $Y_{1t}$, and the short-term interest rate, and we also include the long-term interest rate in $Y_{1t}$ for robustness.
Under IH$_{2}$ the shadow rate is censored in the ELB regime but affects $Y_{1t}$ equally in both regimes according to the reduced form equation (\ref{eq: CKSVAR-RF}).
We test IH\textsubscript{2} using a likelihood ratio test as we tested IH$_{1}$.

Table \ref{tb:csvar_lr_long} reports the results of the likelihood ratio test of IH\textsubscript{2} for the U.S. and Japan. The results for our baseline specifications are reported in the row starting with `10 or 9' in the first column. We include $4$ lags for the U.S. and $2$ lags for Japan, according to the AIC. The results show that the IH$_{2}$ is rejected for both economies at the 5 percent level of significance. The result continues to hold for Japan with 10-year yields and is reported in Table \ref{tb:csvar-cksvar_10yearJP} in \OA \ref{app: empirical}. Table \ref{tb:csvar_lr_long} also shows that the result of the rejection of IH$_{2}$ continues to hold even if we include an alternative yield within the range of 1 to 7 years in $Y_{1t}$.

\begin{table}[t]
\begin{centering}
\caption{Testing CSVAR against CKSVAR with long rates by maturity}%
\label{tb:csvar_lr_long}%
\begin{tabular}
[c]{c|lcrcc|lcrcc}\hline\hline
 & \multicolumn{5}{c|}{United States} & \multicolumn{5}{c}{Japan} \\\hline
YTM & sample & p & LR & df & $p$-val & sample & p & LR & df & $p$-val \\\hline
10 or 9 & 1960q1--2019q1 & 4 & 34.42 & 19 & 0.016 & 1985q3-2019q1 & 2 & 51.02 & 11 & 0.000 \\
%8 & - & - & - & - & - & 1985q3--2019q1 & 2 & 48.65 & 11 & 0.0000\\
7 & 1970q4--2019q1 & 3 & 26.80 & 15 & 0.030 & 1985q3--2019q1 & 2 & 50.86 & 11 & 0.000\\
%6 & - & - & - & - & - & 1985q3--2019q1 & 2 & 52.24 & 11 & 0.0000\\
5 & 1963q2--2019q1 & 3 & 29.26 & 15 & 0.015 & 1985q3--2019q1 & 2 & 49.27 & 11 & 0.000\\
%4 & - & - & - & - & - & 1985q3--2019q1 & 2 & 47.41 & 11 & 0.0000\\
3 & 1963q2--2019q1 & 3 & 31.83 & 15 & 0.007 & 1985q3--2019q1 & 2 & 53.09 & 11 & 0.000\\
2 & 1977q4--2019q1 & 4 & 45.49 & 19 & 0.001 & 1985q3--2019q1 & 2 & 56.03 & 11 & 0.000\\
1 & 1963q2--2019q1 & 3 & 51.04 & 15 & 0.000 & 1985q3--2019q1 & 2 & 70.62 & 11 & 0.000\\\hline\hline
\end{tabular}
\\
\fnote{\footnotesize Note: The estimated model is a CKSVAR(p) for the U.S. and Japan with inflation, output gap, policy rate, and a measure of long-term rate. Long rates are government bond yields with year to maturity reported in the first column (YTM). As baseline specifications, the 10-year yield is used for the U.S. and the 9-year yields is used for Japan. Sample availability varies for each long-term rate used. LR is the value of the LR test statistic. df is the number of exclusion restrictions. $p$-val is the asymptotic $\chi^2_{df}$ $p$-value of the test.}
\end{centering}
\end{table}

\paragraph{Robustness checks}\label{robust_check_editor} It is possible that the short rate is necessary to explain long rates even if the short rate is not necessary to characterize inflation and output dynamics. This could occur if real rates trend over time, so having both in the VAR picks up low-frequency movements.
We address this concern by testing IH$_{1}: C_{12}=C_{12}^{\ast}=0$ and $\tilde{\beta}=0$ for the inflation and output equations only while allowing the short and shadow rates to affect the long rate.
We find that this weaker version of IH$_{1}$ is also robustly rejected (Table \ref{tb:excl_sr_no_inf_out} in \OA \ref{app: empirical}). 

Next we check the robustness of our results to possible omission of alternative channels of unconventional monetary policy, by adding money growth to the $Y_{1t}$ variables of the VAR that we use to test the null hypothesis IH\textsubscript{1} in (\ref{eq: IH}). Using several different monetary aggregates for the U.S., we consistently reach the same conclusion: the IH\textsubscript{1} is firmly rejected (Table \ref{tb:excl_sr_mon} in \OA \ref{app: empirical}).

\label{robust_check}We also check the robustness of the U.S. results to the well-documented fall in macroeconomic volatility in the mid-1980s, known as the Great Moderation, as well as a possible change in monetary policy regime occurring at that time by performing the same tests over the subsample 1984q1--2019q1 (see Tables \ref{tb:excl_sr_sub} and \ref{tb:csvar-cksvar_sub} in \OA \ref{app: empirical}). In addition, \cite{Dario_AEJM19} show that accounting for the endogenous reaction of financial conditions is critical to avoiding an attenuation bias in the responses of variables to monetary policy shocks. \cite{GZ_AER12} use credit spreads to internalize the central contribution of financial frictions %and the change in macroeconomic volatility
during ELB episodes. Thus, we check the robustness of our results by including credit spreads \citep{GZ_AER_DATA, FredCS} in the VAR (see Tables \ref{tb:excl_sr_app}-\ref{tb:csvar-cksvar_app} in \OA \ref{app: empirical}). Our conclusion remains the same: the IH$_1$ and IH$_2$ are firmly rejected.

Finally, we check the power of our irrelevance tests IH$_1$ and IH$_2$ by generating simulated series from our theoretical model for values of $\xi^{\ast}$ in the range $[0.7, 0.99]$. Since the theoretical model is solved under Proposition \ref{prop2}, $\xi^{\ast}$ in the theoretical model coincides with $\xi$ in the empirical model and the simulated series are equivalent between the two models. We find that the rejection rate of the tests declines as $\xi^{\ast}$ approaches the value of one (i.e., the irrelevance hypothesis holds true), showing that our tests are powerful. The result is reported in Tables \ref{tb:excl_sr_tab1} and \ref{tb:excl_sr_tab2} in \OA \ref{app: empirical}.

\subsection{Testing the attenuation effect\label{sec_attenuation}}

We use our CKSVAR model to test the null hypothesis of no attenuation in the response of long-term interest rates to monetary policy shocks at the ELB, formalised as H$_{NA}$ in \eqref{eq: hna} in Section \ref{sec_att_effect}. Table \ref{tb:no_atten} reports the results of the test for the VAR model that includes long rates of different maturities ranging from 1 year to 10 years. The null is firmly and consistently rejected across different yields to maturity for both countries.
The result implies that long rates of different maturities have responded differently to monetary policy shocks between the ELB and non-ELB regimes in both the U.S. and Japan.

\label{ref_reply_SW}Our approach to testing no attenuation effect differs from the one employed by \cite{SwansonWilliams2014} who find that responses of 1- and 2-year U.S. bond yields to various macroeconomic news shocks, identified using high frequency data over the period 1990--2012, were not attenuated throughout 2008 to 2010 but became attenuated in late 2011. We instead focus on responses to a monetary policy shock -- a shock to the shadow rate using quarterly data over a longer sample. We discuss the implications of different sample periods at the end of \OA \ref{app: empirical}.

\begin{table}[t]
\begin{centering}
\caption{Test for no attenuation}%
\label{tb:no_atten}%
\begin{tabular}
[c]{c|rrr|rrr}\hline\hline
 & \multicolumn{3}{c|}{United States} & \multicolumn{3}{c}{Japan} \\\hline
YTM & p & LR & $p$-val & p & LR & $p$-val \\\hline
10 or 9 & 4 & 15.84 & 0.000 & 2 & 16.63 & 0.000\\
7 & 3 & 10.33 & 0.001 & 2 & 14.95 & 0.000\\
5 & 3 & 13.19 & 0.000 & 2 & 15.13 & 0.000\\
3 & 3 & 17.15 & 0.000 & 2 & 27.70 & 0.000\\
2 & 4 & 16.81 & 0.000 & 2 & 33.17 & 0.000\\
1 & 3 & 35.90 & 0.000 & 2 & 46.06 & 0.000\\\hline\hline
\end{tabular}
\\
\end{centering}
\fnote{\footnotesize {Note: The table reports corresponding results for a CKSVAR(p) that includes shadow rates. Estimation sample is 1960q1--2019q1 for the U.S. and 1985q3--2019q1 for Japan. Long rates are government bond yields with year to maturity reported in the first column. As baseline specifications, the 10-year yield is used for the U.S. and the 9-year yield is used for Japan. p is the preferred lag length selected by AIC criteria. LR is the likelihood ratio test statistic for the hypothesis of no attenuation. $p$-value is the asymptotic $p$-value of the text.}}\end{table}

\subsection{Testing the (ir)relevance of long rates}

Our statistical tests rejected the irrelevance hypothesis of the ELB and
 the possibility of excluding the short rate by controlling for the long rate.
We now assess whether movements in the short rate, including the shadow rate during ELB regimes, are sufficient to encapsulate the effects of both conventional and unconventional monetary policies for inflation and output by testing the exclusion restriction on the long rate. Proposition \ref{prop2} implies that the exclusion restriction on the long rate holds in the theoretical model.

\begin{table}[t]
\begin{centering}
\caption{Test for excluding long rates from VAR by maturity}%
\label{tb:cksvar_excl_long}%
\begin{tabular}
[c]{c|lcrcc|lcrcc}\hline\hline
 & \multicolumn{5}{c|}{United States} & \multicolumn{5}{c}{Japan} \\\hline
YTM & sample & p & LR & df & $p$-val & sample & p & LR & df & $p$-val \\\hline
10 or 9 & 1960q1--2019q1 & 3 & 2.43 & 6 & 0.876 & 1985q3--2019q1 & 2 & 6.32 & 4 & 0.177 \\
7 & 1970q4--2019q1 & 3 & 5.15 & 6 & 0.524 & 1985q3--2019q1 & 2 & 10.38 & 4 & 0.035\\
5 & 1963q2--2019q1 & 3 & 4.99 & 6 & 0.545 & 1985q3--2019q1 & 2 & 7.00 & 4 & 0.136\\
3 & 1963q2--2019q1 & 3 & 7.05 & 6 & 0.316 & 1985q3--2019q1 & 2 & 4.99 & 4 & 0.288\\
2 & 1977q4--2019q1 & 4 & 11.53 & 8 & 0.174 & 1985q3--2019q1 & 2 & 3.40 & 4 & 0.493\\
1 & 1963q2--2019q1 & 3 & 11.70 & 6 & 0.069 & 1985q3-2019q1 & 2 & 5.19 & 4 & 0.268\\\hline\hline
\end{tabular}
\\
\fnote{\footnotesize Note: The estimated model is a CKSVAR(p) for the U.S. and Japan with inflation, output gap, policy rate and a different measure of long-term rates. Long rates are government bond yields with year to maturity reported in
the first column (YTM). As baseline specifications, the 10-year yield is used for the U.S. and the 9-year yield is used for Japan. Sample availability varies for each long-term rate used. The null is that the lags of long rates can be excluded from the equations describing the dynamics of inflation and the output gap from the unrestricted model. $p$ is chosen by AIC. LR is the value of the LR test statistic, df is number of exclusion restrictions, and $p$-val is the asymptotic $\chi^2_{df}$ $p$-value of the test.}
\end{centering}
\end{table}

We perform this test using the CKSVAR model that includes inflation, output gap, the long rate, and the short rate. 
The null hypothesis is that lags of the long rate 
can be excluded from the equations for inflation and output gap. In performing this test, we impose the assumption that the short rate does not react to the long rate by imposing zero restrictions on the lags of long rates in the interest rate equation (\ref{eq: RF Y2}). This is consistent with the standard Taylor rule (\ref{Taylor}) in the theoretical model, and our modelling of QE in equation (\ref{QErule}). Table \ref{tb:cksvar_excl_long} reports the results using long rates of various maturities. It shows that the null hypothesis cannot be rejected and that the long rate can be excluded from the model at the 5 percent level of significance in both countries and at most maturities. The result implies that the short rate and the shadow rate may be sufficient indicators of monetary policy during the non-ELB and ELB regimes in both countries.

\section{The effectiveness of UMP}\label{sec: impact of mp}

Our testing results established that the dynamics of the economy are different across the ELB and non-ELB regimes for both the U.S. and Japan, leading us to conclude that the ELB has been empirically relevant. But the results are silent on the magnitude of the differences in the effects of monetary policy between the two regimes.  Here we address this issue by estimating the (partially identified) impulse responses to a monetary policy shock from the CKSVAR models over time to gauge the effectiveness of UMP relative to conventional monetary policy.
In doing so, we use inflation, output gap, the short rate, and the shadow rate for the CKSVAR models, as this specification is broadly in line with the test results obtained in Section \ref{s: data}.

\subsection{State-dependent impulse responses}\label{sec: IRFs}

Since the empirical model is nonlinear, the impulse response functions (IRFs) are state-dependent. We will follow the approach in \cite{KoopPesaranPotter1996}, already used in Section \ref{s: simulations}, according to which the IRF to a monetary policy shock of magnitude $\varsigma$ is given by the difference in the expected path of the endogenous variables when the policy shock takes the value $\varsigma$, versus the path when the shock is zero, conditional on the state of the economy prior to the shock. This approach is the
most commonly used in the literature, see, e.g., \cite{HayashiKoeda2019}. In our model, there is an additional complication that lagged shadow rates are unobserved, so we evaluate the IRFs at the smoothed estimates of those latent variables. In the same notation of the empirical model in Section \ref{s: cksvar}, let $Y_{t}\equiv (Y_{1t}^{\prime}, i_{t})^{\prime}$ denote the vector of endogenous variables, and $\overline{X}_{t,j}^{\ast}$ denote a state vector whose $j$-th component is given by $\min(  i_{t-j}^{\ast}-\underline{i}_{t-j},0)$ for
$j=1,...,p$, where $p$ is the order of the VAR. Then, our IRFs starting from period $t$ up to the horizon $h$ are given by:
\begin{equation}
IRF_{h,t}\left(  \varsigma,X_{t},\widehat{\overline{X}}_{t}^{\ast}\right)
=E\left(  Y_{t+h}|\varepsilon_{2t}=\varsigma,X_{t},\widehat{\overline{X}}%
_{t}^{\ast}\right)  -E\left(  Y_{t+h}|\varepsilon_{2t}=0,X_{t}%
,\widehat{\overline{X}}_{t}^{\ast}\right)  ,\label{eq: IRFs}%
\end{equation}
where $X_{t}$ consists of the lagged values of $Y_{t-j}$ for $j=1,...,p$, and $\widehat{\overline{X}}_{t}^{\ast}$ is the smoothed estimate of the
state vector $\overline{X}_{t}^{\ast}$ when it is unobserved.

\subsection{The impact effects of UMP\label{sec_foot_shintani}}

As explained in Section \ref{s: cksvar}, the IRFs are generally set-identified unless we assume there is no contemporaneous effect of UMP
on $Y_{1t}$, which corresponds to setting $\xi=0$ in the CKSVAR model.
We will not be imposing such an assumption in our analysis. We proceed by first obtaining the identified set for $\xi$, $\beta$ and $\gamma$ by solving equations (\ref{eq: betatilde}) and
(\ref{eq: gamma}) at the estimated values of $\widetilde{\beta}$ and $\Omega,$
as explained in Section \ref{s: cksvar} (see the discussion following equations (\ref{eq: betatilde}) and (\ref{eq: gamma})), and then we simulate the paths of the empirical model at each of the values of the structural parameters in the identified set.

The estimation results for the parameter $\xi$ are as follows. Recall that the parameter $\xi$ determines the impact effect of UMP, where the two limiting cases of $\xi=0$ and $\xi=1$ correspond to UMP being completely ineffective on impact and as effective as conventional policy in non-ELB regimes on impact, respectively. When we restrict the range of $\xi$ to $\left[  0,1\right]  $ and impose no further
identifying restrictions, the identified set for $\xi$ is $[0,0.75]$ for the
U.S. and $[0,0.34]$ for Japan.

We sharpen the identified sets by using sign restrictions. We follow \cite{DebortoliGaliGambetti2019}
and impose the restrictions that a negative monetary policy shock should have a nonnegative effect on inflation and output, and a nonpositive effect on interest rates at a one-year horizon.\footnote{Note that because IRFs are state-dependent, these sign restrictions need to be imposed for all values of the initial states.  In principle, this means working out the worst cases over the support of the distribution of the variables. However, a very similar conservative estimate of the identified set can be obtained if we simply impose the sign restrictions in every period.} These sign restrictions are consistent with the theoretical model studied in Section
\ref{s: dsge} (see Figure \ref{Fig:MPshock}). With these sign restrictions, the identified set for the parameter $\xi$ narrows down substantially for the U.S., from $[0,0.75]$ to $[0.71,0.73]$. For Japan the impact of the sign restrictions is more modest, from $[0,0.34]$ to $[0,0.26]$.

\label{foot_shintani}\OA \ref{app_Chole} compares our identification based on sign restrictions against the standard Choleski identification that imposes restrictions on the impact response of the variables to the monetary policy shock. Consistent with the results in \cite{GertlerKaradi2015} and \cite{kubota2022macro}, we find that the Choleski identification presents several puzzling responses.

\begin{figure}[t]
\caption{Impulse responses to a monetary policy shock in the U.S.}%
\label{fig: set_irf_us}
%[ptb]
\centering
\includegraphics[
height=3.8245in,
width=6.2289in
]{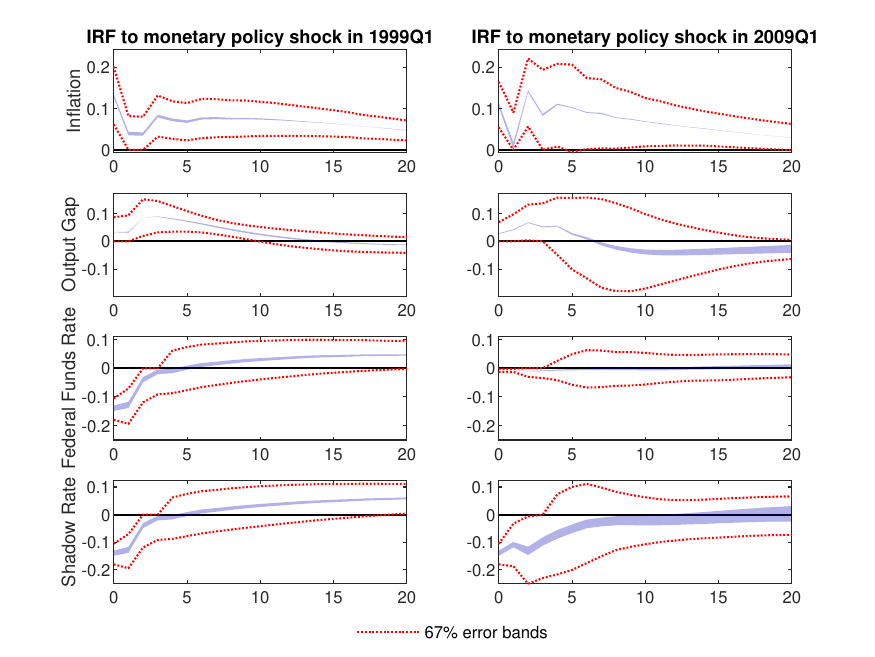}
\fnote{\footnotesize{Note: Identified sets of IRFs in 1999q1 and 2009q1 to a -25bps monetary policy shock estimated from CKSVAR(3) model in inflation, output gap, the Federal Funds Rate, and the shadow rate for the U.S. over the period 1960q1-2019q1, identified by the sign restrictions that the shock has nonnegative effects on inflation and output and nonpositive effects on the federal funds rate and the shadow rate up to four quarters. Dotted lines give 67 percent error bands.}}
\end{figure}

\begin{figure}[tbh]
\caption{Impulse responses to a monetary policy shock in Japan}%
\label{fig: set_irf_jp}
%[ptb]
\centering
\includegraphics[
height=3.8245in,
width=6.2289in
]{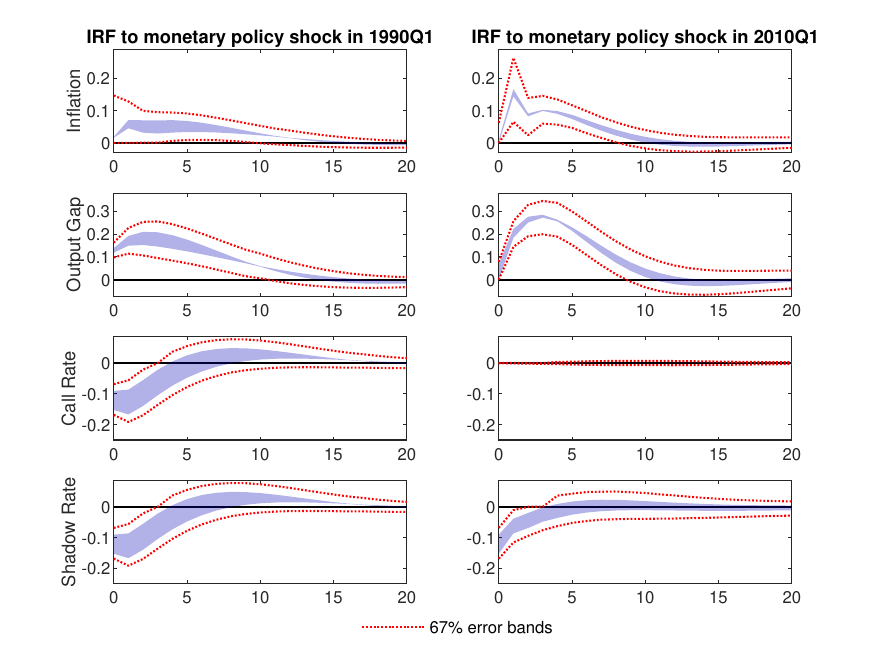}
\fnote{\footnotesize{Note: Identified sets of IRFs in 1990q1 and 2010q1 to a -25bps point monetary policy shock estimated using a CKSVAR(2) model in inflation, output gap, the call rate, and the shadow rate for Japan over the period 1985q3-2019q1, identified by the sign restrictions that the shock has nonnegative effects on inflation and output and nonpositive effects on the call rate and the shadow rate up to four quarters. Dotted lines give 67 percent error bands.}}
\end{figure}

\subsection{The dynamic effects of UMP}

Now we study the dynamic effects of UMP by examining the IRFs for the identified sets of $\xi$ for the U.S. and Japan. Figures \ref{fig: set_irf_us} and \ref{fig: set_irf_jp} report the identified IRFs of inflation, the output gap, the policy rate, and the shadow rate to a -25 basis points monetary policy shock. The figure also reports asymptotic confidence intervals obtained using the method of \cite{ImbensManski2004}, where we also impose the sign restrictions on the confidence bands, as in \cite{GranzieraMoonSchorfheide2018}. The IRFs are computed at two different dates: the left panels report IRFs at dates when interest rates are well above the ELB (1999q1 for the U.S. and 1990q1 for Japan) so that monetary policy is conventional; the right panels report IRFs at dates when interest rates are at the ELB (2009q1 for the U.S. and 2010q1 for Japan) so that monetary policy is unconventional. 

The policy effects differ across the two periods. For both countries, the conventional monetary policy shock -- the reduction in the short-term interest rate -- has a larger \emph{contemporaneous} effect on all variables in the pre-ELB dates than the corresponding unconventional policy shock -- the reduction in the shadow rate -- during the ELB dates, and the difference is larger in Japan than in the U.S. This is because $\xi$ is estimated lower in Japan than in the U.S. However, in Japan the impulse responses to the unconventional policy shock appear to be stronger a few quarters out.\footnote{The reason why the delayed effects of UMP can be stronger than conventional policy even though $\xi<1$ is because in the empirical model the coefficients on the lags of the shadow rate are completely unrestricted. This is more general than the theoretical model of Section \ref{s: dsge} with the monetary policy rule (\ref{Taylor}), where the coefficient on the lagged shadow rate was restricted to be a constant fraction $\lambda^{\ast}$ of the coefficient on the lagged policy rate above the ELB. The result of that restriction was that $\lambda^{\ast}<1$ restricted UMP to have a uniformly weaker effect than conventional policy over all horizons. We did not need to impose this overidentifying restriction in the empirical analysis.} In addition, the responses of short rates and shadow rates are identical when both rates are far above the ELB. By contrast, the shadow rates decrease whilst the responses of short rates are muted during the ELB dates. The reductions in shadow rates affect inflation and output through unconventional policy during the ELB regime.

\begin{figure}[t]
\caption{Responses to monetary policy shock in the U.S. over time}%
\label{fig: set_irf_us_cum}
%[ptb]
\centering
\includegraphics[
height=3.5245in,
width=5.6289in
]{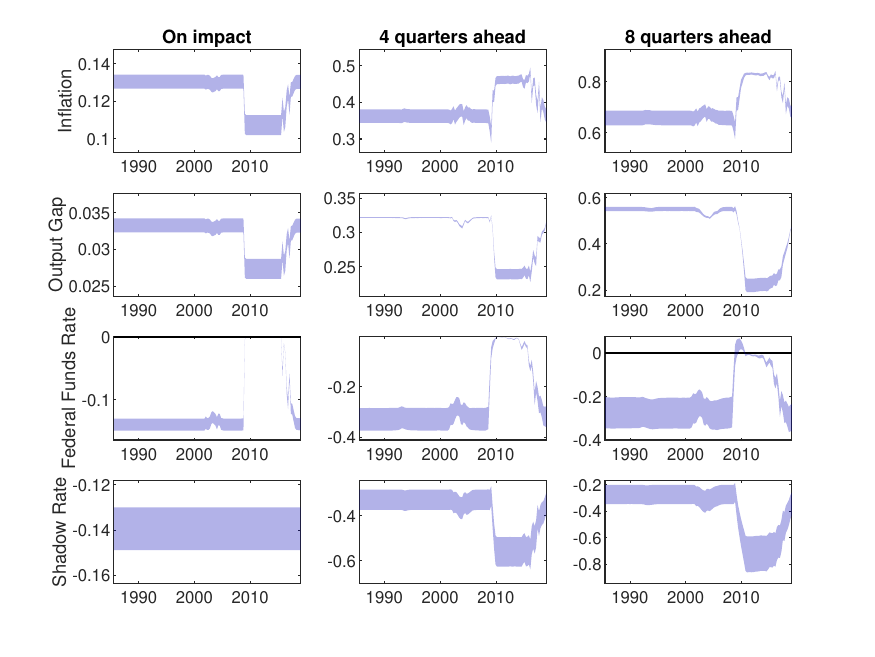}
\vspace{.1cm}
\fnote{\footnotesize {Note: Identified sets of (cumulative) impulse responses to a negative -25bps monetary policy shock at three different horizons, on impact, one year out and two years out, estimated using a CKSVAR(3) model in inflation, output gap, the Federal Funds Rate, and the shadow rate for the U.S. over the period 1960q1-2019q1, identified by the sign restrictions that the shock has nonnegative effects on inflation and output and nonpositive effects on the federal funds rate and the shadow rate up to four quarters.}}
\end{figure}

\begin{figure}[tbh]
\caption{Responses to monetary policy shock in Japan over time}%
\label{fig: set_irf_jp_cum}
%[ptb]
\centering
\includegraphics[
height=3.5245in,
width=5.6289in
]{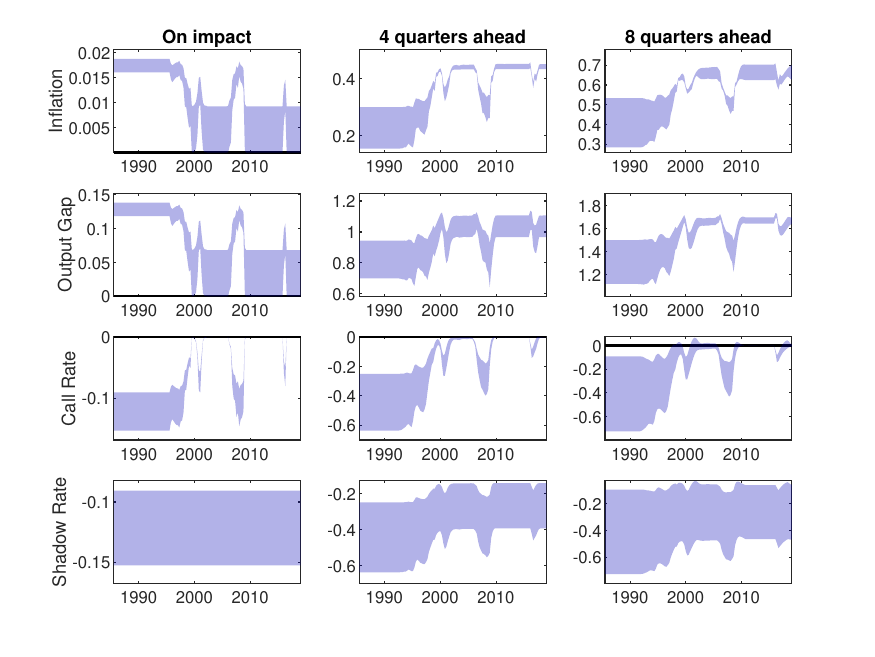}
\vspace{.1cm}
\fnote{\footnotesize {Note: Identified sets of (cumulative) impulse responses to a -25bps monetary policy shock at three different horizons, on impact, one year out and two years out, estimated using a CKSVAR(2) model in inflation, output gap, the call rate, and the shadow rate for Japan over the period 1985q3-2019q1, identified by the sign restrictions that the shock has nonnegative effects on inflation and output and nonpositive effects on the call rate and the shadow rate up to four quarters.}}
\end{figure}

To shed further light on the differences in the responses over time, and in light of the fact that the IRFs are time-varying, we look at the evolution of the impulse responses at given horizons, 0, 4 and 8 quarters, over time. The results are reported in Figures \ref{fig: set_irf_us_cum} and \ref{fig: set_irf_jp_cum} for the U.S. and Japan, respectively. In each figure, the panels on the left column report the impact effects of a -25 basis point monetary policy shock at each quarter from 1985q3 to the end of our sample. The panels in the middle column show the cumulative impulse responses after a year, while the panels on the right column give the corresponding cumulative responses after two years. 

Before discussing each country, it is worth noting that the impact effect of the shock on shadow rates is invariant and non-zero over time in both countries. This is because the shadow rates are unrestricted across both ELB and non-ELB regimes. By contrast, the responses of short rates vary according to the probability of entering the ELB regime over time. When the economy is far above the ELB, the effects on short rates are identical to those on shadow rates. As the short rates approach the ELB, the probability that the ELB constraint binds increases, and the average response of short rates diminishes.

Now we discuss each country in turn.
In the U.S. (Figure \ref{fig: set_irf_us_cum}), we see a drop in the impact effect of policy during the ELB period relative to the pre-ELB period. The relative difference in the effectiveness of policy on both inflation and output on impact is over 15 percent. For the output gap, this difference remains, and increases one and two years ahead. However, the effect on inflation is the reverse: the cumulative effect of UMP on inflation is much stronger one and two years ahead. Therefore, UMP in the U.S. seems to have had a delayed but strong effect on inflation, but has been persistently less effective on output than conventional policy.

In Japan (Figure \ref{fig: set_irf_jp_cum}), there is a clear drop in the contemporaneous effect of policy on inflation and output during the ELB periods. There are three distinguishable ELB periods, 1999q2-2000q2, 2001q2-2006q2, and 2009q1 to the end of the sample.  The contemporaneous policy effect on inflation is negligible, but the delayed effect one and two years later is stronger during the ELB periods than outside them. Like for the U.S., there is a stronger delayed effect of UMP on inflation in Japan. Turning to the policy effect on output, we see that UMP is more than 50 percent weaker on impact, but catches up within one year, and stays stronger two years out. So, unlike the U.S., where UMP has smaller effects on output at all horizons, in Japan, this is not the case.\footnote{The estimation results for Japan are based on a second-order VAR selected by the AIC and the sequential LR tests. Because these criteria are known to lead to overfitting and our sample for Japan is relatively small (34 years), we investigated the robustness of the above results to a first-order VAR, and found that our conclusions continue to hold. The results are available through our replication code.}

In sum, in the U.S. the response of inflation to UMP has been stronger than to conventional policy one or two years ahead, while it has been weaker for output over the same horizon. In Japan, UMP has had weaker effects than conventional policy on inflation and output on impact, but has had stronger delayed effects on these variables. In general, the findings show that the effects of UMP have been different across time, lending further support to our earlier finding that the ELB has been empirically relevant in both countries over that period.\footnote{\OA \ref{app:shadow_rates} reports estimates of the shadow rate for Japan and the U.S. that can be interpreted as indicators of the desired monetary policy stance in terms of the short-term interest rate during ELB regimes.}

\subsection{Discussion}\label{sec:discussion}

Our analysis of the Japanese data draws heavily on the seminal contribution of \cite{HayashiKoeda2019}. Specifically, the use of trend growth that they proposed to control for the decline in the short rate over our sample is essential to get a VAR that satisfies the sign restrictions on the IRFs. There are several apparent methodological differences between our papers. \cite{HayashiKoeda2019} use monthly data over a shorter period 1992-2012, while we use quarterly data from 1985 to 2019. They model QE via excess reserves and FG via an exit condition on inflation, while we rely on the shadow rate to capture both forms of UMP and other forms of UMP such as the purchases of long-term government bonds, motivated by the theoretical model of Section \ref{s: dsge}. They use recursive identification as well as exclusion restrictions on the dynamics of the policy reaction function, assuming that inflation and output are predetermined, while we do not, and rely instead on sign restrictions and the changes in the dynamics and variances across regimes for identification. However, these differences are not as large as they appear. For example, their results show that inflation and output are not predetermined at quarterly frequency, which is consistent with our findings, since the identified impact effect of monetary policy easing on inflation and output is positive within the quarter. Both models have regime-dependent decision rules that are fairly similar when translated to quarterly frequency. Finally, even though they provide convincing documentation that an inflation exit condition fits better the narrative of Japanese monetary policy over their sample period, the use of a \cite{ReifschneiderWilliams2000} FG rule appears nearly observationally equivalent to an inflation exit condition because of the relative scarcity of movements in and out of the ELB regime over the sample. This explains why our conclusions are broadly consistent with theirs.\footnote{For example, the rejection of the irrelevance hypothesis (\ref{eq: IH}) for Japan is due to both $\widetilde{\beta}\neq 0$ and $C^*_{12}\neq 0$. This accords with \cite{HayashiKoeda2019}, who report significant changes both in the constant term as well as the coefficient on the lag of the policy variable across regimes.}

Finally, our theoretical model abstracts from possible negative effects of UMP such as those of `the reversal interest rate' \citep{BrunnermeierKoby19}, and our empirical analysis excludes those effects through sign restrictions. If we remove the sign restrictions during ELB regimes and maintain them during non-ELB regimes, we can allow $\xi$ to be negative, which can capture policy reversals on impact. It turns out that removing the sign restrictions during the ELB periods does not affect the identified set for $\xi$ for the impulse responses for the U.S., while the effect for Japan is limited.\footnote{The identified set for $\xi$ for Japan becomes $[-0.08,0.29]$, which includes negative responses of inflation and output to a decrease in the shadow rate under the ELB (an expansionary UMP shock) on impact. However, the delayed responses are not significantly affected, and thus any possible negative effects of UMP are short-lived.}

\section{Conclusion\label{s: conclusion}}

The paper develops theoretical and empirical models to study the effectiveness
of unconventional monetary policy in the U.S. and Japan. The theoretical model
allows the degree of effectiveness of unconventional policy to range from
being as effective as conventional policy to being completely ineffective, and it provides theoretical underpinnings to the empirical model and our approach to testing the irrelevance hypothesis of the ELB. The empirical analysis is based on
an agnostic structural VAR model that accounts for the ELB on the policy rate and captures unconventional policy via the shadow rate. Our results provide strong evidence
against the hypothesis that the ELB is empirically\ irrelevant, and support the view that the ELB has been an important constraint on monetary policy in both the U.S. and Japan. However, our results also reveal strong delayed effects and country-specific differences in the effect of unconventional policy relative to conventional policy.

% Remove or comment out the next two lines if you are not using bibtex.
\bibliographystyle{aea}
\bibliography{sophAll}

{\normalsize \appendix\newpage}
\begin{flushleft}
\Large{\textbf{Appendix}}
\end{flushleft}

\footnotesize
\titleformat*{\section}{\normalsize\bfseries}
\titleformat*{\subsection}{\footnotesize\bfseries}
\titleformat*{\subsubsection}{\footnotesize\bfseries}
\titleformat*{\paragraph}{\footnotesize\bfseries}

\numberwithin{equation}{section}
\setcounter{equation}{0}

\input{appendix}

\end{document}

%% file: appendix.tex
\appendix
\renewcommand{\baselinestretch}{1}

\section{Theoretical Model\label{A: theoretical model}}

Appendix \ref{A: theoretical model} presents a simple New Keynesian model with an effective lower bound (ELB) and unconventional monetary policy (UMP). The model is a
version of a preferred habitat model such as \cite{ChenCurdiaFerrero2012}, extended to incorporate two things: a policy rule for quantitative easing (QE) that is operated using a shadow rate as policy guidance and forward guidance (FG) in the spirit of \cite{ReifschneiderWilliams2000}. To keep the analysis
focused on the salient features of the transmission mechanisms of UMP, the
model abstracts from capital accumulation and consumption habit formation. There are three shocks: a demand (preference) shock, a
supply (productivity) shock, and a monetary policy shock.

\subsection{Model building blocks}\label{A: equilibrium conditions}

\subsubsection{Long-term bonds\label{app_long_bonds}}

There is a long-term government bond (consol bond). The long-term bond issued at time $t$ yields $\mu^{j-1}$ dollars at time $t+j$
over time. Let $R_{L,t+1}$ denote the gross nominal rate from time $t$ to
$t+1$. The period-$t$ price of the bond issued at time $t$, $P_{L,t}$, is
defined as
\begin{align}
P_{L,t}  &  =E_{t}\left(  \frac{1}{R_{L,t+1}}+\frac{\mu}{R_{L,t+1}%
R_{L,t+2}}+\frac{\mu^{2}}{R_{L,t+1}R_{L,t+2}R_{L,t+3}}+...\right)
\nonumber\\
&  =E_{t}\left(  \frac{1}{R_{L,t+1}}+\frac{\mu}{R_{L,t+1}}P_{L,t+1}\right)
. \label{Pl}%
\end{align}
The gross yield to maturity (or the long-term interest rate) at time $t$,
$\bar{R}_{L,t}$ is defined as%
\[
E_{t}\left(  \frac{1}{\bar{R}_{L,t}}+\frac{\mu}{\left(  \bar{R}%
_{L,t}\right)  ^{2}}+\frac{\mu^{2}}{\left(  \bar{R}_{L,t}\right)  ^{3}%
}+...\right)  =P_{L,t},
\]
or%
\begin{equation}
P_{L,t}=\frac{1}{\bar{R}_{L,t}-\mu}. \label{Rlbar}%
\end{equation}

Let $B_{L,t|t-s}$ denote period-$t$ holdings of bonds that were issued at time $t-s$.
Suppose that a household has $B_{L,t|t-s}$ for $s=1,2,...$ in the beginning
of period $t$. The total amount of dividends the household receives in period
$t$ is%
\[
\sum_{s=1}^{\infty}\mu^{s-1}B_{L,t|t-s}.
\]
Note that having one unit of $B_{L,t|t-s}$ is equivalent to having
$\mu^{s-1}$ units of $B_{L,t|t-1}$ because they both yield $\mu^{s-1}$
dollars. The total amount of dividends then can be expressed in terms of
$B_{L,t|t-1}$ as%
\[
\sum_{s=1}^{\infty}\mu^{s-1}B_{L,t|t-s}\equiv B_{L,t-1},
\]
where $B_{L,t-1}$ denotes the amount of bonds in units of the bonds issued at
time $t-1$, held by the household in the beginning of period $t$. Let
$P_{L,t|t-s}$ denote the time-$t$ price of the bond issued at time $t-s$.
Then, the value of all bonds at time $t$ is%
\[
\sum_{s=1}^{\infty}P_{L,t|t-s}B_{L,t|t-s}.%
\]
The bond price satisfies
\begin{align*}
P_{L,t|t-s}  &  =E_{t}\left(  \frac{\mu^{s}}{R_{L,t+1}}+\frac{\mu^{s+1}%
}{R_{L,t+1}R_{L,t+2}}+\frac{\mu^{s+2}}{R_{L,t+1}R_{L,t+2}R_{L,t+3}%
}+...\right) \\
&  =\mu^{s}P_{L,t}%
\end{align*}
Then the value of all bonds at time $t$ is%
\[
\sum_{s=1}^{\infty}P_{L,t|t-s}B_{L,t|t-s}=P_{L,t}\mu\sum_{s=1}^{\infty
}\mu^{s-1}B_{L,t|t-s}=\mu P_{L,t}B_{L,t-1}.
\]
So the return of holding $B_{L,t-1}$ is given by the sum of dividends and the
value of all bonds as:%
\[
B_{L,t-1}+\mu P_{L,t}B_{L,t-1}=\left(  1+\mu P_{L,t}\right)
B_{L,t-1}=P_{L,t}\bar{R}_{L,t}B_{L,t-1}=\frac{\bar{R}_{L,t}}{\bar{R}%
_{L,t}-\mu}B_{L,t-1}.
\]

\subsubsection{Households}\label{a: households}

There are two types of households: unrestricted households (U-households) and
restricted households (R-households). U-households, with population
$\omega_{u}$, can trade both short-term and long-term government bonds subject
to a transaction cost $\zeta_{t}$ per unit of long-term bonds purchased.
R-households, with population $\omega_{r}=1-\omega_{u}$, can trade only
long-term government bonds. For $j=u,r$, each household chooses consumption
$c_{t}^{j}$, hours worked $h_{t}^{j}$, the long-term government bond holdings $B_{L,t}^{j}$,
and the short-term government bond holdings $B_{S,t}^{j}$ to maximize utility,%
\[
\sum_{t=0}^{\infty}\beta_{j}^{t}d_{t}\left[  \frac{\left(  c_{t}^{j}\right)
^{1-\sigma}}{1-\sigma}-\psi\frac{\left(  h_{t}^{j}\right)  ^{1+1/\nu}}%
{1+1/\nu}\right]  ,
\]
subject to: for a U-household,%
\[
P_{t}c_{t}^{u}+B_{S,t}^{u}+\left(  1+\zeta_{t}\right)  P_{L,t}B_{L,t}%
^{u}=\left(  1+i_{t-1}\right)  B_{t-1}^{u}+P_{L,t}\bar{R}_{L,t}B_{L,t-1}%
^{u}+W_{t}h_{t}^{u}-T_{t}^{u}+\Pi_{t}^{u},
\]
and for a R-household,%
\[
P_{t}c_{t}^{r}+P_{L,t}B_{L,t}^{r}=P_{L,t}\bar{R}_{L,t}B_{L,t-1}^{r}+W_{t}%
h_{t}^{r}-T_{t}^{r}+\Pi_{t}^{r},
\]
where $P_{t}$ is the price level and $i_{t}$ is the short-term interest rate.
In addition $\bar{R}_{L,t}$ denotes the gross yield to maturity at time $t$ on
the long-term bond%
\[
\bar{R}_{L,t}=\frac{1}{P_{L,t}}+\mu,\text{ \ \ }0<\mu\leq1\text{.}%
\]
The average duration of the bond is given by $\bar{R}_{L,t}/\left(  \bar
{R}_{L,t}-\mu\right)  $. \ There is a shock $d_{t}$ to the preference, and
it is given by:%
\[
d_{t}=\left\{
\begin{array}
[c]{c}%
e^{z_{1}^{b}}e^{z_{2}^{b}}...e^{z_{t}^{b}}\\
1
\end{array}
\right.
\begin{array}
[c]{c}%
\text{for }t\geq1\\
\text{for }t=0
\end{array}
,
\]
where $z_{t}^{b}$ is a preference (demand) shock, which is assumed to follow
an AR(1) process%
\[
z_{t}^{b}=\rho_{b}z_{t-1}^{b}+\epsilon_{t}^{b},
\]
with $\epsilon_{t}^{b}$ $\sim$ i.i.d. $N\left(  0,\sigma_{b}^{2}\right)  $.

We assume that the transaction cost of trading long-term bonds for the
U-households is collected by financial firms and redistributed as a lump-sum
profits to the U-households. Under the assumption, the transaction cost does
not appear in the goods market clearing condition, which is given by:%
\begin{equation}
y_{t}=\omega_{u}c_{t}^{u}+\left(  1-\omega_{u}\right)  c_{t}^{r}.\label{goods}%
\end{equation}

Arranging the first-order conditions of the U-household's problem yields the
following optimality conditions:%
\begin{align}
w_{t}  &  =\psi\left(  c_{t}^{u}\right)  ^{\sigma}\left(  h_{t}^{u}\right)
^{1/\nu},\label{HHu1}\\
1  &  =E_{t}\beta_{u}e^{z_{t+1}^{b}}\left(  \frac{c_{t+1}^{u}}{c_{t}^{u}%
}\right)  ^{-\sigma}\frac{1+i_{t}}{\pi_{t+1}},\label{HHu2}\\
1+\zeta_{t}  &  =E_{t}\beta_{u}e^{z_{t+1}^{b}}\left(  \frac{c_{t+1}^{u}}%
{c_{t}^{u}}\right)  ^{-\sigma}\frac{R_{L,t+1}}{\pi_{t+1}}, \label{HHu3}%
\end{align}
where $w_{t}\equiv W_{t}/P_{t}$ denotes the real wage, $\pi_{t}\equiv
P_{t}/P_{t-1}$ denotes the inflation rate, and $R_{L,t+1}$ denotes the
yield of the long-term bond between periods $t$ and $t+1$, given by
\[
R_{L,t+1}\equiv\frac{P_{L,t+1}}{P_{L,t}}\bar{R}_{L,t+1}=\frac{P_{L,t+1}%
}{P_{L,t}}\left(  \frac{1}{P_{L,t+1}}+\mu\right)  =\frac{1+\mu
P_{L,t+1}}{P_{L,t}}.
\]
Similarly, arranging the first-order conditions of the R-household's problem
yields%
\begin{align}
w_{t}  &  =\psi\left(  c_{t}^{r}\right)  ^{\sigma}\left(  h_{t}^{r}\right)
^{1/\nu},\label{HHr1}\\
1  &  =E_{t}\beta_{r} e^{z_{t+1}^{b}}\left(  \frac{c_{t+1}^{r}}{c_{t}^{r}}\right)
^{-\sigma}\frac{R_{L,t+1}}{\pi_{t+1}}. \label{HHr2}%
\end{align}

\subsubsection{Firms}

The firm sector consists of two types of firms: final-goods-producing firms and
intermediate-goods-producing firms. The problems of these firms are standard except that
the average discount rate between U-households and R-households is used in
discounting the profits of these firms. The profits need to be derived
explicitly because one of the two households' budget constraints constitutes
an equilibrium condition as well as the goods market clearing condition.

Each final-goods-producing firm produces a unit of final goods $y_{t}$ in a competitive market by combining intermediate goods $\left\{  y_{t}\left(
l\right)  \right\}  _{l=0}^{1}$ according
to%
\[
y_{t}=\left[  \int_{0}^{1}y_{t}\left(  l\right)  ^{\frac{1}{\lambda_{p}}%
}dl\right]  ^{\lambda_{p}},\text{ }\lambda_{p}>1.
\]
The demand function for the $l$-th intermediate good is given by%
\[
y_{t}\left(  l\right)  =\left(  \frac{P_{t}\left(  l\right)  }{P_{t}}\right)
^{\frac{\lambda_{p}}{1-\lambda_{p}}}y_{t}.
\]

Each intermediate-goods-producing firm uses labor and produce intermediate goods according to%
\[
y_{t}\left(  l\right)  =e^{z_{t}^{a}}h_{t}\left(  l\right)  ^{\theta},\text{
}0<\theta\leq1.
\]
where \thinspace$z_{t}^{a}$ is a productivity shock, which is assumed to
follow%
\[
z_{t}^{a}=\rho_{a}z_{t-1}^{a}+\epsilon_{t}^{a},
\]
with $\epsilon_{t}^{a}\sim$ i.i.d. $N\left(  0,\sigma_{a}^{2}\right)  $.
Because there is no price dispersion in steady state, the aggregate output can be expressed up to the first-order
approximation as:%
\begin{equation}
\hat{y}_{t}=z_{t}^{a} + \theta\hat{h}_{t}, \label{Y}%
\end{equation}
where $\hat{y}_{t}$ and $\hat{h}_{t}$ denote the aggregate output and hours worked in terms of deviation from the steady state. The
total cost of producing $y_{t}\left(  l\right)  $ is equal to
\[
W_{t}h_{t}\left(  l\right)  =W_{t}\left(  \frac{y_{t}\left(  l\right)
}{e^{z_{t}^{a}}}\right)  ^{\frac{1}{\theta}}.
\]

In each period, intermediate-goods-producing firms can change their price with
probability $\xi$ identically and independently across firms and over time.
For each $l$, the $l$-th intermediate-goods producing firm chooses the price, $\tilde
{P}_{t}\left(  l\right)  $, to maximize the discounted sum of profits,%
\[
\max_{\tilde{P}_{t}\left(  l\right)  }E_{t}\sum_{s=0}^{\infty}\left(
\xi\delta\right)  ^{s}\bar{\Lambda}_{t+s|t}\left[  P_{t+s}\left(  l\right)
y_{t+s}\left(  l\right)  -W_{t+s}\left(  \frac{y_{t+s}\left(  l\right)
}{e^{z_{t+s}^{a}}}\right)  ^{\frac{1}{\theta}}\right]  ,
\]
subject to the demand curve,%
\[
y_{t+s}(l)=\left(  \frac{P_{t+s}\left(  l\right)  }{P_{t+s}}\right)
^{\frac{\lambda_{p}}{1-\lambda_{p}}}y_{t+s},
\]
where%
\begin{align*}
& \delta = \omega_{u}\beta_{u}+(1-\omega_{u})\beta_{r}, \\
&  \bar{\Lambda}_{t+s|t}\equiv d_{t+s|t}\left( \omega_{u}\Lambda_{t+s|t}^{u}+\left(
1-\omega_{u}\right)  \Lambda_{t+s|t}^{r}\right),\\
&  \Lambda_{t+s|t}^{j}=\left(\frac{c_{t+s}^{j}}{c_{t}%
^{j}}\right)  ^{-\sigma}\frac{1}{P_{t+s}},\text{ \ }d_{t+s|t}=\begin{cases} 1 & \text{if $s=0$} \\ e^{z_{t+1}^{b}%
}e^{z_{t+2}^{b}}...e^{z_{t+s}^{b}} & \text{if $s=1,2,...$}
\end{cases}\\
&  P_{t+s}\left(  l\right)  =\tilde{P}_{t}(l)\Pi_{t,t+s}^{p},\\
&  \Pi_{t+s|t}^{p}=%
\begin{cases}
1 & \text{if $s=0$}\\
\prod_{k=1}^{s}(\pi_{t+k-1})^{\iota_{p}}\left(  \pi\right)  ^{1-\iota_{p}} &
\text{if $s=1,2,...$}%
\end{cases}
\end{align*}
The presence of $\Pi_{t+s|t}^{p}$ implies price indexation for firms that do not have a chance to change prices and $0\leq \iota_{p} \leq 1$ governs the degree of indexation to the past inflation rates.
Substituting the demand curve into the objective function yields%
\[
\max_{\tilde{P}_{t}\left(  l\right)  }E_{t}\sum_{s=0}^{\infty}\left(
\xi\delta\right)  ^{s}\bar{\Lambda}_{t+s|t}\left[  \tilde{P}_{t}(l)\Pi_{t+s|t}%
^{p}\left(  \frac{\tilde{P}_{t}(l)\Pi_{t+s|t}^{p}}{P_{t+s}}\right)
^{\frac{\lambda_{p}}{1-\lambda_{p}}}Y_{t+s}-W_{t+s}\left(  \frac{\tilde{P}%
_{t}(l)\Pi_{t+s|t}^{p}}{P_{t+s}}\right)  ^{\frac{\lambda_{p}}{\left(
1-\lambda_{p}\right)  \theta}}\left(  \frac{y_{t+s}}{e^{z_{t+s}^{a}}}\right)
^{\frac{1}{\theta}}\right]  .
\]
The first-order condition is%
\[
0=E_{t}\sum_{s=0}^{\infty}\left(  \xi\delta\right)  ^{s}\bar{\Lambda}_{t+s|t}\left[
\frac{1}{1-\lambda_{p}}\Pi_{t+s|t}^{p}y_{t+s}\left(  l\right)  -W_{t+s}%
\frac{\lambda_{p}}{\left(  1-\lambda_{p}\right)  \theta}\left(  \frac
{y_{t+s}\left(  l\right)  }{e^{z_{t+s}^{a}}}\right)  ^{\frac{1}{\theta}}%
\frac{1}{\tilde{P}_{t}(l)}\right]  .
\]
Since $\tilde{P}_{t}\left(  l\right)  $ does not depend on $l$, index $l$ is
omitted hereafter. Define $\tilde{p}_{t}\equiv\tilde{P}_{t}/P_{t}$ and
\[
\tilde{\Pi}_{t+s|t}^{p}=%
\begin{cases}
1 & \text{if $s=0$}\\
\prod_{k=1}^{s}\frac{(\pi_{t+k-1})^{\iota_{p}}\left(  \pi\right)
^{1-\iota_{p}}}{\pi_{t+k}} & \text{if $s=1,2,...$}%
\end{cases}
\]
The first-order condition can be transformed as
\begin{align*}
0  &  =E_{t}\sum_{s=0}^{\infty}\left(  \xi\delta\right)  ^{s}\bar{\Lambda}%
_{t+s}P_{t+s}\left[  \frac{1}{1-\lambda_{p}}\frac{\Pi_{t+s|t}^{p}}{P_{t+s}%
}\left(  \tilde{p}_{t}\tilde{\Pi}_{t+s|t}^{p}\right)  ^{\frac{\lambda_{p}%
}{1-\lambda_{p}}}y_{t+s}\right. \\
&  \left.  -\frac{W_{t+s}}{P_{t+s}}\frac{\lambda_{p}}{\left(  1-\lambda
_{p}\right)  \theta}\left(  \tilde{p}_{t}\tilde{\Pi}_{t+s|t}^{p}\right)
^{\frac{\lambda_{p}}{\left(  1-\lambda_{p}\right)  \theta}}\left(
\frac{Y_{t+s}}{e^{z_{t+s}^{a}}}\right)  ^{\frac{1}{\theta}}\frac{1}{\tilde
{P}_{t}}\right]  ,
\end{align*}
This equation can be written as:%
\begin{equation}
\tilde{p}_{t}=\left(  \frac{\lambda_{p}}{\theta}\frac{\omega_{u}K_{p,t}%
^{u}+\left(  1-\omega_{u}\right)  K_{p,t}^{r}}{\omega_{u}F_{p,t}^{u}+\left(
1-\omega_{u}\right)  F_{p,t}^{r}}\right)  ^{\frac{\left(  1-\lambda
_{p}\right)  \theta}{\theta-\lambda_{p}}}, \label{pt}%
\end{equation}
where for $j \in \{r, u\}$
\begin{align}
&  F_{p,t}^{j}=(  c_{t}^{j})  ^{-\sigma}y_{t}+\xi\delta%
E_{t}e^{z_{t+1}^{b}}(\tilde{\Pi}_{t+1|t}^{p})^{\frac{1}{1-\lambda_{p}}%
}F_{p,t+1}^{j},\label{Fp}\\
&  K_{p,t}^{j}=(  c_{t}^{j})  ^{-\sigma}\left(  \frac{y_{t}%
}{e^{z_{t}^{a}}}\right)  ^{\frac{1}{\theta}}w_{t}+\xi\delta%
E_{t}e^{z_{t+1}^{b}}(\tilde{\Pi}_{t+1|t}^{p})^{\frac{\lambda_{p}}{\left(
1-\lambda_{p}\right)  \theta}}K_{p,t+1}^{j}. \label{Kp}%
\end{align}
The aggregate price level evolves following
\[
P_{t}=\left[  \xi[(\pi_{t-1})^{\iota_{p}}\left(  \pi\right)  ^{1-\iota
_{p}}P_{t-1}]^{\frac{1}{1-\lambda_{p}}}+(1-\xi)\tilde{P}_{t}^{\frac
{1}{1-\lambda_{p}}}\right]  ^{1-\lambda_{p}},
\]
which can be written as
\begin{equation}
\tilde{p}_{t}=\left[  \frac{1-\xi(\tilde{\Pi}_{t|t-1}^{p})^{\frac
{1}{1-\lambda_{p}}}}{1-\xi}\right]  ^{1-\lambda_{p}}. \label{Ps}%
\end{equation}
The conditions, (\ref{pt})-(\ref{Ps}), summarize the price setting behavior of
intermediate-goods-producing firms.

The aggregate nominal profits earned by intermediate-goods-producing firms are given by:%
\[
\Pi_{t}^{m}=\int_{0}^{1}\left(  P_{t}\left(  l\right)  y_{t}\left(  l\right)
-W_{t}\left(  \frac{y_{t}\left(  l\right)  }{e^{z_{t}^{a}}}\right)  ^{\frac
{1}{\theta}}\right)  dl = P_{t}y_{t}-W_{t}\left(  \frac{y_{t}}%
{e^{z_{t}^{a}}}\right)  ^{\frac{1}{\theta}},
\]
where the last equality holds up to the first-order approximation.
Then, the aggregate real profits are given by $\pi_{t}^{m}=y_{t}-w_{t}\left(
y_{t}/e^{z_{t}^{a}}\right)  ^{1/\theta}$.

\subsubsection{Government}

The government flow budget constraint is%
\[
\left(  1+i_{t-1}\right)  B_{S,t-1}+\left(  1+\mu P_{L,t}\right)
B_{L,t-1}=B_{S,t}+P_{L,t}B_{L,t}+T_{t},
\]
where $T_{t}=\omega_{u}T_{t}^{u}+\left(  1-\omega_{u}\right)  T_{t}^{r}$. We
assume that the lump-sum tax is imposed on households equally so that
$T_{t}^{u}=T_{t}^{r}=T_{t}$. To focus on the role of long-term government bonds, we assume that the
amount of short-term bonds is constant at $b_{S,t}\equiv B_{S,t}/P_{t}=\bar{b}_{S}$.

\subsubsection{Central bank}
The nominal interest rate $i_{t}$ set by the central bank is bounded below by
the ELB as
\begin{equation}
i_{t}=\max\left\{  i_{t}^{\ast},\underline{i}\right\},  \label{i_b}%
\end{equation}
where $\underline{i}$ is the ELB and $i_{t}^{\ast}$ is a shadow rate -- the short-term rate the central bank would set if there were no ELB.
The shadow rate $i_{t}^{\ast}$ is given by\footnote{\cite{ReifschneiderWilliams2000} use the following rule: $i_{t}^{\ast}=i_{t}^{\text{Taylor}}-\alpha Z_{t}$ and $Z_{t}=\rho_{Z}Z_{t-1}+(  i_{t}-i_{t}^{\text{Taylor}})$ with $\rho_{Z}=1$.}
\begin{equation}
i_{t}^{\ast}=i_{t}^{\text{Taylor}} - \alpha \left(i_{t}-i_{t}^{\text{Taylor}}\right).  \label{i*_a}
\end{equation}
The shadow rate $i_{t}^{\ast}$ consists of two
parts: $i_{t}^{\text{Taylor}}$ and $\alpha (i_{t}-i_{t}^{\text{Taylor}})$. First, $i_{t}%
^{\text{Taylor}}$ is the Taylor-rule-based rate that responds to inflation
$\pi_{t}$, output $y_{t}$, and the lagged `effective' interest rate $(1-\lambda^{\ast})i_{t-1} + \lambda^{\ast}i^{\ast}_{t-1}$:%
\begin{equation}
i_{t}^{\text{Taylor}}-i=\rho_{i}\left( (1-\lambda^{\ast})i_{t-1} + \lambda^{\ast}i^{\ast}_{t-1} - i\right)  +\left(
1-\rho_{i}\right)  \left[  r_{\pi}\log\left(  \pi_{t}/\pi\right)  +r_{y}%
\log\left(  y_{t}/y\right)  \right]  +\epsilon_{t}^{i}, \label{iTaylor}%
\end{equation}
where $\epsilon_{t}^{i}$ is a
monetary policy shock and variables without subscripts denote those in steady state. The parameter $\lambda^{\ast}$ will be derived later in this appendix.
Second, $\alpha (i_{t}-i_{t}^{\text{Taylor}})$ in equation (\ref{i*_a}) encapsulates the strength of FG.
A positive value for $\alpha$ will maintain the target rate $i_{t}^{\ast}$ below the Taylor rate $i_{t}^{\text{Taylor}}$. Under the ELB of
$i_{t}=\underline{i}$, the more the central bank has missed to set the interest rate at its Taylor rate, the lower the central bank sets its target rate
$i_{t}^{\ast}$ through equation (\ref{i*_a}) as long as $\rho_{i}\lambda^{\ast}>0$ in equation (\ref{iTaylor}).\footnote{\cite{DebortoliGaliGambetti2019} consider the case of $\alpha=0$ and $\lambda^{\ast}=1$ in equation (\ref{iTaylor}) and interpret $\rho_{i}$ -- the coefficient of interest rate smoothing -- as FG when $i_{t}^{\ast}$ is below the ELB.}

The central bank activates QE when the economy hits the ELB. The central bank continues using the shadow rate as policy guidance in an ELB regime as in a non-ELB regime. Specifically, the amount of long-term bond purchases depends on the shadow rate, and as a result the amount of long-term government bonds, $b_{L,t}\equiv B_{L,t}/P_{t}$, held by the
private agents is given by:%
\begin{equation}
\hat{b}_{L,t}=\left\{
\begin{array}
[c]{c}%
0\\
\gamma\frac{i_{t}^{\ast}-\underline{i}}{1+i}%
\end{array}
\right.
\begin{array}
[c]{c}%
\text{if }i_{t}^{\ast}\geq \underline{i}\\
\text{if }i_{t}^{\ast}< \underline{i}
\end{array}
,\label{QE}%
\end{equation}
where the caret on a variable denotes a deviation from the steady state. This QE rule implies that asset purchases by the central bank is zero (relative to the
steady state) when the ELB is not binding (i.e. $i_{t}=i_{t}^{\ast}\geq \underline{i}$) and, given $\gamma>0$, the purchases are positive (i.e.,
$\hat{b}_{L,t}<0$) when the shadow rate goes below the ELB (i.e.,
\thinspace$i_{t}^{\ast}<\underline{i}$).

\subsubsection{Market clearing and equilibrium}

As well as the goods market clearing condition (\ref{goods}), there are market
clearing conditions for labor, long-term government bonds, and short-term
government bonds:%
\begin{align}
\omega_{u}h_{t}^{u}+\left(  1-\omega_{u}\right)  h_{t}^{r}  &  =h_{t}%
,\label{labor}\\
\omega_{u}b_{L,t}^{u}+\left(  1-\omega_{u}\right)  b_{L,t}^{r}  &
=b_{L,t},\label{lbonds}\\
\omega_{u}b_{S,t}^{u}  &  =b_{S,t} \label{sbonds}%
\end{align}
Also, either the U-household's budget constraint or the R-household's budget
constraint should be added as an equilibrium condition. Here the latter budget
constraint is added:%
\begin{equation}
c_{t}^{r}+P_{L,t}b_{L,t}^{r}=\left(  \bar{R}_{L,t}/\pi_{t}\right)
P_{L,t}b_{L,t-1}^{r}+w_{t}h_{t}^{r}-T_{t}^{r}/P_{t}+\Pi_{t}^{r}/P_{t},
\label{rbudget}%
\end{equation}
where%
\begin{align*}
\frac{T_{t}^{r}}{P_{t}}  &  =-\left(  b_{S,t}+P_{L,t}b_{L,t}\right)
+\frac{1+i_{t-1}}{\pi_{t}}b_{S,t-1}+\frac{1+\mu P_{L,t}}{\pi_{t}}b_{L,t-1},\\
\frac{\Pi_{t}^{r}}{P_{t}}  &  =y_{t}-w_{t}h_{t}.
\end{align*}

The cost of trading long-term bonds, $\zeta_{t}$, is specified as%
\begin{equation}
\frac{\zeta_{t}}{\zeta}= \left(  \frac{b_{L,t}}{b_{L}}\right)  ^{\rho_{\zeta}}, \hspace{1cm} \rho_{\zeta}>0. \label{tcost}
\end{equation}
The trading cost is increasing in the amount of
long-term bonds relative to its steady state value. The trading cost is $\zeta$ in steady state.

The system of equations for the economy consists of 19 equations,
(\ref{goods})-(\ref{rbudget}), with the following endogenous variables:%
\[
c_{t}^{u},c_{t}^{r},h_{t}^{u},h_{t}^{r},h_{t},b_{L,t}^{u},b_{L,t}^{r}%
,b_{L,t},b_{t}^{u},y_{t},w_{t},i_{t},i_{t}^{\ast},i_{t}^{\text{Taylor}},R_{L,t}%
,\pi_{t},\tilde{p}_{t},F_{p,t}^{j},K_{p,t}^{j}.
\]

\subsection{Log-linearized equations}
\label{A: log-linear}

We log-linearize the equilibrium conditions of the theoretical model presented in Appendix \ref{A: equilibrium conditions} around the steady state in which inflation is equal to the target rate of inflation set by the central bank. By doing so, we derive key equations in the system of equations (\ref{i})-(\ref{Phillips}) presented in Section \ref{s: dsge} of the main text. We also derive a log-linearized equation for the long-term yield.

\paragraph{Euler equation.} Log-linearizing equations
(\ref{goods}), (\ref{HHu2}), (\ref{HHu3}), (\ref{HHr2}),  and (\ref{tcost}), we obtain\footnote{The variable $\hat{i}$ represents the deviation of the gross interest rate from the steady state.}%
\begin{align}
\hat{y}_{t} &  =\frac{\omega_{u}c^{u}}{y}\hat{c}_{t}^{u}+\frac{\left(
1-\omega_{u}\right)  c^{r}}{y}\hat{c}_{t}^{r}.\label{l_goods} \\
0 &  =E_{t}\left[  -\sigma\left(  \hat{c}_{t+1}^{u}-\hat{c}_{t}^{u}\right)
+\hat{\imath}_{t}-\hat{\pi}_{t+1}+z_{t+1}^{b}\right]  ,\label{l_HHu2}\\
\frac{\zeta}{1+\zeta}\hat{\zeta}_{t} &  =E_{t}\left[  -\sigma\left(  \hat
{c}_{t+1}^{u}-\hat{c}_{t}^{u}\right)  +\hat{R}_{L,t+1}-\hat{\pi}_{t+1}%
+z_{t+1}^{b}\right]  ,\label{l_HHu3}\\
0 &  =E_{t}\left[  -\sigma\left(  \hat{c}_{t+1}^{r}-\hat{c}_{t}^{r}\right)
+\hat{R}_{L,t+1}-\hat{\pi}_{t+1}+z_{t+1}^{b}\right]  ,\label{l_HHr2}\\
\hat{\zeta}_{t} & = \rho_{\zeta}\hat{b}_{L,t}. \label{l_tcost}
\end{align}
Equation (\ref{l_goods}) can be written as:%
\[
\hat{c}_{t}^{u}=\frac{y}{\omega_{u}c^{u}}\left\{  \hat{y}_{t}-\frac{\left(
1-\omega_{u}\right)  c^{r}}{y}\hat{c}_{t}^{r}\right\}  .
\]
Subtracting $\hat{c}_{t+1}^{u}$ from $\hat{c}_{t}^{u}$ yields:%
\begin{align}
\hat{c}_{t+1}^{u}-\hat{c}_{t}^{u} &  =\frac{y}{\omega_{u}c^{u}}\left\{
\hat{y}_{t+1}-\hat{y}_{t}-\frac{\left(  1-\omega_{u}\right)  c^{r}}{y}\left(
\hat{c}_{t+1}^{r}-\hat{c}_{t}^{r}\right)  \right\}  ,\nonumber\\
&  =\frac{y}{\omega_{u}c^{u}}\left\{  \hat{y}_{t+1}-\hat{y}_{t}-\frac{\left(
1-\omega_{u}\right)  c^{r}}{y}\frac{\left(  \hat{R}_{L,t+1}-\hat{\pi}%
_{t+1}+z_{t+1}^{b}\right)  }{\sigma}\right\}  ,\label{l_goods_a}%
\end{align}
where equation (\ref{l_HHr2}) was used in the second equality. Substituting
equation (\ref{l_goods_a}) into equation (\ref{l_HHu2}) yields:%
\begin{align}
0 &  =E_{t}\left[  -\sigma\left(  \hat{c}_{t+1}^{u}-\hat{c}_{t}^{u}\right)
+\hat{\imath}_{t}-\hat{\pi}_{t+1}+z_{t+1}^{b}\right]  ,\nonumber\\
&  =E_{t}\left[  -\frac{\sigma y}{\omega_{u}c^{u}}\left(  \hat{y}_{t+1}%
-\hat{y}_{t}\right)  +\frac{\sigma y}{\omega_{u}c^{u}}\frac{\left(
1-\omega_{u}\right)  c^{r}}{y}\frac{\left(  \hat{R}_{L,t+1}-\hat{\pi}%
_{t+1}+z_{t+1}^{b}\right)  }{\sigma}\right.  \nonumber\\
&  \left.  +\hat{\imath}_{t}-\hat{\pi}_{t+1}+z_{t+1}^{b}\right]  ,\nonumber
\end{align}
or, by using equation (\ref{goods}) in steady state, \begin{equation}
0  =E_{t}\left[  -\sigma\left(  \hat{y}_{t+1}-\hat{y}_{t}\right)
+\frac{\left(  1-\omega_{u}\right)  c^{r}}{y}\hat{R}_{L,t+1}+\frac{\omega
_{u}c^{u}}{y}\hat{\imath}_{t}-\hat{\pi}_{t+1}+z_{t+1}^{b}\right]
.\label{l_HHu2_a}%
\end{equation}
Equation (\ref{l_HHu2_a}) shows that the interest rate relevant to the aggregate variables such as output and inflation is the weighted sum of the return of holding the long-term bonds $\hat{R}_{L,t+1}$ and the short-term interest rate $\hat{i}_{t}$. Also, substituting equation (\ref{l_goods_a}) into equation (\ref{l_HHu3})
yields:
\begin{align*}
\frac{\zeta}{1+\zeta}\hat{\zeta}_{t} &  =E_{t}\left[  -\sigma\left(  \hat
{c}_{t+1}^{u}-\hat{c}_{t}^{u}\right)  +\hat{R}_{L,t+1}-\hat{\pi}_{t+1}\right]
\\
&  =E_{t}\left[  -\sigma\frac{y}{\omega_{u}c^{u}}\left\{  \hat{y}_{t+1}%
-\hat{y}_{t}-\frac{\left(  1-\omega_{u}\right)  c^{r}}{y}\frac{\left(  \hat
{R}_{L,t+1}-\hat{\pi}_{t+1}+z_{t+1}^{b}\right)  }{\sigma}\right\}  \right.  \\
&  \left.  +\hat{R}_{L,t+1}-\hat{\pi}_{t+1}+z_{t+1}^{b}\right]  ,\\
&  =E_{t}\left[  -\frac{\sigma y}{\omega_{u}c^{u}}\left(  \hat{y}_{t+1}%
-\hat{y}_{t}\right)  +\frac{y}{\omega_{u}c^{u}}\left(  \hat{R}_{L,t+1}%
-\hat{\pi}_{t+1}+z_{t+1}^{b}\right)  \right]  ,
\end{align*}
or, substituting out $\hat{\zeta}_{t}$ by using equation (\ref{l_tcost}) yields
\begin{align}
E_{t}\left(  \hat{R}_{L,t+1}-\hat{\pi}_{t+1}\right)   &  =\sigma E_{t}\left(
\hat{y}_{t+1}-\hat{y}_{t}\right)  -E_{t}\left(  z_{t+1}^{b}\right)
+\frac{\omega_{u}c^{u}}{y}\frac{\zeta}{1+\zeta}\rho_{\zeta}\hat{b}_{L,t}.\label{l_HHu3_a}%
\end{align}
Combining equations (\ref{l_HHu2_a}) and (\ref{l_HHu3_a}) yields:%
\begin{align*}
0 &  =E_{t}\left[  -\sigma\left(  \hat{y}_{t+1}-\hat{y}_{t}\right)
+\frac{\left(  1-\omega_{u}\right)  c^{r}}{y}\hat{R}_{L,t+1}+\frac{\omega
_{u}c^{u}}{y}\hat{\imath}_{t}-\hat{\pi}_{t+1}+z_{t+1}^{b}\right]  \\
&  =E_{t}\left[  -\sigma\left(  \hat{y}_{t+1}-\hat{y}_{t}\right)
+\frac{\left(  1-\omega_{u}\right)  c^{r}}{y}\left(  \sigma\left(  \hat
{y}_{t+1}-\hat{y}_{t}\right)  -z_{t+1}^{b}+\frac{\omega_{u}c^{u}}{y}%
\frac{\zeta}{1+\zeta}\rho_{\zeta}\hat{b}_{L,t}+\hat{\pi}_{t+1}\right)
\right.  \\
&  \left.  +\frac{\omega_{u}c^{u}}{y}\hat{\imath}_{t}-\hat{\pi}_{t+1}%
+z_{t+1}^{b}\right]  \\
&  =E_{t}\left[  -\frac{\omega_{u}c^{u}\sigma}{y}\left(  \hat{y}_{t+1}-\hat
{y}_{t}\right)  +\frac{\omega_{u}c^{u}}{y}\hat{\imath}_{t}-\frac{\omega
_{u}c^{u}}{y}\left(  \hat{\pi}_{t+1}-z_{t+1}^{b}\right)  +\frac{\left(
1-\omega_{u}\right)  c^{r}}{y}\frac{\omega_{u}c^{u}}{y}\frac{\zeta}{1+\zeta
}\rho_{\zeta}\hat{b}_{L,t}\right]  ,
\end{align*}
or
\begin{equation*}
0   =E_{t}\left[  -\sigma\left(  \hat{y}_{t+1}-\hat{y}_{t}\right)
+\hat{\imath}_{t}-\hat{\pi}_{t+1}+z_{t+1}^{b}+\frac{\left(  1-\omega
_{u}\right)  c^{r}}{y}\frac{\zeta}{1+\zeta}\rho_{\zeta}\hat{b}\right]  ,
\end{equation*}
or%
\begin{align*}
\hat{y}_{t} &  =E_{t}\hat{y}_{t+1}-\frac{1}{\sigma}\left(  \hat{\imath}%
_{t}-E_{t}\hat{\pi}_{t+1}+E_{t}z_{t+1}^{b}\right)  -\frac{1}{\sigma}%
\frac{\left(  1-\omega_{u}\right)  c^{r}}{y}\frac{\zeta}{1+\zeta}\rho_{\zeta
}\hat{b}_{L,t}\\
&  =E_{t}\hat{y}_{t+1}-\frac{1}{\sigma}\left(  \hat{\imath}_{t}-E_{t}\hat{\pi
}_{t+1}\right)  -\frac{1}{\sigma}\frac{\left(  1-\omega_{u}\right)  c^{r}}%
{y}\frac{\zeta}{1+\zeta}\rho_{\zeta}\hat{b}_{L,t}-\frac{\rho_{b}}{\sigma
}z_{t}^{b}.
\end{align*}
This equation shows that the central bank's government bond purchase -- a
decrease in $\hat{b}_{L,t}$ -- stimulates output, given $E_{t}\hat{y}_{t+1}$
and the real rate $\hat{i}_{t}-E_{t}\hat{\pi}_{t+1}$. Since $\hat{b}_{L,t}$ follows the simple rule (\ref{QE}), the equation can be written as equation
(\ref{Euler}) in the main text, which is reproduced here for convenience:
\begin{equation}
    \hat{y}_{t} = E_{t}\hat{y}_{t+1} - \frac{1}{\sigma}\left((1-\lambda^{\ast})\hat{i}_{t} + \lambda^{\ast}\hat{i}_{t}^{\ast} - E_{t}\hat{\pi}_{t+1}\right) - \chi_{b}z_{t}^{b} \label{Euler_app}
\end{equation}
where%
\begin{align}
\lambda^{\ast} &=\frac{\left(  1-\omega_{u}\right)  c^{r}}{y}\frac{\zeta}{1+\zeta}%
\rho_{\zeta}\gamma \notag \\
\chi_{b} &  =\frac{\rho_{b}}{\sigma}\label{chib}%\label{chi_b}
\end{align}

The case of $\lambda^{\ast}=1$ (and $\alpha=0$) corresponds to the fully effective UMP, which makes the ELB irrelevant. \ Such a case
can be achieved, e.g., when the central bank responds to the shadow rate
aggressively enough to satisfy%
\[
\gamma=\left[  \frac{\left(  1-\omega_{u}\right)  c^{r}}{y}\frac{\zeta
}{1+\zeta}\rho_{\zeta}\right]  ^{-1}.
\]

\paragraph{Phillips curve.}
The Phillips curve can be derived from equations (\ref{pt})-(\ref{Ps}).
Log-linearizing equation (\ref{Ps}) yields:%
\begin{equation}
\widehat{\tilde{p}}_{t}=-\frac{\xi}{1-\xi}\widehat{\tilde{\Pi}%
}_{t|t-1}^{p},\label{Ps_ll}%
\end{equation}
where%
\[
\widehat{\tilde{\Pi}}_{t|t-1}^{p}=\left(  1-\nu_{p}\right)  \hat{\pi}%
_{t-1}-\hat{\pi}_{t}.
\]
Log-linearizing equation (\ref{pt}) yields:%
\begin{align}
\frac{\theta-\lambda_{p}}{\left(  1-\lambda_{p}\right)  \theta}\widehat{\tilde
{p}}_{t} &  =\frac{\omega_{u}K_{p}^{u}}{\omega_{u}K_{p}^{u}+\left(
1-\omega_{u}\right)  K_{p}^{r}}\hat{K}_{p,t}^{u}+\frac{\left(  1-\omega
_{u}\right)  K_{p}^{r}}{\omega_{u}K_{p}^{u}+\left(  1-\omega_{u}\right)
K_{p}^{r}}\hat{K}_{p,t}^{r}\nonumber\\
&  -\frac{\omega_{u}F_{p}^{u}}{\omega_{u}F_{p}^{u}+\left(  1-\omega
_{u}\right)  F_{p}^{r}}\hat{F}_{p,t}^{u}-\frac{\left(  1-\omega_{u}\right)
F_{p}^{r}}{\omega_{u}F_{p}^{u}+\left(  1-\omega_{u}\right)  F_{p}^{r}}\hat
{F}_{p,t}^{r}.\label{pt_ll}%
\end{align}
Combining equations (\ref{Ps_ll}) and (\ref{pt_ll}) leads to:%
\begin{align}
-\frac{\xi}{1-\xi}\frac{\theta-\lambda_{p}}{\left(  1-\lambda
_{p}\right)  \theta}\left[  \left(  1-\nu_{p}\right)  \hat{\pi}_{t-1}-\hat
{\pi}_{t}\right]   &  =\frac{\omega_{u}K_{p}^{u}}{\omega_{u}K_{p}^{u}+\left(
1-\omega_{u}\right)  K_{p}^{r}}\hat{K}_{p,t}^{u}+\frac{\left(  1-\omega
_{u}\right)  K_{p}^{r}}{\omega_{u}K_{p}^{u}+\left(  1-\omega_{u}\right)
K_{p}^{r}}\hat{K}_{p,t}^{r}\nonumber\\
&  -\frac{\omega_{u}F_{p}^{u}}{\omega_{u}F_{p}^{u}+\left(  1-\omega
_{u}\right)  F_{p}^{r}}\hat{F}_{p,t}^{u}-\frac{\left(  1-\omega_{u}\right)
F_{p}^{r}}{\omega_{u}F_{p}^{u}+\left(  1-\omega_{u}\right)  F_{p}^{r}}\hat
{F}_{p,t}^{r}.\label{K_F}%
\end{align}
Log-linearizing equation (\ref{Fp}) and (\ref{Kp}) yields:%
\begin{align*}
\hat{F}_{p,t}^{j} &  =\left(  1-\xi\delta\right)  \left(  -\sigma\hat
{c}_{t}^{j}+\hat{y}_{t}\right)  +\xi\delta E_{t}\left(  z_{t+1}^{b}+\frac
{1}{1-\lambda_{p}}\widehat{\tilde{\Pi}}_{t+1|t}^{p}+\hat{F}_{p,t+1}%
^{j}\right)  ,\\
\hat{K}_{p,t}^{j} &  =\left(  1-\xi\delta\right)  \left(  -\sigma\hat
{c}_{t}^{j}+\frac{1}{\theta}\hat{y}_{t}-\frac{1}{\theta}z_{t}^{a}+\hat{w}%
_{t}\right)  +\xi\delta E_{t}\left(  z_{t+1}^{b}+\frac{\lambda_{p}}{\left(
1-\lambda_{p}\right)  \theta}\widehat{\tilde{\Pi}}_{t+1|t}^{p}+\hat{K}%
_{p,t+1}^{j}\right)  ,
\end{align*}
for $j \in \{r,u\}$. The term involving $\hat{F}_{p,t}^{u}$ and $\hat{F}%
_{p,t}^{r}$ in equation (\ref{K_F}) is calculated as follows.%
\begin{align*}
&  \frac{\omega_{u}F_{p}^{u}}{\omega_{u}F_{p}^{u}+\left(  1-\omega_{u}\right)
F_{p}^{r}}\hat{F}_{p,t}^{u}+\frac{\left(  1-\omega_{u}\right)  F_{p}^{r}%
}{\omega_{u}F_{p}^{u}+\left(  1-\omega_{u}\right)  F_{p}^{r}}\hat{F}_{p,t}%
^{r}\\
&  =\left(  1-\xi\delta\right)  \left(  -\sigma\frac{\omega_{u}F_{p}%
^{u}\hat{c}_{t}^{u}+\left(  1-\omega_{u}\right)  F_{p}^{r}\hat{c}_{t}^{r}%
}{\omega_{u}F_{p}^{u}+\left(  1-\omega_{u}\right)  F_{p}^{r}}+\hat{y}%
_{t}\right)  \\
&  +\xi\delta E_{t}\left(  z_{t+1}^{b}+\frac{1}{1-\lambda_{p}}%
\widehat{\tilde{\Pi}}_{t+1|t}^{p}+\frac{\omega_{u}F_{p}^{u}}{\omega_{u}%
F_{p}^{u}+\left(  1-\omega_{u}\right)  F_{p}^{r}}\hat{F}_{p,t+1}^{u}%
+\frac{\left(  1-\omega_{u}\right)  F_{p}^{r}}{\omega_{u}F_{p}^{u}+\left(
1-\omega_{u}\right)  F_{p}^{r}}\hat{F}_{p,t+1}^{r}\right)  .
\end{align*}
Similarly, the term involving $\hat{K}_{p,t}^{u}$ and $\hat{K}_{p,t}^{r}$ in
equation (\ref{K_F}) is calculated as:%
\begin{align*}
&  \frac{\omega_{u}K_{p}^{u}}{\omega_{u}K_{p}^{u}+\left(  1-\omega_{u}\right)
K_{p}^{r}}\hat{K}_{p,t}^{u}+\frac{\left(  1-\omega_{u}\right)  K_{p}^{r}%
}{\omega_{u}K_{p}^{u}+\left(  1-\omega_{u}\right)  K_{p}^{r}}\hat{K}_{p,t}%
^{r}\\
&  =\left(  1-\xi\delta \right)  \left(  -\sigma\frac{\omega_{u}K_{p}%
^{u}\hat{c}_{t}^{u}+\left(  1-\omega_{u}\right)  K_{p}^{r}\hat{c}_{t}^{r}%
}{\omega_{u}K_{p}^{u}+\left(  1-\omega_{u}\right)  K_{p}^{r}}+\frac{1}{\theta
}\hat{y}_{t}-\frac{1}{\theta}z_{t}^{a}+\hat{w}_{t}\right)  \\
&  +\xi\delta E_{t}\left(  z_{t+1}^{b}+\frac{\lambda_{p}}{\left(
1-\lambda_{p}\right)  \theta}\widehat{\tilde{\Pi}}_{t+1|t}^{p}+\frac
{\omega_{u}K_{p}^{u}}{\omega_{u}K_{p}^{u}+\left(  1-\omega_{u}\right)
K_{p}^{r}}\hat{K}_{p,t+1}^{u}+\frac{\left(  1-\omega_{u}\right)  K_{p}^{r}%
}{\omega_{u}K_{p}^{u}+\left(  1-\omega_{u}\right)  K_{p}^{r}}\hat{K}%
_{p,t+1}^{r}\right)  .
\end{align*}
Let the right-hand-side of equation (\ref{K_F}) be denoted as $\hat{X}_{t}$. Then,
using the above relationships just derived, $\hat{X}_{t}$ can be written as:%
\begin{equation*}
\hat{X}_{t}  =\left(  1-\xi\delta \right)  \left[  \left(  \frac{1}%
{\theta}-1\right)  \hat{y}_{t}-\frac{1}{\theta}z_{t}^{a}+\hat{w}_{t}\right]
  +\xi\delta E_{t}\left(  -\frac{\lambda_{p}-\theta}{\left(  \lambda
_{p}-1\right)  \theta}\widehat{\tilde{\Pi}}_{t+1|t}^{p}+\hat{X}_{t+1}\right)
\end{equation*}
Because $\hat{X}_{t}$ is the right-hand-side of equation (\ref{K_F}), equation
(\ref{K_F}) can be written as:%
\begin{align*}
- &  \frac{\xi}{1-\xi}\frac{\lambda_{p}-\theta}{\left(  \lambda
_{p}-1\right)  \theta}\left[  \left(  1-\nu_{p}\right)  \hat{\pi}_{t-1}%
-\hat{\pi}_{t}\right]  =\left(  1-\xi\delta \right)  \left[  \left(
\frac{1}{\theta}-1\right)  \hat{y}_{t}-\frac{1}{\theta}z_{t}^{a}+\hat{w}%
_{t}\right]  \\
&  +\xi\delta E_{t}\left(  -\frac{\lambda_{p}-\theta}{\left(  \lambda
_{p}-1\right)  \theta}\widehat{\tilde{\Pi}}_{t+1|t}^{p}-\frac{\xi}%
{1-\xi}\frac{\lambda_{p}-\theta}{\left(  \lambda_{p}-1\right)  \theta
}\left[  \left(  1-\nu_{p}\right)  \hat{\pi}_{t}-\hat{\pi}_{t+1}\right]
\right)  ,
\end{align*}
or%
\[
\hat{\pi}_{t}=\frac{\xi\left(  1-\nu_{p}\right)  }{\left(  \xi%
+1-\nu_{p}\right)  }\hat{\pi}_{t-1}+\frac{\left(  1-\xi\delta \right)  \left(
1-\xi\right)  \left(  \lambda_{p}-1\right)  \theta}{\left(  \lambda
_{p}-\theta\right)  \left(  \xi+1-\nu_{p}\right)  }\left[  \left(  \frac
{1}{\theta}-1\right)  \hat{y}_{t}-\frac{1}{\theta}z_{t}^{a}+\hat{w}%
_{t}\right]  +\frac{\xi\delta}{\left(  \xi+1-\nu_{p}\right)  }E_{t}%
\hat{\pi}_{t+1}.
\]
From equations (\ref{HHu1}) and (\ref{HHr1}), the wage $\hat{w}_{t}$ can be
written as:%
\begin{align*}
\hat{w}_{t} &  =\omega_{u}\left(  \sigma\hat{c}_{t}^{u}+\frac{1}{\nu}\hat
{h}_{t}^{u}\right)  +\left(  1-\omega_{u}\right)  \left(  \sigma\hat{c}%
_{t}^{r}+\frac{1}{\nu}\hat{h}_{t}^{r}\right)  ,\\
&  =\sigma\hat{y}_{t}+\frac{1}{\nu}\hat{h}_{t}
  =\left(  \sigma+\frac{1}{\nu\theta}\right)  \hat{y}_{t}-\frac{1}{\nu\theta
}z_{t}^{a},
\end{align*}
where the market clearing conditions (\ref{goods}) and (\ref{labor}) were used
in the second equality and the production function (\ref{Y}) was used in the
third equality. Since we assume $c^{u}=c^{r}$, the second equality
holds. By using the expression for $\hat{w}_{t}$, the Phillips curve can be
written as%
\begin{align}
&  \hat{\pi}_{t}=\frac{\xi\left(  1-\nu_{p}\right)  }{\left(  \xi
+1-\nu_{p}\right)  }\hat{\pi}_{t-1}\nonumber\\
+ &  \frac{\left(  1-\xi\delta\right)  \left(  1-\xi\right)  \left(
\lambda_{p}-1\right)  \theta}{\left(  \lambda_{p}-\theta\right)  \left(
\xi+1-\nu_{p}\right)  }\left[  \frac{\nu+\nu\theta\left(  \sigma-1\right)
+1}{\nu\theta}\hat{y}_{t}-\frac{1+\nu}{\nu\theta}z_{t}^{a}\right]
+\frac{\xi\delta}{\left(  \xi+1-\nu_{p}\right)  }E_{t}\hat{\pi}%
_{t+1}.\notag
\end{align}
In the case of no price indexation to the past inflation rate and a linear
production function, that is, in the case of $\nu_{p}=1$ and $\theta=1$, the
Phillips curve collapses to the standard form:%
\[
\hat{\pi}_{t}=\frac{\left(  1-\xi\delta\right)  \left(  1-\xi\right)
}{\xi} \left(  \sigma+\frac{1}{\nu}\right)  \hat{y}_{t}  +\delta E_{t}\hat{\pi}_{t+1} - \frac{\left(  1-\xi\delta\right)  \left(  1-\xi\right)
}{\xi}\frac
{1+\nu}{\nu}z_{t}^{a}.
\]
This completes the derivation of equation (\ref{Phillips}) in the main text, where
\begin{align}
  &\kappa = \frac{\left(  1-\xi\delta\right)  \left(  1-\xi\right)
}{\xi} \left(  \sigma+\frac{1}{\nu}\right), \label{kappa} \\
 &\chi_{a}=\frac{\left(  1-\xi\delta\right)  \left(  1-\xi\right)
}{\xi}\frac
{1+\nu}{\nu}. \label{chia}
\end{align}

\paragraph{Long-term yield.}
From equations  (\ref{l_HHu2_a}) and (\ref{Euler_app}), the interest rate relevant to the aggregate variables has the following equality:
\begin{equation}
  \frac{(1-\omega_{u})c^{r}}{y}E_{t}\hat{R}_{L,t+1} + \frac{\omega_{u}c^{u}}{y}\hat{i}_{t} = (1-\lambda^{\ast})\hat{i}_{t} + \lambda^{\ast}\hat{i}_{t}^{\ast}.\label{R_equality}
\end{equation}
This equation can be written as
\begin{equation}
    E_{t}\hat{R}_{L,t+1} = \begin{cases}
    \hat{i}_{t}^{\ast} & \hat{i}_{t}^{\ast} \geq \hat{\underline{i}} \\
    \frac{\lambda^{\ast}y}{(1-\omega_{u})c^{r}}\hat{i}_{t}^{\ast} + \frac{(1-\lambda^{\ast})y-\omega_{u}c^{u}}{(1-\omega_{u})c^{r}}\hat{\underline{i}} & \hat{i}_{t}^{\ast} < \hat{\underline{i}}
    \end{cases} \label{ERL}
\end{equation}
By using $R_{L,t+1}=\bar{R}_{L,t+1}(\bar{R}_{L,t}-\mu)/(\bar{R}_{L,t+1}-\mu)$, which relates the return of holding long-term bonds $\hat{R}_{L,t+1}$ to the long-term yield $\hat{\bar{R}}_{L,t}$, the long-term yield can be written as
\begin{equation}
   \hat{\bar{R}}_{L,t} = \frac{\bar{R}_{L}-\mu}{\bar{R}_{L}}E_{t}\hat{R}_{L,t+1} + \frac{\mu}{\bar{R}_{L}}E_{t}\hat{\bar{R}}_{L,t+1}, \label{Rbhat}
\end{equation}
where $\bar{R}_{L}>\mu$ in steady state.
Substitution equation (\ref{ERL}) into equation (\ref{Rbhat}) yields
\begin{equation}
    \hat{\bar{R}}_{L,t} = \begin{cases}
    \frac{\bar{R}_{L}-\mu}{\bar{R}_{L}}\hat{i}^{\ast}_{t} + \frac{\mu}{\bar{R}_{L}}E_{t}\hat{\bar{R}}_{L,t+1} & \hat{i}_{t}^{\ast} \geq \hat{\underline{i}} \\
    \frac{\bar{R}_{L}-\mu}{\bar{R}_{L}}\left[\frac{y\lambda^{\ast}}{(1-\omega_{u})c^{r}}\hat{i}^{\ast}_{t}+\frac{(1-\lambda^{\ast})y-\omega_{u}c^{u}}{(1-\omega_{u})c^{r}}\hat{\underline{i}} \right] + \frac{\mu}{\bar{R}_{L}}E_{t}\hat{\bar{R}}_{L,t+1} & \hat{i}^{\ast}_{t}< \hat{\underline{i}}
    \end{cases} \label{Rbhat1}
\end{equation}
Equation (\ref{Rbhat1}) shows that the long-term yield is the discounted sum of the current and future short-term returns, where the short-term return is given by the shadow rate in the non-ELB regime and in the ELB regime it is given by the first two terms in the square brackets in (\ref{Rbhat1}).

\subsection{Parameterization of the model\label{A: parameterization}}

Instead of parameterizing the model presented in Appendix \ref{A: equilibrium conditions}, we parameterize the system of log-linearized equations (\ref{i})-(\ref{Phillips}) in the main text. It is worth emphasizing that we use the parameterized model to illustrate the implications of the theoretical model, and not to study the quantitative implications, which would require a more complex system.

The relative risk aversion parameter is set at $\sigma=2$. The discount factor is set close to unity at $\delta=0.997$. The slope of the Phillips curve $\kappa$ is set at $\kappa=0.336$ using equation (\ref{kappa}) with the Calvo parameter of $\xi=0.75$ and the Frisch labor elasticity of $\nu=0.5$. In the monetary policy rule, the persistence parameter is set at $\rho_{i}=0.7$; the inflation coefficient is set at $r_{\pi}=1.5$; the output coefficient is set at $r_{y}=0.5$. The AR(1) coefficients for the supply and demand shocks are set at $\rho_{a}=\rho_{b}=0.9$, and the coefficients $\chi_{b}$ and $\chi_{a}$ are set according to equations (\ref{chib}) and (\ref{chia}), respectively. The term premium in steady state is set at $\zeta=0.01/4$. We consider different values for the parameters $\lambda^{\ast}$ and $\alpha$ (reported in the main text) to study the effects of UMP.

\subsection{Proof of Proposition \ref{prop1}}\label{a: proof 1}
\textbf{Part (i)} Because of the equivalence established in Lemma \ref{lemma1}, without loss of generality, consider the case of $\lambda^{\ast}=1$ and $\alpha=0$ in the theoretical model. In this case, the variables $\hat{y}_{t}$, $\hat{\pi}_{t}$, and $\hat{i}_{t}^{\ast}$ have a closed system of equations, consisting of equation (\ref{Euler}) with $\lambda^{\ast}=1$, equation (\ref{Phillips}), and $\hat{i}_{t}^{\ast}=\hat{i}_{t}^{\text{Taylor}}$, where $\hat{i}_{t}^{\text{Taylor}}$ is given by equation (\ref{Taylor}).

In this case, the state of the economy in period $t$ can be summarized by $\hat{i}_{t-1}^{\ast}$, $\epsilon_{t}^{i}$, $z_{t}^{a}$, and $z_{t}^{b}$. Then decision rules for $\hat{y}_{t}$ and $\hat{\pi}_{t}$ have the following form:
\begin{align}
    &\hat{y}_{t} = d_{y i^{\ast}}\hat{i}_{t-1}^{\ast} + d_{yi}\epsilon_{t}^{i} + d_{ya}z_{t}^{a} + d_{yb}z_{t}^{b}, \notag \\
    &\hat{\pi}_{t} = d_{\pi i^{\ast}}\hat{i}_{t-1}^{\ast} + d_{\pi i}\epsilon_{t}^{i} + d_{\pi a}z_{t}^{a} + d_{\pi b}z_{t}^{b}, \notag
\end{align}
with coefficients $\{d_{y i^{\ast}}, d_{yi}, d_{ya}, d_{yb}, d_{\pi i^{\ast}}, d_{\pi i}, d_{\pi a}, d_{\pi b}\}$ uniquely determined under standard assumptions of the model (such as the Taylor principle).
With these decision rules, the equation for $\hat{i}_{t}^{\ast}$ can be written as
\begin{align}
  \hat{i}_{t}^{\ast} = &\left[\rho_{i}+ (1-\rho_{i})\left(r_{\pi}d_{\pi i^{\ast}}+r_{y}d_{y i^{\ast}}\right)\right]\hat{i}_{t-1}^{\ast}+ \left[(1-\rho_{i})\left(r_{\pi}d_{\pi i} + r_{y}d_{y i}\right)+1\right]\epsilon_{t}^{i} \notag \\
+&
(1-\rho_{i})\left(r_{\pi}d_{\pi a} + r_{y}d_{y a}\right)z_{t}^{a} + (1-\rho_{i})\left(r_{\pi}d_{\pi b} + r_{y}d_{y b}\right)z_{t}^{b} \notag \\
=& d_{i^{\ast}i^{\ast}}\hat{i}_{t-1}^{\ast} + d_{i^{\ast}i}\epsilon_{t}^{i} + d_{i^{\ast}a}z_{t}^{a} + d_{i^{\ast}b}z_{t}^{b}. \label{i*hat}
\end{align}
Let $\mathbf{y}_{t} \equiv [\hat{y}_{t}, \hat{\pi}_{t}, \hat{i}_{t}^{\ast}]^{\prime}$ denote the vector of endogenous variables. The decision rule implies
\begin{align}
   \mathbf{y}_{t} = &\begin{bmatrix}
                               d_{y i^{\ast}} & d_{yi} & d_{ya} & d_{yb} \\
d_{\pi i^{\ast}} & d_{\pi i} & d_{\pi a} & d_{\pi b} \\
d_{i^{\ast}i^{\ast}} & d_{i^{\ast}i} & d_{i^{\ast}a} & d_{i^{\ast}b}
                              \end{bmatrix}
                  \begin{bmatrix}
                   i_{t-1}^{\ast} \\
                   \epsilon_{t}^{i} \\
                   \rho_{a}z_{t-1}^{a} + \epsilon_{t}^{a} \\
                   \rho_{b}z_{t-1}^{b} + \epsilon_{t}^{b}
                 \end{bmatrix} \notag \\
        =&\begin{bmatrix}
                               d_{y i^{\ast}} &  \rho_{a}d_{ya} & \rho_{b}d_{yb} \\
d_{\pi i^{\ast}}  & \rho_{a}d_{\pi a} & \rho_{b}d_{\pi b} \\
d_{i^{\ast}i^{\ast}}  & \rho_{a}d_{i^{\ast}a} & \rho_{b}d_{i^{\ast}b}
                              \end{bmatrix}
                  \begin{bmatrix}
                   i_{t-1}^{\ast} \\
                   z_{t-1}^{a}\\
                   z_{t-1}^{b}
                 \end{bmatrix}+
                 \begin{bmatrix}
                               d_{yi} & d_{ya} & d_{yb} \\
d_{\pi i} & d_{\pi a} & d_{\pi b} \\
d_{i^{\ast}i} & d_{i^{\ast}a} & d_{i^{\ast}b}
                              \end{bmatrix}
                              \begin{bmatrix}
                              \epsilon_{t}^{i} \\
                              \epsilon_{t}^{a} \\
                              \epsilon_{t}^{b}
                              \end{bmatrix} \notag \\
=& \mathbf{C} \mathbf{x}_{t-1} + \mathbf{D} \mathbf{\epsilon}_{t}. \label{LMy}
\end{align}
The law of motion for $\mathbf{x}_{t}\equiv [\hat{i}_{t}^{\ast}, z_{t}^{a}, z_{t}^{b}]^{\prime}$ is:
\begin{align}
  \mathbf{x}_{t} =& \begin{bmatrix}
                 d_{i^{\ast}i^{\ast}} & \rho_{a}d_{i^{\ast}a} & \rho_{b}d_{i^{\ast}b} \\
                 0   & \rho_{a} & 0 \\
                 0   &  0 & \rho_{b}
                 \end{bmatrix}\mathbf{x}_{t-1}
                  +
                 \begin{bmatrix}
                 d_{i^{\ast}i} & d_{i^{\ast}a} & d_{i^{\ast}b} \\
                 0    & 1 & 0 \\
                 0    & 0  & 1
                \end{bmatrix}\mathbf{\epsilon}_{t}\notag \\
=& \mathbf{A} \mathbf{x}_{t-1} + \mathbf{B}\mathbf{\epsilon}_{t}.\label{LMx}
\end{align}
Solving equation (\ref{LMy}) for ${\epsilon}_{t}$, and substituting the outcome in equation (\ref{LMx}) yields:
\begin{equation}
   \mathbf{x}_{t} = \left(\mathbf{A}-\mathbf{B}\mathbf{D}^{-1}\mathbf{C}\right)\mathbf{x}_{t-1} + \mathbf{B}\mathbf{D}^{-1}\mathbf{y}_{t}. \notag
\end{equation}
If $\mathbf{A}-\mathbf{B}\mathbf{D}^{-1}\mathbf{C}=\mathbf{0}$,  the vector of endogenous variables, $\mathbf{y}_{t}$, has a VAR(1) representation:
\begin{equation}
   \mathbf{y}_{t} = \mathbf{C}\mathbf{B}\mathbf{D}^{-1}\mathbf{y}_{t-1} + \mathbf{D}\mathbf{\epsilon}_{t}. \notag
\end{equation}
The rest of the proof shows $\mathbf{A}-\mathbf{B}\mathbf{D}^{-1}\mathbf{C}=\mathbf{0}$.
Substituting the matrices $\mathbf{A}$ and $\mathbf{B}$ in equation (\ref{LMx}) into this condition yields:
\begin{equation}
   \mathbf{D}^{-1}\mathbf{C} = \begin{bmatrix}
d_{i^{\ast}i^{\ast}}/d_{i^{\ast}i} & 0 & 0 \\
0 & \rho_{a} & 0 \\
0 & 0 & \rho_{b}
\end{bmatrix}.\notag
\end{equation}
Further substituting the matrices $\mathbf{C}$ and $\mathbf{D}$ in equation (\ref{LMy}) into this condition leads to: $\mathbf{A}-\mathbf{B}\mathbf{D}^{-1}\mathbf{C}=\mathbf{0}$ if and only if $d_{y i^{\ast}} = d_{yi} \left(d_{i^{\ast}i^{\ast}}/d_{i^{\ast}i}\right)$ and $d_{\pi i^{\ast}} = d_{\pi i} \left(d_{i^{\ast}i^{\ast}}/d_{i^{\ast}i}\right)$.
Substituting the decision rules into equation (\ref{Euler}) yields:
\begin{align}
 &\hat{y}_{t} = \left(d_{y i^{\ast}} - \frac{1}{\sigma} + \frac{d_{\pi i}}{\sigma}\right)d_{i^{\ast} i^{\ast}}\hat{i}_{t-1}^{\ast} + \left(d_{y i^{\ast}} - \frac{1}{\sigma} + \frac{d_{\pi i}}{\sigma}\right)d_{i^{\ast} i}\epsilon_{t}^{i} + ..., \notag
\end{align}
where terms related to $z_{t}^{a}$ and $z_{t}^{b}$ are omitted. Matching coefficients on $\hat{i}_{t-1}^{\ast}$ and $\epsilon_{t}^{i}$ of both sides of the equation yields:
\begin{align}
   d_{yi^{\ast}} = &\left(d_{y i^{\ast}} - \frac{1}{\sigma} + \frac{d_{\pi i}}{\sigma}\right)d_{i^{\ast} i^{\ast}}, \notag \\
    d_{y i} = &\left(d_{y i^{\ast}} - \frac{1}{\sigma} + \frac{d_{\pi i}}{\sigma}\right)d_{i^{\ast} i}. \notag
\end{align}
These two equations imply $d_{yi^{\ast}} = d_{yi}\left(d_{i^{\ast}i^{\ast}}/d_{i^{\ast}i}\right)$. Next, substituting the decision rules into equation (\ref{Phillips}) yields:
\begin{equation}
  \hat{\pi}_{t} = \left(\delta d_{\pi i^{\ast}} + \kappa d_{y i^{\ast}}\right) \hat{i}_{t-1}^{\ast} + \left(\delta d_{\pi i^{\ast}}d_{i^{\ast}i} + \kappa d_{yi}\right)\epsilon_{t}^{i} + ..., \notag
\end{equation}
where terms related to $z_{t}^{a}$ and $z_{t}^{b}$ are omitted. Matching coefficients on $\hat{i}_{t-1}^{\ast}$ and $\epsilon_{t}^{i}$ of both sides of the equation yields:
\begin{align}
  d_{\pi i^{\ast}} = &\delta d_{\pi i^{\ast}} + \kappa d_{yi}\left(\frac{d_{i^{\ast}i^{\ast}}}{d_{i^{\ast}i}}\right), \notag \\
  d_{\pi i} = &\delta d_{\pi i^{\ast}}d_{i^{\ast}i} + \kappa d_{yi}, \notag
\end{align}
where $d_{yi^{\ast}}=d_{yi}\left(d_{i^{\ast}i^{\ast}}/d_{i^{\ast}i}\right)$ is used in the first equation. Solving these two equations for $d_{\pi i^{\ast}}$ yields $d_{\pi i^{\ast}} = d_{\pi i}\left(d_{i^{\ast}i^{\ast}}/d_{i^{\ast}i}\right)$.
\vspace{.3cm}

\noindent
\textbf{Part (ii)} Again, without loss of generality, consider the case of $\lambda^{\ast}=1$ and $\alpha=0$. Under Assumption \ref{assumption1} and the irrelevance hypothesis, the long-term yield can be written as (\ref{RLbarhat3}) with $\lambda^{\ast}=1$ as
\begin{equation}
    \hat{\bar{R}}_{L}=\frac{\bar{R}_{L}-\mu}{\bar{R}_{L}}\hat{i}^{\ast}_{t} + \frac{\mu}{\bar{R}_{L}}E_{t}\hat{\bar{R}}_{L,t+1}. \notag
\end{equation}
Solving this equation forward yields
\begin{equation}
  \hat{\bar{R}}_{L,t} = \left(\frac{\bar{R}_{L}-\mu}{\bar{R}_{L}}\right) E_{t}\left[\hat{i}_{t}^{\ast} + \frac{\mu}{\bar{R}_{L}}\hat{i}_{t+1}^{\ast} + \left(\frac{\mu}{\bar{R}_{L}}\right)^{2}\hat{i}_{t+2}^{\ast} + ... \right].\notag
\end{equation}
Because the right-hand-side of the equation depends on information in period $t$, which consist of $\hat{i}_{t}^{\ast}$, $z_{t}^{a}$, and $z_{t}^{b}$, the long-term interest rate can be written as:
\begin{equation}
  \hat{\bar{R}}_{L,t} = f_{i^{\ast}}\hat{i}_{t}^{\ast} + f_{a}z_{t}^{a} + f_{b}z_{t}^{b}, \notag
\end{equation}
where $f_{i^{\ast}}$, $f_{a}$, and $f_{b}$ are coefficients derived by using equation (\ref{i*hat}) as
\begin{align*}
 &f_{i^{\ast}} = \frac{\bar{R}_{L}-\mu}{\bar{R}_{L}- d_{i^{\ast}i^{\ast}}\mu}, \\
 &f_{a} = \frac{(\bar{R}_{L}-\mu) d_{i^{\ast}a}\rho_{a}\mu}{(\bar{R}_{L}-\rho_{a}\mu)(\bar{R}_{L}- d_{i^{\ast}i^{\ast}}\mu)}, \\
 &f_{b} = \frac{(\bar{R}_{L}-\mu) d_{i^{\ast}b}\rho_{b}\mu}{(\bar{R}_{L}-\rho_{b}\mu)(\bar{R}_{L}- d_{i^{\ast}i^{\ast}}\mu)}.
\end{align*}
Again by using equation (\ref{i*hat}) the equation for the long-term yield can be written as:
\begin{equation}
  \hat{\bar{R}}_{L,t} = f_{i^{\ast}}d_{i^{\ast}i^{\ast}}\hat{i}_{t-1}^{\ast} + f_{i^{\ast}}d_{i^{\ast}i}\epsilon_{t}^{i} + \left(f_{i^{\ast}}d_{i^{\ast}a}+f_{a}\right)z_{t}^{a} + \left(f_{i^{\ast}}d_{i^{\ast}b}+f_{b}\right)z_{t}^{b}. \notag
\end{equation}
Define $\mathbf{y}_{t}\equiv [\hat{y}_{t}, \hat{\pi}_{t}, \hat{\bar{R}}_{L,t}]^{\prime}$, $\mathbf{x}_{t}\equiv [\hat{i}_{t-1}^{\ast}, z_{t}^{a}, z_{t}^{b}]^{\prime}$, and $\mathbf{\epsilon}_{t}=[\epsilon_{t}^{i}, \epsilon_{t}^{a}, \epsilon_{t}^{b}]^{\prime}$. Then, the state space representation for $\mathbf{y}$ is
\begin{align}
   \mathbf{y}_{t}
        =&\begin{bmatrix}
                               d_{y i^{\ast}} &  \rho_{a}d_{ya} & \rho_{b}d_{yb} \\
d_{\pi i^{\ast}}  & \rho_{a}d_{\pi a} & \rho_{b}d_{\pi b} \\
f_{i^{\ast}}d_{i^{\ast}i^{\ast}}  & \rho_{a}(f_{i^{\ast}}d_{i^{\ast}a}+f_{a}) & \rho_{b}(f_{i^{\ast}}d_{i^{\ast}b}+f_{b})
                              \end{bmatrix}
                  \mathbf{x}_{t}+
                 \begin{bmatrix}
                               d_{yi} & d_{ya} & d_{yb} \\
d_{\pi i} & d_{\pi a} & d_{\pi b} \\
f_{i^{\ast}}d_{i^{\ast}i} & f_{i^{\ast}}d_{i^{\ast}a}+f_{a} & f_{i^{\ast}}d_{i^{\ast}b}+f_{b}
                              \end{bmatrix}
                              \mathbf{\epsilon}_{t} \notag \\
=& \mathbf{C} \mathbf{x}_{t-1} + \mathbf{D} \mathbf{\epsilon}_{t}. \notag
\end{align}
and
\begin{align}
  \mathbf{x}_{t} =& \begin{bmatrix}
                 d_{i^{\ast}i^{\ast}} & \rho_{a}d_{i^{\ast}a} & \rho_{b}d_{i^{\ast}b}  \\
                 0   & \rho_{a} & 0 \\
                 0   &  0 & \rho_{b}
                 \end{bmatrix}\mathbf{x}_{t-1}
                  +
                 \begin{bmatrix}
                 d_{i^{\ast}i} & d_{i^{\ast}a} & d_{i^{\ast}b}  \\
                 0    & 1 & 0 \\
                 0    & 0  & 1
                \end{bmatrix}\mathbf{\epsilon}_{t}\notag \\
=& \mathbf{A} \mathbf{x}_{t-1} + \mathbf{B}\mathbf{\epsilon}_{t}.\notag
\end{align}
Similar to the part (i) in Proposition \ref{prop1}, a solution for $\mathbf{y}_{t}$ has a VAR(1) representation if and only if $\mathbf{A}-\mathbf{B}\mathbf{D}^{-1}\mathbf{C}=\mathbf{0}$. This condition holds if and only if $d_{yi^{\ast}}=d_{yi}(d_{i^{\ast}i^{\ast}}/d_{i^{\ast}i})$ and $d_{\pi i^{\ast}}=d_{\pi i}(d_{i^{\ast}i^{\ast}}/d_{i^{\ast}i})$. The latter two conditions hold as shown in Part (i).

\subsection{Proof of Proposition \ref{prop2}}\label{a: proof2}
We show that equations (\ref{i}), (\ref{i*}), (\ref{Taylor}), (\ref{Euler}), and (\ref{Phillips}) can be written in the empirical structural form of equations (\ref{eq: Y2}), (\ref{eq: Y2*}), and (\ref{eq: Y1}). This will prove the proposition since the structural form has a piecewise linear representation, as explained in the main text. It is straightforward to see that equations (\ref{i}), (\ref{i*}), and (\ref{Taylor}) in the theoretical model can be written in the form of equations (\ref{eq: Y2}) and (\ref{eq: Y2*}) in the empirical model. Below, we are going to show that equations (\ref{Euler}) and (\ref{Phillips}) can be represented by the structural form equation (\ref{eq: Y1}).

Without loss of generality, consider a case in which agents forming expectations assuming: $\lambda^{\ast}=1$ and $\alpha=0$. When forming expectations about variables in period $t+1$, the initial condition is given by $\tilde{\mathbf{x}}_{t}\equiv [(1-\lambda^{\ast})\hat{i}_{t}+\lambda^{\ast}\hat{i}_{t}^{\ast}, z_{t}^{a}, z_{t}^{b}]^{\prime}$. Under Assumption \ref{assumption3},  the decision rule used for forming expectations about period $t+1$ variables is  $\mathbf{y}_{t+1} = \mathbf{C}\tilde{\mathbf{x}}_{t} + \mathbf{D}\mathbf{\epsilon}_{t+1}$, where $\mathbf{C}$ and $\mathbf{D}$ are those defined in equation (\ref{LMy}).
From period $t+s$ onward, for $s=2,3,...$, time $t+s$ variables are expected \textit{in period $t$} to follow $\mathbf{y}_{t+s}=\mathbf{C}\mathbf{x}_{t+s-1} + \mathbf{D}\mathbf{\epsilon}_{t+s}$, where $\mathbf{x}_{t}\equiv [\hat{i}_{t}^{\ast},z_{t}^{a},z_{t}^{b}]^{\prime}$.
But, once the time proceeds and becomes period $t+1$, the initial condition is updated to $\tilde{\mathbf{x}}_{t+1}$ and this is used for forming expectations about $t+2$ variables as $E_{t+1}\mathbf{y}_{t+2}=\mathbf{C}\tilde{\mathbf{x}}_{t+1}$. Hence, under the assumption about expectations, the decision rule is given by $\mathbf{y}_{t+s} = \mathbf{C}\tilde{\mathbf{x}}_{t+s-1} + \mathbf{D}\mathbf{\epsilon}_{t+s}$ for $s=1,2,...$ In this system, in every period information is updated and $\tilde{\mathbf{x}}_{t+s-1}$ is used as an initial condition.  The interest rate $\hat{i}_{t+s-1}$ in the initial condition is treated as if it were an exogenous variable.

By substituting the decision rule into the expected variables, equations (\ref{Euler}) and (\ref{Phillips}) can be written as:
\begin{align}
 &\hat{y}_{t} = \left(-\frac{1}{\sigma} +d_{y i^{\ast}}+\frac{d_{\pi i^{\ast}}}{\sigma}\right) \left((1-\lambda^{\ast})\hat{i}_{t} + \lambda^{\ast}\hat{i}_{t}^{\ast}\right) +  \left(\rho_{a}d_{y a}+\frac{\rho_{a} d_{\pi a}}{\sigma}\right)z_{t}^{a} + \left(\rho_{b}d_{y b}
  +\frac{\rho_{b} d_{\pi b}}{\sigma} - \chi_{z}\right)z_{t}^{b}, \label{yhat_app} \\
   &-\kappa y_{t} + \hat{\pi}_{t} =  \delta d_{\pi i^{\ast}}\left((1-\lambda^{\ast})\hat{i}_{t} + \lambda^{\ast}\hat{i}_{t}^{\ast}\right) + (\delta \rho_{a} d_{\pi a} - \chi_{a})z_{t}^{a} + \delta \rho_{b} d_{\pi b}z_{t}^{b}. \label{pihat_app}
\end{align}
Since $z_{t}^{a}$ and $z_{t}^{b}$ follow AR(1) processes, equations (\ref{yhat_app}) and (\ref{pihat_app}) can be written in a matrix form as:
\begin{equation}
  \mathbf{H}_{1}\begin{bmatrix} \hat{y}_{t} \\ \hat{\pi}_{t} \end{bmatrix} = \mathbf{H}_{2}\left((1-\lambda^{\ast})\hat{i} + \lambda^{\ast}\hat{i}_{t}^{\ast}\right)+ \mathbf{H}_{3}\begin{bmatrix} z_{t-1}^{a} \\ z_{t-1}^{b} \end{bmatrix} + \mathbf{H}_{4}\begin{bmatrix} \epsilon_{t}^{a} \\ \epsilon_{t}^{b}\end{bmatrix}. \notag
\end{equation}
or
\begin{equation}
  \begin{bmatrix} \hat{y}_{t} \\ \hat{\pi}_{t} \end{bmatrix} = \mathbf{H}_{1}^{-1}\mathbf{H}_{2}\left((1-\lambda^{\ast})\hat{i} + \lambda^{\ast}\hat{i}_{t}^{\ast}\right)+ \mathbf{H}_{1}^{-1}\mathbf{H}_{3}\begin{bmatrix} z_{t-1}^{a} \\ z_{t-1}^{b} \end{bmatrix} + \mathbf{H}_{1}^{-1}\mathbf{H}_{4}\begin{bmatrix} \epsilon_{t}^{a} \\ \epsilon_{t}^{b}\end{bmatrix}. \label{structural}
\end{equation}
Also, under Assumption \ref{assumption3}, the expected values can be written as: $E_{t}\tilde{\mathbf{y}}_{t+1} = \mathbf{G}\tilde{\mathbf{y}}_{t}$, where $\tilde{\mathbf{y}}_{t}\equiv [\hat{y}_{t}, \hat{\pi}_{t}, (1-\lambda^{\ast})\hat{i}_{t} + \lambda^{\ast}\hat{i}_{t}^{\ast}]^{\prime}$ and $\mathbf{G}\equiv \mathbf{C}\mathbf{B}\mathbf{D}^{-1}$, as derived in the proof of Proposition \ref{prop1}.
By using this equation, equations (\ref{Euler}) and (\ref{Phillips}) can be written as:
\begin{align}
   &\chi_{z}z_{t}^{b} = \left(g_{yy} + \frac{g_{\pi y}}{\sigma}-1\right) \hat{y}_{t} + \left(g_{y\pi} + \frac{g_{\pi \pi}}{\sigma}\right)\hat{\pi}_{t}+ \left(g_{yi^{\ast}} + \frac{g_{\pi i^{\ast}}}{\sigma} - \frac{1}{\sigma}\right)\left((1-\lambda^{\ast})\hat{i}_{t} + \lambda^{\ast}\hat{i}_{t}^{\ast}\right), \notag \\
&\chi_{a}z_{t}^{a} = \left(\delta g_{\pi y} + \kappa\right)\hat{y}_{t} + \left(\delta g_{\pi \pi} - 1\right)\hat{\pi}_{t} + \delta g_{\pi i^{\ast}}\left((1-\lambda^{\ast})\hat{i}_{t} + \lambda^{\ast}\hat{i}_{t}^{\ast}\right), \notag
\end{align}
where $g_{ij}$'s correspond to elements in the matrix $\mathbf{G}$.
Then, the lagged shocks $z_{t-1}^{b}$ and $z_{t-1}^{a}$ in equation (\ref{structural}) can be represented by a function of $\tilde{\mathbf{y}}_{t-1}\equiv [\hat{y}_{t-1}, \hat{\pi}_{t-1}, (1-\lambda^{\ast})\hat{i}_{t-1} + \lambda^{\ast}\hat{i}_{t-1}^{\ast}]^{\prime}$.  From this result, equation (\ref{structural}) is in the same form of equation (\ref{eq: Y1}) in the structural form.

\subsection{Impulse responses to demand and supply shocks\label{app:pref_supp_shocks}}

We study impulse responses to a demand shock and a supply shock, respectively, in an ELB regime, using the theoretical model presented in Section \ref{s: dsge} of the main text.  We show that the responses differ significantly depending on the effectiveness of UMP.

Figure \ref{Fig:demandshock} plots impulse responses of output and inflation in the theoretical model to the contractionary demand shock of $\epsilon_{t}^{b}=0.25/400$ under the ELB. The responses are calculated exactly in the same way as those to a monetary policy shock, shown in Figure \ref{Fig:MPshock} in the main text. In the case of no UMP ($\xi=0$), a negative demand shock causes the largest declines in output and inflation. As the effectiveness of UMP increases, i.e., as $\xi$ increases, the negative responses of output and inflation become smaller.

\begin{figure}[t]
\caption{Impulse responses to a demand shock at the ELB}%
\label{Fig:demandshock}
\centering
\includegraphics[width=16.5cm]{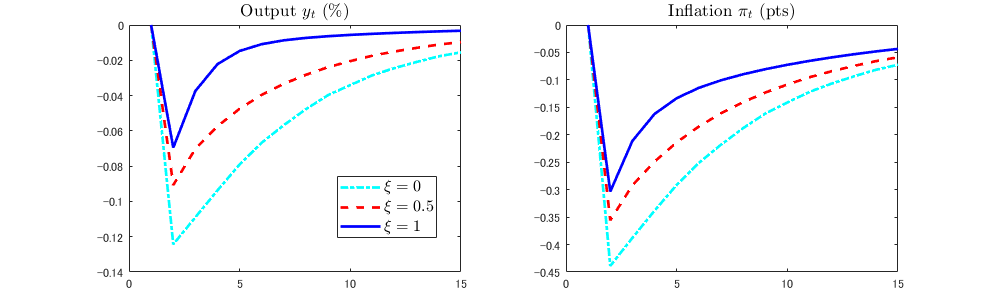}
\vspace{-.5cm}
\end{figure}

Figure \ref{Fig:supplyshock} plots impulse responses of output and inflation to the negative supply shock of $\epsilon_{t}^{a}=-0.25/100$ under the ELB. In the case of fully effective UMP ($\xi=1$), output decreases and inflation increases in response to the negative supply shock, as in the responses in a non-ELB regime.  However, as the effectiveness of UMP decreases, the degree of a decrease in output shrinks, and output even increases on impact in response to the negative supply shock in the case of no UMP ($\xi=0$).  This is driven by a stronger increase in inflation under the ELB.   Such an increase in inflation mitigates the negative impact of the ELB on output.  This effect dominates the direct effect of the negative supply shock, resulting in an increase in output on impact.

While we exclusively focus on monetary policy shocks in this paper, the same co-movement of the variables in response to supply and demand shocks when the economy approaches the ELB would pose a challenge for identifying responses to demand and supply shocks.

\begin{figure}[t]
\caption{Impulse responses to a supply shock at the ELB}%
\label{Fig:supplyshock}
\centering
\includegraphics[width=16.5cm]{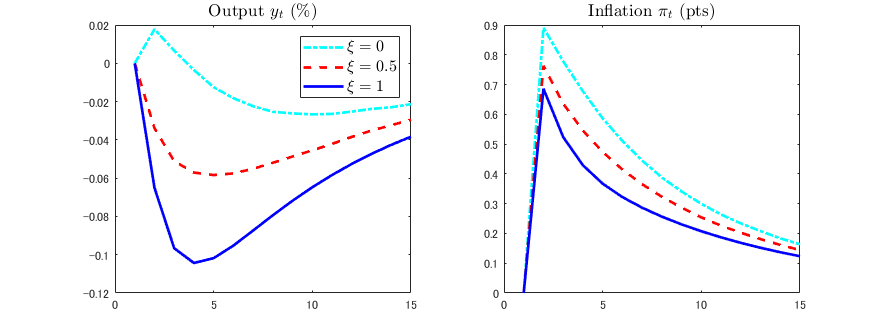}
\vspace{-.5cm}
\end{figure}

\section{Derivation of the attenuation effect (\ref{eq: attenuation})}\label{app: attenuation}

Start from the definition of the IRF to the monetary policy shock
$\bar{\varepsilon}_{2t}:=A_{22}^{\ast-1}\varepsilon_{2t}$.\footnote{Note that
the derivation of (\ref{eq: attenuation}) remains the same if we worked with
perturbations to $\varepsilon_{2t}$ instead of $\bar{\varepsilon}_{2t},$ but
we choose the latter to avoid carrying $A_{22}^{\ast}$ around in the
derivation.} This is a function of the shock  magnitude $\varsigma$ and horizon
$h$:%

\begin{equation}
IRF_{h,t}\left(  \varsigma\right)  =E\left(  Y_{1t+h}|\bar{\varepsilon}%
_{2t}=\varsigma,x_{t}\right)  -E\left(  Y_{1t+h}|\bar{\varepsilon}%
_{2t}=0,x_{t}\right)  ,\label{eq: IRF to MPS}%
\end{equation}
where $x_{t}=\left(  x_{1t}^{\prime},x_{2t}\right)  ^{\prime}$, $x_{it}%
:=\allowbreak C_{i}X_{t}+C_{i}^{\ast}X_{t}^{\ast}\ $embodies all the relevant
history of $Y_{t}$ up to period $t-1$.\footnote{$x_{t}$ is the sufficient
statistic for the entire history of $Y_{t}$ in the conditional expectations, i.e.,
$E\left(  Y_{t+h}|\bar{\varepsilon}_{2t},Y_{t-1},Y_{t-2},...\right)
\allowbreak=E\left(  Y_{t+h}|\bar{\varepsilon}_{2t},x_{t}\right)  $.} We only
need to discuss the impact effects, so we set $h=0$ in (\ref{eq: IRF to MPS})
and write $IRF_{0,t}\left(  \varsigma\right)  =g\left(  \varsigma
;x_{t}\right)  -g\left(  0;x_{t}\right)  ,$ where%
\begin{equation}
g\left(  \varsigma;x_{t}\right)  :=E\left(  Y_{1t}|\bar{\varepsilon}%
_{2t}=\varsigma,x_{t}\right)  .\label{eq: g}%
\end{equation}
Note that, despite the kink in the model, the fact that we are taking
expectations with respect to the remaining shocks $\bar{\varepsilon}%
_{1t}:=A_{11}^{-1}\varepsilon_{1t}$ implies that the function $g$ is smooth in both
arguments. 
 So, we can consider
infinitesimal interventions by computing $\lim_{\varsigma\rightarrow0}%
\frac{g\left(  \varsigma;x_{t}\right)  -g\left(  0;x_{t}\right)  }{\varsigma
}=:\frac{\partial g\left(  \varsigma;x_{t}\right)  }{\partial\varsigma}.$

Let $D_{t}=1_{\left\{  i_{t}^{\ast}<\underline{i}_{t}\right\}  }$ denote the indicator
that the interest rate is at the ELB. Then, using equation (\ref{eq: CKSVAR-RF}) in the main text,
we obtain:
\begin{align}
g\left(  \varsigma,x_{t}\right)   &  =C_{1}X_{t}+C_{12}^{\ast}X_{2t}^{\ast
}+E\left(  u_{1t}|\bar{\varepsilon}_{2t}=\varsigma,x_{t}\right)  \nonumber\\
&  -\widetilde{\beta}E\left(  D_{t}\left(  C_{2}X_{t}+C_{22}^{\ast}%
X_{2t}^{\ast}-\underline{i}_{t}+u_{2t}\right)  |\bar{\varepsilon}_{2t}=\varsigma
,x_{t}\right)  \nonumber\\
&  =x_{1t}+\left(  I_{k-1}-\beta\gamma\right)  ^{-1}\beta\varsigma
-\widetilde{\beta}E\left(  D_{t}|\bar{\varepsilon}_{2t}=\varsigma
,x_{t}\right)  \left(  x_{2t}-\underline{i}_{t}+\frac{\varsigma}{1-\gamma\beta}\right)
\nonumber\\
&  -\widetilde{\beta}E\left(  \left.  \frac{\gamma\bar{\varepsilon}_{1t}%
}{1-\gamma\beta}D_{t}\right\vert \bar{\varepsilon}_{2t}=\varsigma
,x_{t}\right)  \nonumber\\
&  =x_{1t}+\left(  I_{k-1}-\beta\gamma\right)  ^{-1}\beta\varsigma
-\widetilde{\beta}\Phi\left(  \frac{\underline{i}_{t}-x_{2t}-\frac{\varsigma}%
{1-\gamma\beta}}{\varpi}\right)  \left(  x_{2t}-\underline{i}_{t}+\frac{\varsigma
}{1-\gamma\beta}\right)  \nonumber\\
&  +\widetilde{\beta}\varpi\phi\left(  \frac{\underline{i}_{t}-x_{2t}-\frac{\varsigma
}{1-\gamma\beta}}{\varpi}\right)  .\label{eq: g solved}%
\end{align}
The second equality in equation (\ref{eq: g solved}) follows from the definitions of $x_{it}$ and
the fact that:
\[
u_{1t}=\left(  I_{k-1}-\beta\gamma\right)  ^{-1}\left(  \bar{\varepsilon}%
_{1t}+\beta\bar{\varepsilon}_{2t}\right)  ,\text{ \ and \ \ }u_{2t}%
:=\frac{\gamma\bar{\varepsilon}_{1t}+\bar{\varepsilon}_{2t}}{1-\gamma\beta},
\]
see \cite[equations 32 and 33]{Mavroeidis2019}, and the third equality in equation (\ref{eq: g solved}) follows
from:
\begin{align*}
E\left(  D_{t}|\bar{\varepsilon}_{2t}=\varsigma,x_{t}\right)   &  =\Pr\left(
u_{2t}<\underline{i}_{t}-x_{2t}|\bar{\varepsilon}_{2t}=\varsigma,x_{t}\right)  \\
&  =\Pr\left(  \frac{\gamma\bar{\varepsilon}_{1t}}{1-\gamma\beta}<\underline{i}_{t}%
-x_{2t}-\frac{\varsigma}{1-\gamma\beta}\right)  \\
&  =\Phi\left(  \frac{b-x_{2t}-\frac{\varsigma}{1-\gamma\beta}}{\varpi
}\right)  ,\quad\varpi^{2}:=var\left(  \frac{\gamma\bar{\varepsilon}_{1t}%
}{1-\gamma\beta}\right)
\end{align*}
and%
\begin{align*}
E\left(  \left.  \frac{\gamma\bar{\varepsilon}_{1t}}{1-\gamma\beta}%
D_{t}\right\vert \bar{\varepsilon}_{2t}=\varsigma,x_{t}\right)   &  =E\left(
\frac{\gamma\bar{\varepsilon}_{1t}}{1-\gamma\beta}\left\vert \frac{\gamma
\bar{\varepsilon}_{1t}}{1-\gamma\beta}<\underline{i}_{t}-x_{2t}-\frac{\varsigma}%
{1-\gamma\beta},\bar{\varepsilon}_{2t}=\varsigma,x_{t}\right.  \right)  \\
&  \times\Pr\left(  \left.  \frac{\gamma\bar{\varepsilon}_{1t}}{1-\gamma\beta
}<\underline{i}_{t}-x_{2t}-\frac{\varsigma}{1-\gamma\beta}\right\vert \bar{\varepsilon
}_{2t}=\varsigma,x_{t}\right)  \\
&  =-\varpi\phi\left(  \frac{\underline{i}_{t}-x_{2t}-\frac{\varsigma}{1-\gamma\beta}%
}{\varpi}\right)  ,
\end{align*}
where the second equality in the last expression follows from the independence
of $\bar{\varepsilon}_{1t}$ from $\bar{\varepsilon}_{2t}$ and $x_{t},$ and the
properties of the truncated standard normal distribution, i.e., $E\left(
z|z<a\right)  =-\phi\left(  a\right)  /\Phi\left(  a\right)  $.

Differentiating equation (\ref{eq: g solved}) with respect to $\varsigma$ yields:%
\begin{align*}
\frac{\partial g\left(  \varsigma,x_{t}\right)  }{\partial\varsigma} &
=\left(  I_{k-1}-\beta\gamma\right)  ^{-1}\beta-\frac{1}{1-\gamma\beta
}\widetilde{\beta}\Phi\left(  \frac{\underline{i}_{t}-x_{2t}-\frac{\varsigma}%
{1-\gamma\beta}}{\varpi}\right)  \\
&  +\widetilde{\beta}\frac{1}{1-\gamma\beta}\phi\left(  \frac{\underline{i}_{t}%
-x_{2t}-\frac{\varsigma}{1-\gamma\beta}}{\varpi}\right)  \left(  \frac
{x_{2t}-\underline{i}_{t}+\frac{\varsigma}{1-\gamma\beta}}{\varpi}\right)  \\
&  +\widetilde{\beta}\frac{1}{1-\gamma\beta}\left(  \frac{\underline{i}_{t}-x_{2t}%
-\frac{\varsigma}{1-\gamma\beta}}{\varpi}\right)  \phi\left(  \frac
{\underline{i}_{t}-x_{2t}-\frac{\varsigma}{1-\gamma\beta}}{\varpi}\right)  \\
&  =\left(  I_{k-1}-\beta\gamma\right)  ^{-1}\beta-\frac{1}{1-\gamma\beta
}\widetilde{\beta}\Phi\left(  \frac{\underline{i}_{t}-x_{2t}-\frac{\varsigma}%
{1-\gamma\beta}}{\varpi}\right)  ,
\end{align*}
where the first equality follows from the fact that $\partial\phi\left(
z\right)  /\partial z=-z\phi\left(  z\right)  .$ Evaluating the above
expression at $\varsigma=0$ yields the impact effect of a small monetary
policy shock on $Y_{1t}$ in period $t$, which is a $k-1$ vector, namely,%
\[
IR_{t}:=\left.  \frac{\partial g\left(  \varsigma,x_{t}\right)  }%
{\partial\varsigma}\right\vert _{\varsigma=0}=\left(  I_{k-1}-\beta
\gamma\right)  ^{-1}\beta-\frac{1}{1-\gamma\beta}\widetilde{\beta}\Phi\left(
\frac{\underline{i}_{t}-x_{2t}}{\varpi}\right)  .
\]

If there is no attenuation effect, the impact effect of the monetary
policy shock $\bar{\varepsilon}_{2t}$ on $Y_{1t}$ is common across regimes and
is given by:%
\[
IR_{NA}=\left(  I_{k-1}-\beta\gamma\right)  ^{-1}\beta=\frac{\beta}%
{1-\gamma\beta},
\]
where the second equality follows from the fact that $\left(  I_{k-1}-\beta
\gamma\right)  \beta=\beta\left(  1-\gamma\beta\right)  $. Therefore,
\[
IR_{t}=IR_{NA}-\frac{1}{1-\gamma\beta}\widetilde{\beta}\Phi\left(  \frac
{\underline{i}_{t}-x_{2t}}{\varpi}\right)  .
\]
The $j$th element of the $k-1$ vector $IR_{t}$ above can be written as%
\begin{align*}
IR_{j,t}  & =IR_{j,NA}-\frac{1}{1-\gamma\beta}\widetilde{\beta}_{j}\Phi\left(
\frac{\underline{i}_{t}-x_{2t}}{\varpi}\right)  \\
& =\frac{\beta_{j}}{1-\gamma\beta}-\frac{\beta_{j}}{1-\gamma\beta}%
\frac{\widetilde{\beta}_{j}}{\beta_{j}}\Phi\left(  \frac{\underline{i}_{t}-x_{2t}}{\varpi
}\right)  \\
& =\left(  1-\frac{\widetilde{\beta}_{j}}{\beta_{j}}\Phi\left(  \frac
{\underline{i}_{t}-x_{2t}}{\varpi}\right)  \right)  IR_{j,NA}.
\end{align*}
Renaming $x_{2t}=C_{2}X_{t}+C_{22}^{\ast}X_{2t}^{\ast},$ the one-step ahead
forecast of the reduced-form shadow rate, as $i_{t|t-1}^{\ast}$ yields
(\ref{eq: attenuation}) as required.

\section{Data description \label{app:data}}

We construct our quarterly data by taking averages of monthly series. For the U.S., the inflation rate is computed from the implicit price deflator (GDPDEF) as $\pi_{t} = 400\times \log(P_{t}/P_{t-1})$, where $P_{t}$ is the GDP deflator. The output gap is calculated as $100\%\times (GDPC1-GDPPOT)/GDPPOT$, where GDPC1 is the series for the U.S. real GDP and GDPPOT is the U.S. real potential GDP. The long-term interest rate is from the 10-year Treasury constant maturity rate (GS10). All these series are from the FRED database.\footnote{The data can be retrieved from the following websites: GDP deflator \citep{FredGDPDEF} \url{https://fred.stlouisfed.org/series/GDPDEF}; and series to construct the output gap \citep{FredGDP, FredGDPPOT}: \url{https://fred.stlouisfed.org/series/GDPC1} and \url{https://fred.stlouisfed.org/series/GDPPOT}; the Federal Funds Rate \citep{FredFEDFUNDS} \url{https://fred.stlouisfed.org/series/FEDFUNDS}; and the long yield \citep{FredLR1, FredLR, FredLR2, FredLR3, FredLR5, FredLR7} \url{https://fred.stlouisfed.org/series/GS10}. The data for the different monetary aggregates \citep{FredSHO, FredM1, FredM2M, FredM2, FredMB, FredMZM, Divisia} is available at: \url{https://fred.stlouisfed.org/categories/24} and \url{http://www.centerforfinancialstability.org/amfm_data.php}.} Money growth data for the U.S. are computed from 12 alternative indicators as listed in Table \ref{tb:mon_agg} as $m_{t}=400\times \log(M_{t}/M_{t-1})$, where $M_{t}$ is the particular money supply considered.
All $M_{t}$ values are quarterly and computed by taking averages of their corresponding monthly values. The traditional monetary aggregates (MB, M1, M2, M2M, MZM), and securities held outright are from the FRED database. The Divisia monetary aggregates (DIVM1, DIVM2, DIVM2M, DIVMZM, DIVM4) are from the Center for Financial Stability Divisia database.

For Japan, the quarterly call rate, bond yields, and the core CPI are computed as the averages of their monthly counterparts. The quarterly inflation rate is computed from the core CPI (consumption tax changes adjusted) as $\pi_{t}=400\times (CPI_{t}-CPI_{t-1})/CPI_{t-1}$. The GDP gap is that published by the Bank of Japan. The trend growth is defined by the annualised growth rate of potential GDP from the previous quarter, which comes from the estimates of the Cabinet Office. The interest on reserves (IOR) is constructed from the interest rate that the Bank of Japan applies to the Complementary Deposit Facility \citep{BOJIOR1,BOJIOR2}.\footnote{The data can be retrieved from the following websites: call rate \citep{BOJCR}: \url{http://www.stat-search.boj.or.jp/index_en.html}; 9-year and 10-year government bond yields \citep{MOFLR}: \url{https://www.mof.go.jp/jgbs/reference/interest_rate/data/jgbcm_all.csv};
%real GDP (Cabinet Office): \url{https://www.esri.cao.go.jp/en/sna/sokuhou/sokuhou_top.html};
GDP gap \cite{BOJOG}: \url{https://www.boj.or.jp/en/research/research_data/index.htm/}; core CPI inflation \citep{SBJINF}: \url{https://www.e-stat.go.jp/stat-search/file-download?statInfId=000031431696&fileKind=1}; trend growth rate \citep{CAOTG}: \url{https://www5.cao.go.jp/keizai3/getsurei-e/index-e.html}.}

\begin{table}[ptb]
\caption{Monetary Aggregates Data used in the Model}%
\label{tb:mon_agg}%
\begin{tabularx}{\linewidth}{|X|X|X|}
\hline
Monetary Aggregate ($M_t$) & Mnemonics in the Corresponding Database & Available Sample Periods \\
\hline
Monetary Base (MB) & BOGMBASE & 1948Q1-2019Q1 \\
\hline
M1 & M1SL & 1959Q2-2019Q1 \\
\hline
M2 & M2SL &  1959Q2-2019Q1 \\
\hline
M2M & M2MSL & 1959Q2-2019Q1 \\
\hline
MZM & MZMSL & 1959Q2-2019Q1 \\
\hline
Securities Held Outright & WSECOUT & 1989Q3-2019Q1 \\
\hline
Divisia M1 (DIVM1) & Divisia M1 & 1967Q2-2019Q1 \\
\hline
Divisia M2 (DIVM2) & Divisia M2 & 1967Q2-2019Q1 \\
\hline
Divisia M2M (DIVM2M) & Divisia M2M & 1967Q2-2019Q1 \\
\hline
Divisia MZM (DIVMZM) & Divisia MZM & 1967Q2-2019Q1 \\
\hline
Divisia M4 (DIVM4) & DM4 & 1967Q2-2019Q1 \\
\hline
\end{tabularx}
\end{table}

\section{Additional empirical results\label{app: empirical}}

\paragraph{Weaker version of IH$_{1}$.}
Table \ref{tb:excl_sr_no_inf_out} shows the results of the weaker version of IH$_{1}$, i.e., $C_{12}=C_{12}^{\ast}=\tilde{\beta}=0$ for inflation and output equations only. As in the baseline CKSVAR specification reported in Table \ref{tb:excl_sr} in the main text, 4 lags are selected for the U.S. and 2 lags are selected for Japan. The $p$-values reported in Table \ref{tb:excl_sr_no_inf_out} show that the weaker version of IH$_{1}$ is firmly rejected for both countries.

\begin{table}[thp]
\begin{centering}
\caption{Test for excluding short rates from VAR that includes long rates}%
\label{tb:excl_sr_no_inf_out}%
\begin{tabular}
[c]{c|rrrrrr|rrrrrr}\hline\hline
\multicolumn{13}{c}{Panel A: KSVAR} \\ \hline
 & \multicolumn{6}{c|}{United States} & \multicolumn{6}{c}{Japan} \\\hline
p & loglik  & pv-p & AIC & LR & df & $p$-val & loglik &  pv-p & AIC & LR & df & $p$-val \\\hline
5 & -213.4  & - & 2.62 & 36.08 & 12 & 0.000 & 248.1 & - & -2.18 & 12.74 & 12 & 0.389\\
4 & -221.5 & 0.446 & 2.55 & 33.42 & 10 & 0.000 & 239.9 & 0.425 & -2.30 & 14.13 & 10 & 0.167\\
3 & -234.4  & 0.112 & 2.53 & 27.12 & 8 & 0.001 & 232.2 & 0.471 & -2.42 & 14.89 & 8 & 0.061\\
2 & -266.0  & 0.000 & 2.66 & 28.29 & 6 & 0.000 & 223.8 & 0.445 & -2.53 & 15.70 & 6 & 0.015\\
1 & -296.7  & 0.000 & 2.78 & 24.62 & 4 & 0.000 & 184.8 & 0.000 & -2.19 & 25.15 & 4 & 0.000\\\hline
\multicolumn{13}{c}{} \\[-2ex]
\multicolumn{13}{c}{Panel B: CKSVAR} \\ \hline
p & loglik  & pv-p & AIC & LR & df & $p$-val & loglik &  pv-p & AIC & LR & df & $p$-val \\\hline
5 & -191.3 & - & 2.60 & 52.08 & 22 & 0.000 & 284.7 & - & -2.42 & 48.79 & 22 & 0.001\\
4 & -202.7 & 0.290 & 2.53 & 44.48 & 18 & 0.000 & 277.1 & 0.766 & -2.61 & 52.42 & 18 & 0.000\\
3 & -223.0 & 0.011 & 2.53 & 33.32 & 14 & 0.003 & 258.1 & 0.081 & -2.62 & 42.61 & 14 & 0.000\\
2 & -256.3 & 0.000 & 2.64 & 31.10 & 10 & 0.001 & 242.1 & 0.018 & -2.68 & 28.23 & 10 & 0.002\\
1 & -290.2 & 0.000 & 2.76 & 27.19 & 6 & 0.000 & 204.8 & 0.000 & -2.43 & 43.99 & 6 & 0.000 \\\hline\hline
\end{tabular}
\\
\end{centering}
\fnote{\footnotesize {Note: Panel A reports results for a KSVAR(p) with inflation, output gap, long rate, and policy rate. Panel B reports corresponding results for a CKSVAR(p) that includes shadow rates. The sample period is 1960q1-2019q1 for the U.S. and 1985q3-2019q1 for Japan. Long rates are 10-year government bond yields for the U.S. and 9-year yields for Japan. Under the null hypothesis, the short rate is excluded from the equations for inflation and output only. loglik is the value of the log-likelihood. pv-p is the $p$-value of the test for lag reduction. AIC is the Akaike information criterion. LR is the value of the LR test statistic for excluding short rates from equations for inflation and output gap. df is the number of restrictions. $p$-val is the asymptotic $\chi^2_{df}$ $p$-value of the test.}}\end{table}

\paragraph{Adding alternative measures of monetary policy.} Table \ref{tb:excl_sr_mon} shows the results of tests for exclusion of the Federal Funds Rate from a SVAR that includes inflation, the output gap, the 10-year bond yield, and various alternative measures of the growth of monetary aggregates outlined in column (1). Column (3) shows the order of the VAR selected by the AIC, which varies between 3 and 4 lags, consistent with the benchmark model in Table \ref{tb:excl_sr} in the main text.  Columns (4)
and (6) report the likelihood ratio test statistics for the joint exclusion
hypothesis and the corresponding asymptotic $p$-values, respectively. These results show
that the data strongly and consistently reject the joint exclusion restrictions on the Federal
Funds Rate across all the alternative specifications for all measures of money
supply, which corroborates the findings in the baseline 4-equation model in Table \ref{tb:excl_sr} in the main text.

\begin{table}[H]
\begin{centering}
\caption{Test for excluding short rates from VARs that include long rates and money}%
\label{tb:excl_sr_mon}%
\begin{tabular}
[c]{llcrcc}\hline\hline
Mon. Aggr. & sample & p & LR & df & $p$-val \\\hline
MB & 1960q1--2019q1 & 3 & 55.05 & 16 & 0.0000\\
M1 & 1960q3-2019q1 & 3 & 55.50 & 16 & 0.0000\\
M2 & 1960q3-2019q1 & 3 & 54.77 & 16 & 0.0000\\
M2M & 1960q3-2019q1 & 4 & 73.78 & 20 & 0.0000\\
MZM & 1960q3-2019q1 & 4 & 79.65 & 20 & 0.0000\\
DIVM1 & 1968q3-2019q1 & 4 & 80.68 & 20 & 0.0000\\
DIVM2 & 1968q3-2019q1 & 4 & 111.50 & 20 & 0.0000\\
DIVM2M & 1968q3-2019q1 & 4 & 110.88 & 20 & 0.0000\\
DIVMZM & 1968q3-2019q1 & 4 & 107.10 & 20 & 0.0000\\
DIVM4 & 1968q3-2019q1 & 4 & 135.38 & 20 & 0.0000\\
SHO &1990q4-2019q1 & 3 & 94.38 & 16 & 0.0000\\\hline\hline
\end{tabular}
\\
\fnote{\footnotesize Note: The estimated model is a KSVAR(p) for the U.S. with inflation, output gap, the Federal Funds Rate, the 10-year government bond yield, and a different measure of money growth in each row. Sample availability varies for each monetary aggregate used. LR is the value of the LR test statistic
for the testing that lags of the Federal Funds Rate can be excluded from all other equations in the model, df is the number of exclusion restrictions, and $p$-val is the asymptotic $\chi^2_{df}$ $p$-value of the test. }
\end{centering}
\end{table}

\paragraph{Robustness of test results for the U.S. to the Great Moderation.}

The test results of the IH over the full sample are subject to a possible misspecification arising from the `Great Moderation', a drop in U.S. macroeconomic volatility in the mid-1980s.
Therefore, we assess the robustness of our results by estimating the model and
performing the above tests of the IH over the sub-sample which starts in 1984q1.
Tables \ref{tb:excl_sr_sub} and \ref{tb:csvar-cksvar_sub} report the results over this subsample, which correspond to the results reported in Tables \ref{tb:excl_sr} and \ref{tb:csvar_lr_long} in the main text for the full sample, respectively. The results of the tests of the IH remain the same: the hypothesis is firmly rejected.

\begin{table}[H]
%[ptb]
\begin{centering}
\caption{Test for excluding short rates form VAR that includes long rates post-1984}%
\label{tb:excl_sr_sub}%
\begin{tabular}
[c]{c|rrrrrr|rrrrrr}\hline\hline
 & \multicolumn{6}{c|}{KSVAR(p)} & \multicolumn{6}{c}{CKSVAR(p)} \\\hline
p & loglik  & pv-p & AIC & LR & df & $p$-val & loglik &  pv-p & AIC & LR & df & $p$-val \\\hline
5 & 97.92 & - & -0.01 & 25.63 & 18 & 0.11 & 122.37 & - & -0.07 & 61.65 & 33 & 0.002\\
4 & 92.83 & 0.857 & -0.17 & 28.01 & 15 & 0.022 & 119.99 & 1.000 & -0.33 & 70.78 & 27 & 0.000 \\
3 & 85.07 & 0.776 & -0.28 & 22.31 & 12 & 0.034 & 103.32 & 0.556 & -0.37 & 48.44 & 21 & 0.001\\
2 & 66.06 & 0.064 & -0.24 & 22.99 & 9 & 0.006 & 77.86 & 0.009 & -0.30 & 40.06 & 15 & 0.000\\
1 & 14.30 & 0.000 & 0.27 & 5.33 & 6 & 0.502 & 19.29 & 0.000 & 0.25 & 13.59 & 9 & 0.138\\\hline\hline
\end{tabular}
\\
\end{centering}
\fnote{\footnotesize Note: The estimated model is a (C)KSVAR(p) for the U.S. with inflation, output gap, Federal Funds Rate, and the 10-year government bond yield. Estimation sample is 1984q1-2019q1. loglik is the value of the log-likelihood. pv-p is the $p$-value
of the test for lag reduction. AIC is the Akaike information criterion. LR is the test statistic for excluding short rates from equations for inflation, output gap and long rates. df is the number of restrictions.  $p$-val is the asymptotic $\chi^2_{df}$ $p$-value of the test.}\end{table}

\begin{table}[H]
%[ptb]
\begin{centering}
\caption{Testing CSVAR against CKSVAR post-1984}%
\label{tb:csvar-cksvar_sub}%
\begin{tabular}
[c]{lrrrr}\hline\hline
Country & p & LR & df & $p$-val\\\hline
U.S. & 3 & 31.17 & 15 & 0.008\\\hline\hline
\end{tabular}
\\
\end{centering}
\fnote{\footnotesize Note: The unrestricted model is a CKSVAR(3) for the U.S. with inflation, output gap, 10-year government bond yields, and the Federal Funds Rate. Sample: 1984q1-2019q1. LR is the test statistics of the restrictions that the model reduces to CSVAR(3). Lag order is chosen by AIC. df is the number of restrictions. $p$-val is the asymptotic $\chi^2_{df}$ $p$-value of the test.}\end{table}

\paragraph{Robustness of results to the inclusion of credit spreads in the VAR.}

For the U.S., we use Moody's seasoned BAA corporate bond yield relative to 10-year treasury yield \citep{FredCS} as the credit spreads, and the excess bond premium in \cite{GZ_AER12} and \cite{GZ_AER_DATA} for Japan. The test results show that our baseline results for IH$_{1}$ (Tables \ref{tb:excl_sr_app}-\ref{tb:excl_sr_2app}) and IH$_{2}$ (Table \ref{tb:csvar-cksvar_app}) are robust to the inclusion of credit spreads in the VAR.

\begin{table}[H]
\begin{centering}
\caption{Test for excluding short rates from VAR that includes long rates and credit spreads}%
\label{tb:excl_sr_app}%
\begin{tabular}
[c]{c|rrrrrr|rrrrrr}\hline\hline
\multicolumn{13}{c}{Panel A: KSVAR} \\ \hline
 & \multicolumn{6}{c|}{United States} & \multicolumn{6}{c}{Japan} \\\hline
p & loglik  & pv-p & AIC & LR & df & $p$-val & loglik &  pv-p & AIC & LR & df & $p$-val \\\hline
5 & 298.6 & - & -2.34 & 40.59 & 24 & 0.018 & 391.8 & - & -3.52 & 45.64 & 24 & 0.005\\
4 & 286.3 & 0.486 & -2.54 & 42.69 & 20 & 0.002 & 373.9 & 0.073 & -3.63 & 38.10 & 20 & 0.009\\
3 & 266.5  & 0.086 & -2.62 & 37.49 & 16 & 0.002 & 358.7 & 0.061 & -3.77 & 34.18 & 16 & 0.005\\
2 & 232.1  & 0.000 & -2.47 & 18.24 & 12 & 0.109 & 346.3 & 0.100 & -3.96 & 31.84 & 12 & 0.001\\
1 & 178.4  & 0.000 & -2.02 & 10.79 & 8 & 0.214 & 303.4 & 0.000 & -3.69 & 44.30 & 8 & 0.000\\\hline
\multicolumn{13}{c}{} \\[-2ex]
\multicolumn{13}{c}{Panel B: CKSVAR} \\ \hline
p & loglik  & pv-p & AIC & LR & df & $p$-val & loglik &  pv-p & AIC & LR & df & $p$-val \\\hline
5 & 328.6 & - & -2.42 & 87.26 & 44 & 0.000 & 441.6 & - & -3.89 & 131.19 & 44 & 0.000\\
4 & 309.9 & 0.167 & -2.59 & 83.99 & 36 & 0.000 & 412.4 & 0.001 & -3.90 & 100.99 & 36 & 0.000\\
3 & 288.0 & 0.035 & -2.72 & 73.60 & 28 & 0.000 & 376.7 & 0.000 & -3.82 & 62.78 & 28 & 0.000\\
2 & 243.0 & 0.000 & -2.48 & 37.62 & 20 & 0.010 & 356.6 & 0.000 & -3.97 & 47.15 & 20 & 0.001\\
1 & 183.5 & 0.000 & -2.10 & 20.41 & 12 & 0.060 & 316.2 & 0.000 & -3.81 & 62.95 & 12 & 0.000 \\\hline\hline
\end{tabular}
\\
\end{centering}
\fnote{\footnotesize {Note: Panel A reports results for a KSVAR(p) with inflation, output gap, long rate, credit spread, and policy rate. Panel B reports corresponding results for a CKSVAR(p) that includes shadow rates. Estimation sample is 1987q2-2019q1 for the U.S. and 1985q3-2019q1 for Japan. Long rates are 10-year government bond yields for the U.S. and 9-year yields for Japan. The credit spreads are Moody's seasoned BAA corporate bond yield relative to 10-year treasury yield for the U.S., and the excess bond premium introduced by \cite{GZ_AER12} for Japan. loglik is the value of the log-likelihood. pv-p is the $p$-value of the test for lag reduction. AIC is the Akaike information criterion. LR is the test statistic for excluding short rates from equations for inflation, output gap, credit spread, and long rates. df is the number of restrictions. $p$-val is the asymptotic $\chi^2_{df}$ $p$-value of the test.}}\end{table}

\begin{table}[H]
\begin{centering}
\caption{Test for excluding short rates from VAR that includes long rates and credit spreads}%
\label{tb:excl_sr_2app}%
\begin{tabular}
[c]{c|rrrrrr|rrrrrr}\hline\hline
\multicolumn{13}{c}{United States, with Excess Bond Premium} \\ \hline
 & \multicolumn{6}{c}{Panel A: KSVAR} & \multicolumn{6}{c}{Panel B: CKSVAR} \\\hline
p & loglik  & pv-p & AIC & LR & df & $p$-val & loglik  & pv-p & AIC & LR & df & $p$-val\\\hline
5 & 60.9 & - & 0.979 & 73.83 & 24 & 0.000 & 86.7 & - & 0.970 & 113.92 & 44 & 0.000\\
4 & 47.4 & 0.359 & 0.851 & 70.02 & 20 & 0.000 & 70.4 & 0.332 & 0.818 & 100.20 & 36 & 0.000\\
3 & 22.6 & 0.009 & 0.849 & 54.65 & 16 & 0.000 & 36.7 & 0.001 & 0.859 & 69.35 & 28 & 0.000\\
2 & -13.8 & 0.000 & 0.975 & 33.35 & 12 & 0.001 & 0.4 & 0.000 & 0.929 & 49.91 & 20 & 0.000\\
1 & -49.3 & 0.000 & 1.092 & 20.91 & 8 & 0.007 & -36.6 & 0.000 & 1.006 & 37.12 & 12 & 0.000\\\hline\hline
\end{tabular}
\\
\end{centering}
\fnote{\footnotesize {Note: Panel A reports results for a KSVAR(p) with inflation, output gap, long rate, credit spread, and policy rate. Panel B reports corresponding results for a CKSVAR(p) that includes shadow rates. Estimation sample is 1974q2-2019q1. Credit spreads are the excess bond premium in \cite{GZ_AER12}. loglik is the value of the log-likelihood. pv-p is the $p$-value of the test for lag reduction. AIC is the Akaike information criterion. LR is the test statistic for excluding short rates from equations for inflation, output gap, credit spread, and long rates. df is the number of restrictions. $p$-val is the asymptotic $\chi^2_{df}$ $p$-value of the test.}}\end{table}

\begin{table}[H]
\begin{centering}
\caption{Testing CSVAR against CKSVAR with credit spreads}%
\label{tb:csvar-cksvar_app}%
\begin{tabular}
[c]{lrrrr}\hline\hline
Country & p & LR & df & $p$-val\\\hline
U.S.(BAA) & 3 & 49.58 & 19 & 0.000\\\hline
U.S.(EBP) & 4 & 53.56 & 24 & 0.000\\\hline
Japan & 2 & 40.62 & 14 & 0.000\\\hline\hline
\end{tabular}
\\
\end{centering}
\fnote{\footnotesize{Note: The unrestricted model is a CKSVAR($p$) in inflation, output gap, long rate, credit spread, and policy rate. Long rate: 10-year government bond yield (U.S.), 9-year government bond yield (Japan). Policy rate: Federal Funds Rate (U.S.), call rate (Japan). Credit spread: Moody's seasoned BAA corporate bond yield relative to 10-year treasury yield (U.S.), the excess bond premium (U.S. and Japan). Sample: 1987q2-2019q1 (U.S. with BAA spread), 1974q2-2019q1 (U.S. with EBP), 1985q3-2019q1 (Japan). $p$ chosen by AIC. LR is the test statistics of the restrictions that the model reduces to CSVAR($p$). df is the number of restrictions. $p$-val is the asymptotic $\chi^2_{df}$ $p$-value of the test.}}\end{table}

\paragraph{Robustness of Japanese results to 10-year rates.}
Similarly, we test the robustness of our results for the Japanese data by
using the 10-year yields instead. This shortens the available
sample for estimation to 1987q4 to 2019q1. Tables \ref{tb:excl_sr_10yearJP} and \ref{tb:csvar-cksvar_10yearJP} report test statistics for the two types of tests
for the IH. From Tables \ref{tb:excl_sr_10yearJP} and
\ref{tb:csvar-cksvar_10yearJP}, the IH is rejected across all lags. For the
CKSVAR alternative, 2 lags are selected based on the AIC. Table
\ref{tb:csvar-cksvar_10yearJP} also suggests the rejection of the IH.

\begin{table}[H]
\begin{centering}
\caption{Test for excluding short rates from VAR for Japan using 10-year bond yields}%
\label{tb:excl_sr_10yearJP}%
\begin{tabular}
[c]{c|rrrrrr|rrrrrr}\hline\hline
 & \multicolumn{6}{c|}{KSVAR(p)} & \multicolumn{6}{c}{CKSVAR(p)} \\\hline
p & loglik  & pv-p & AIC & LR & df & $p$-val & loglik &  pv-p & AIC & LR & df & $p$-val \\\hline
5 & 285.1 & - & -2.92 & 37.43 & 18 & 0.005 & 320.5 & - & -3.17 & 99.85 & 33 & 0.000\\
4 & 275.1 & 0.217 & -3.02 & 31.62 & 15 & 0.007 & 307.4 & 0.159 & -3.28 & 86.95 & 27 & 0.000\\
3 & 270.5 & 0.605 & -3.20 & 32.05 & 12 & 0.001 & 290.6 & 0.023 & -3.33 & 62.07 & 21 & 0.000\\
2 & 256.2 & 0.155 & -3.23 & 24.60 & 9 & 0.003 & 274.3 & 0.004 & -3.39 & 50.90 & 15 & 0.000\\
1 & 196.4 & 0.000 & -2.53 & 22.84 & 6 & 0.001 & 212.8 & 0.000 & -2.73 & 38.81 & 9 & 0.000\\\hline\hline
\end{tabular}
\\
\end{centering}
\fnote{\footnotesize Note: The estimated model is a (C)KSVAR(p) for Japan with inflation, output gap, 10-year government bond yields, and the call rate. Estimation sample is 1987q4-2019q1. loglik is the value of the log-likelihood. pv-p is the $p$-value
of the test for lag reduction. AIC is the Akaike information criterion. LR is the test statistic for excluding short rates from equations for inflation, output gap and long rates. df is the number of restrictions. $p$-val is the asymptotic $\chi^2_{df}$ $p$-value of the test.}\end{table}

\begin{table}[H]
\begin{centering}
\caption{Testing CSVAR against CKSVAR for Japan using 10-year bond yields}%
\label{tb:csvar-cksvar_10yearJP}%
\begin{tabular}
[c]{lrrrr}\hline\hline
Country & p & LR & df & $p$-val\\\hline
Japan & 2 & 47.54 & 11 & 0.000\\\hline\hline
\end{tabular}
\\
\end{centering}
\fnote{\footnotesize Note: The unrestricted model is a CKSVAR(2) for Japan with inflation, output gap, 10-year government bond yields, and the call rate. Estimation sample: 1987q4-2019q1. LR is the test statistics of the restrictions that the model reduces to CSVAR(2). Lag order is chosen by AIC. df is yje number of restrictions. $p$-val is the asymptotic $\chi^2_{df}$ $p$-value of the test.}\end{table}

\paragraph{Power of the irrelevance tests IH$_1$ and IH$_2$.}

We use the theoretical model to generate 100 artificial time series under values for the parameter $\xi$ in the range $[0.7, 0.99]$. Table \ref{tb:excl_sr_tab1} and \ref{tb:excl_sr_tab2} report the number of rejections for the tests of our irrelevance hypotheses IH$_1$ and IH$_2$, respectively. If our tests are powerful, we would expect the number of rejections to decline with $\xi$ approaching the value of 1 for which the irrelevance hypothesis holds true in the simulated data.

The tables show that the irrelevance tests are powerful. For instance, in the case of the KSVAR as the unrestricted model, the test rejects IH$_1$ at a $1$ percent significance level with the rejection rate (frequency) of 99 percent when $\xi=0.7$, while the rejection rate is 1 percent when $\xi=0.99$ at the same significance level. Similar results hold for alternative significance levels (columns 2, 3), the CKSVAR as the unrestricted model (Table \ref{tb:excl_sr_tab1}, Panel B), and the test for IH$_2$ (Table \ref{tb:excl_sr_tab2}).

\begin{table}[H]
\begin{centering}
\caption{Test for excluding short rates from VAR that includes long rates, with simulated data}%
\label{tb:excl_sr_tab1}%
\begin{tabular}
[c]{c|rrr}\hline\hline
\multicolumn{4}{c}{Panel A: KSVAR} \\ \hline
 &  (1) &  (2)  & (3) \\
$\xi$ & $p \leq 0.01$ & $p \leq 0.05$  & $p \leq 0.1$ \\\hline
0.7 & 99 & 100 & 100 \\
0.75 & 90 & 98 & 99 \\
0.8 & 76 & 92 & 94 \\
0.85 & 58 & 73 & 81 \\
0.9 & 25 & 48 & 60 \\
0.95 & 7 & 19 & 31 \\
0.99 & 1 & 11 & 18 \\\hline
\multicolumn{4}{c}{} \\[-2ex]
\multicolumn{4}{c}{Panel B: CKSVAR} \\ \hline
 &  (1) &  (2)  & (3) \\
$\xi$ & $p \leq 0.01$ & $p \leq 0.05$  & $p \leq 0.1$ \\\hline
0.7 & 99 & 100 & 100 \\
0.75 & 93 & 99 & 100 \\
0.8 & 85 & 92 & 94 \\
0.85 & 69 & 84 & 88 \\
0.9 & 42 & 62 & 71 \\
0.95 & 20 & 33 & 45 \\
0.99 & 13 & 25 & 32 \\\hline\hline
\end{tabular}
\\
\end{centering}
\fnote{\footnotesize {Note: Panel A reports results for a KSVAR(1) with inflation, output gap, long rate, and policy rate. Panel B reports corresponding results for a CKSVAR(1) that includes shadow rates. Estimation sample is data simulated by the calibrated DSGE model for 237 quarters, which equals the length of the U.S. sample in section \ref{s: IHtest}. For each value of $\xi$, we run 100 simulations. Columns 2-4 report how many times the irrelevant hypothesis is rejected with 1 percent, 5 percent, and 10 percent significance levels, respectively.}}\end{table}

\begin{table}[H]
\begin{centering}
\caption{Test CSVAR against CKSVAR, with simulated data}%
\label{tb:excl_sr_tab2}%
\begin{tabular}
[c]{c|rrr}\hline\hline
 &  (1) &  (2)  & (3) \\
$\xi$ & $p \leq 0.01$ & $p \leq 0.05$  & $p \leq 0.1$ \\\hline
0.7 & 100 & 100 & 100 \\
0.75 & 100 & 100 & 100 \\
0.8 & 99 & 100 & 100 \\
0.85 & 91 & 97 & 97 \\
0.9 & 58 & 81 & 85 \\
0.95 & 13 & 20 & 34 \\
0.99 & 3 & 6 & 14 \\\hline\hline
\end{tabular}
\\
\end{centering}
\fnote{\footnotesize {Note: This table reports results for a CKSVAR(1) with inflation, output gap, long rate, and policy rate. Estimation sample is data simulated by the calibrated DSGE model for 237 quarters, which equals the length of the U.S. sample in section \ref{s: IHtest}. For each value of $\xi$, we run 100 simulations. Columns 2-4 report how many times the irrelevant hypothesis is rejected with 1 percent, 5 percent, and 10 percent significance levels, respectively.}}\end{table}

\paragraph{Testing no attenuation effect.}

We repeat our test of no attenuation in the response of long rates to monetary policy shocks for different sample periods for the U.S. Table \ref{tb:no_atten_vs} shows that
no attenuation hypothesis is not rejected if the same sample period of 1990q1--2012q4 is adopted as in \cite{SwansonWilliams2014}. If the sample period is extended backwards (starting from 1960q1), the null is rejected at a 5 percent significance level. These results suggest that the responses of the long rate to a monetary policy shock may differ between non-ELB and ELB regimes, depending on the sample period.

\begin{table}[H]
\begin{centering}
\caption{Test for no attenuation, various sample periods}
\label{tb:no_atten_vs}%
\begin{tabular}
[c]{lcrc}\hline\hline
sample & p & LR & $p$-val \\\hline
1990q1--2012q4 & 3 & 0.03 & 0.872\\
1960q1--2012q4 & 3 & 4.08 & 0.043\\
\hline\hline
\end{tabular}
\\
\fnote{\footnotesize Note: The estimated model is a CKSVAR(p) for the U.S. with inflation, output gap, long-term rate, and policy rate. The long rate is the 10-year government bond yields. The hypothesis is tested with different sample periods, with the first one being consistent with Swanson and Williams. LR is the value of the likelihood ratio test statistic and asymptotic $p$-values are reported. }
\end{centering}
\end{table}

\section{Choleski identification\label{app_Chole}}

In our benchmark analysis we use the combination of the ELB identification developed by \cite{Mavroeidis2019} and the sign restrictions similar to those employed by \cite{DebortoliGaliGambetti2019} to estimate the UMP parameter $\xi$. This appendix shows the results from using the standard Choleski identification. Figures \ref{fig: chol_irf_us} and \ref{fig: chol_irf_jp} reports results for the U.S. and Japan, respectively. They show that the Choleski identification generates several puzzling responses such as the instantaneous decreases in output and inflation in reaction to a negative monetary policy shock. These responses are consistent with the findings in \cite{GertlerKaradi2015} for the U.S. and \cite{kubota2022macro} for Japan, who also show similar responses when using the Choleski identification. Thus, our analysis corroborates the results on the empirically-incongruous responses from the Choleski identification, while showing that the identification based on the combination of the ELB identification and sign restrictions provides plausible responses to monetary policy shocks for the U.S. and Japan when the economy is at the ELB. See \cite{gortz2021vintage} for a discussion of the issue and some additional corroborative evidence on U.K data.

\begin{figure}[H]
\caption{Choleski identification: Impulse responses to a monetary policy shock in the U.S.}%
\label{fig: chol_irf_us}
%[ptb]
\centering
\includegraphics[
height=3.8245in,
width=6.2289in
]{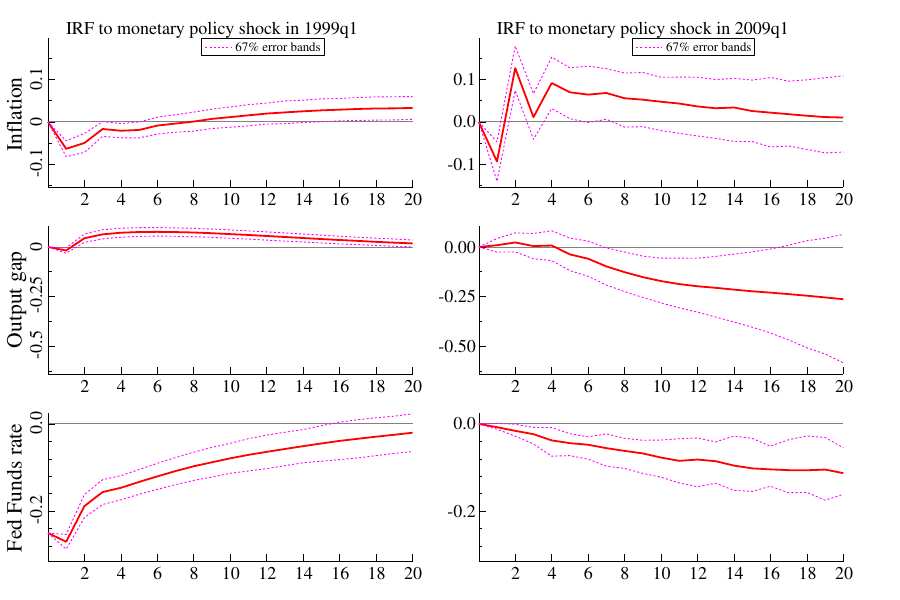}
\fnote{\footnotesize{Note: Identified sets of IRFs in 1999q1 and 2009q1 to a -25bps monetary policy shock estimated from CKSVAR(3) model in inflation, output gap, and the Federal Funds Rate for the U.S. over the period 1960q1-2019q1, identified by the Choleski restrictions that the monetary policy shock has no contemporaneous effects on inflation and output. Dotted lines show the 67 percent asymptotic error bands.}}
\end{figure}

\begin{figure}[H]
\caption{Choleski identification: Impulse responses to a monetary policy shock in Japan}%
\label{fig: chol_irf_jp}
%[ptb]
\centering
\includegraphics[
height=3.8245in,
width=6.2289in
]{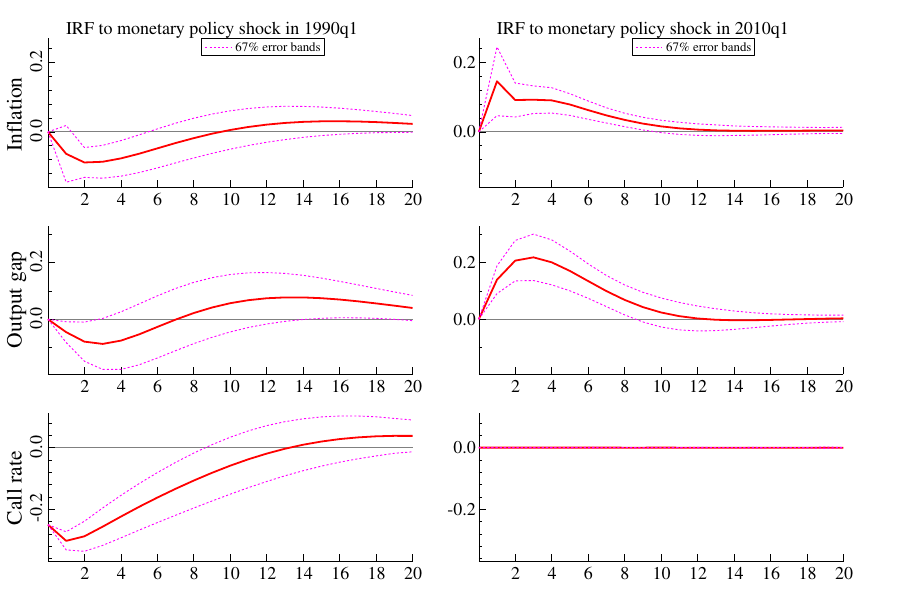}
\fnote{\footnotesize{Note: Identified sets of IRFs in 1990q1 and 2010q1 to a -25bps point monetary policy shock estimated using a CKSVAR(2) model in inflation, output gap, and the call rate for Japan over the period 1985q3-2019q1, identified by the Choleski restrictions that the monetary policy shock has no contemporaneous effects on inflation and output. Dotted lines show the 67 percent asymptotic error bands.}}
\end{figure}

\section{Shadow rates\label{app:shadow_rates}}
\begin{figure}[t]
\centering
\caption{Shadow policy rate for the U.S.}\label{fig: shadow US}
\includegraphics[width=11cm
]{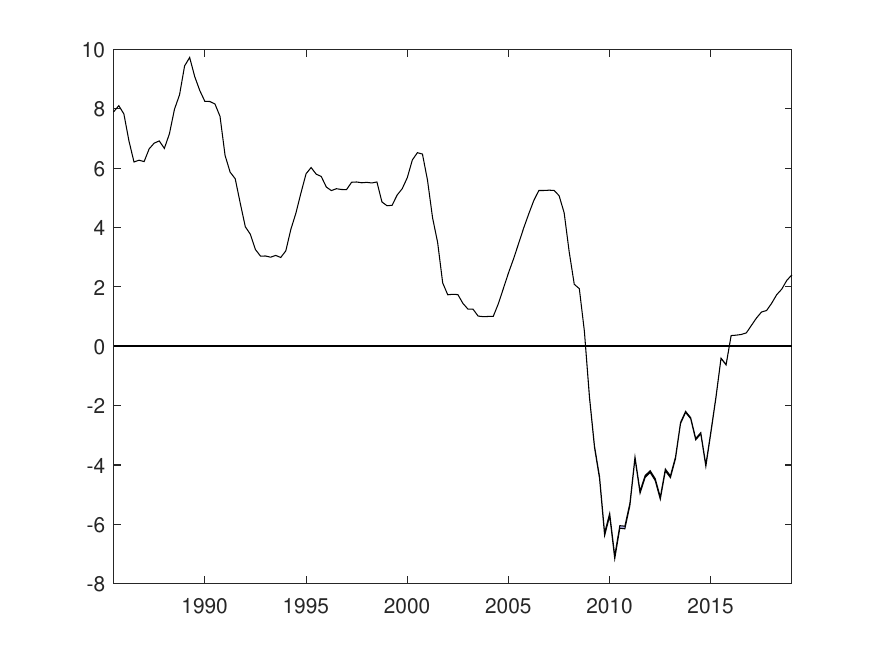}
\fnote{\footnotesize{Note: Estimated using a CKSVAR(3) model in inflation, output gap, and the Federal Funds Rate for the U.S. over the period 1960q1-2019q1 (plotted over the sub-sample 1985q3-2019q1), identified by the sign restrictions that a -25bp monetary policy shock has nonnegative effects on inflation and output and nonpositive effects on the short rate up to four quarters.}}
\end{figure}
Our analysis defines the shadow rate as the short-term interest rate that the central bank would set if there were no ELB. Thus defined, the shadow rate can be interpreted as an indicator of the desired monetary policy stance and we provide estimates of it for Japan and the U.S. Our estimates of the shadow rate do not impose the assumption that the model used to obtain them is constant across regimes, and therefore they explicitly account for the empirical relevance of the ELB over the estimation periods.

The important caveat is that the shadow rates are not identified under our present assumptions. As explained in \cite{Mavroeidis2019}, identifying the shadow rate $i_{t}^*$ in the empirical model (\ref{eq: Y2})-(\ref{eq: Y1}) in the main text requires knowledge of the parameter $\alpha$, which scales the reaction function coefficients and policy shocks during the ELB regimes and is not identified without additional information. This parameter is needed \emph{in addition} to the parameter $\xi$ that measures the overall impact effect of UMP.  In other words, to properly identify the shadow rate and interpret it as a measure of desired policy stance, we need to be able to isolate the effect of FG encapsulated by $\alpha$. This exercise is beyond the scope of the present paper.

\begin{figure}[H]
\centering
\caption{Shadow policy rate for Japan}\label{fig: shadow JP}
\includegraphics[width=11cm
]{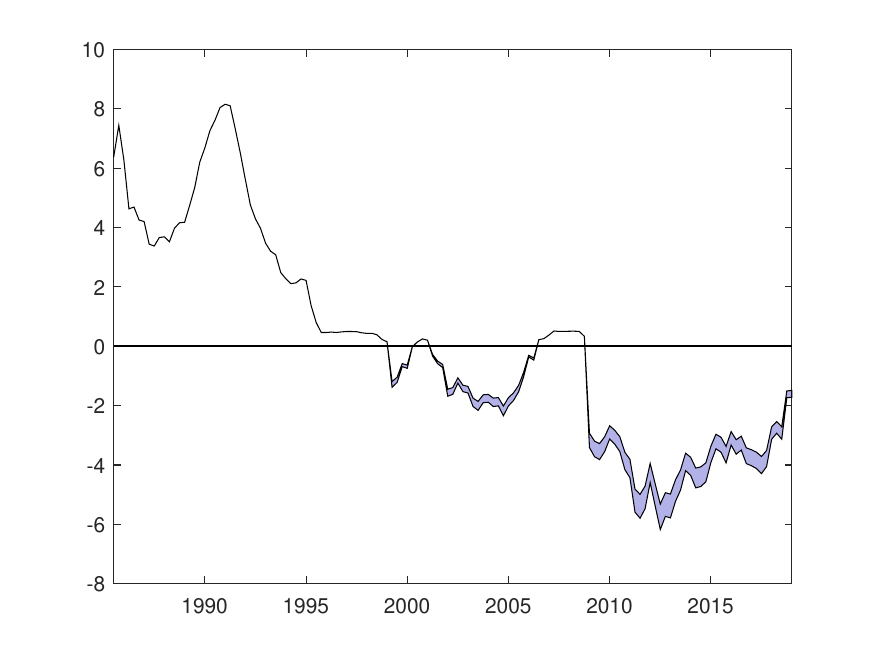}
\fnote{\footnotesize{Note: Estimated using a CKSVAR(2) model in inflation, output gap, and the call rate for Japan over the period 1985q3-2019q1, identified by the sign restrictions that a -25bp monetary policy shock has nonnegative effects on inflation and output and nonpositive effects on the short rate up to four quarters.}}
\end{figure}

With the above caveat in mind, we report identified shadow rates under the assumption of $\alpha=0$. The shadow rates are given in Figures \ref{fig: shadow US} and \ref{fig: shadow JP}  for the U.S. and Japan, respectively. Different values of $\alpha$ would scale those estimates by a factor $1+\alpha$.\footnote{Results are available on request.} Note that, even with $\alpha=0$, the shadow rate is only partially identified because it also depends on the parameter $\xi$ that is partially identified. This uncertainty due to $\xi$ is reflected in the shaded areas below the ELB in the figures.\footnote{The shadow rate is equal to the observed policy rate above the ELB, see equation (\ref{eq: Y2}) in the main text. Below the ELB, it is given by the equation $Y_{2t}^*\allowbreak = \kappa \overline{Y}_{2t} + (1-\kappa) b_t $, where $\kappa = (1+\alpha)(1-\gamma\beta)/(1-\xi\gamma\beta)$ and $\overline{Y}_{2t}$ is a ``reduced-form'' shadow rate that can be filtered from the data using the likelihood, see \cite{Mavroeidis2019}.} In the case of the U.S., the shadow rate dropped sharply soon after the onset of the global financial crisis of 2007-2008. It reached its smallest value at the beginning of 2010 and gradually recovered until the exit from the ELB in 2016.
In Japan, the behaviour of the shadow rate is different during the three ELB episodes. During the first episode, the shadow rate fell modestly. In the second episode, it exhibited a persistent decline until the beginning of 2005, followed by a quick reversal. In the third episode, which coincided with the ELB in the U.S., the decline was sharp, and followed by a second wave of declines that lasted until mid-2012. From that point on, the shadow rate exhibited a steady rise, but stayed far from zero even at the end of the sample, and remained near its trough in the second episode.